\documentclass[11pt]{article}
\usepackage[a4paper,margin=25mm]{geometry}
\usepackage[utf8]{inputenc}
\usepackage[T1]{fontenc}
\usepackage[english]{babel}
\usepackage{
    amsmath,
    amssymb,
    mathtools,
    amsthm,
    bbm,
    amsbsy,
    stackrel,
    mathrsfs
}
\usepackage{dsfont}
\usepackage{subcaption}
\usepackage{enumerate}
\usepackage{hyperref}
\usepackage{color,xcolor}
\usepackage{csquotes}
\usepackage[labelfont=bf]{caption}
\usepackage{algorithm}
\usepackage{algpseudocode}
\usepackage{multirow,booktabs}

\newtheorem{theorem}{Theorem}
\newtheorem{proposition}[theorem]{Proposition}

\newtheorem{lemma}[theorem]{Lemma}

\theoremstyle{definition}

\newtheorem{assumption}[theorem]{Assumption}

\theoremstyle{remark}
\newtheorem{remark}[theorem]{Remark}

\numberwithin{equation}{section}
\numberwithin{theorem}{section}
\DeclareMathOperator{\Pas}{\mathds{P}\text{-as}}
\DeclareMathOperator*{\argmin}{argmin}
\DeclareMathOperator*{\Argmin}{Argmin}
\DeclareMathOperator*{\supp}{supp}

\DeclareMathOperator{\VaR}{VaR}
\DeclareMathOperator{\ES}{ES}
\DeclareMathOperator{\dist}{dist}
\DeclareMathOperator{\oo}{o}
\DeclareMathOperator{\OO}{O}
\DeclareMathOperator{\Cost}{Cost}
\DeclareMathOperator{\minimize}{minimise}
\DeclareMathOperator{\subjectTo}{subject~to}

\DeclareMathOperator{\Argmax}{Argmax}

\hypersetup{
    colorlinks=true,
    urlcolor=black,
    linkcolor=black,
    citecolor=black
}

\title{A Multilevel Stochastic Approximation Algorithm for Value-at-Risk and Expected Shortfall Estimation}
\author{
St\'ephane~Cr\'epey\thanks{
Universit\'e Paris Cit\'e,
Laboratoire de Probabilit\'es, Statistique et Mod\'elisation (LPSM), CNRS UMR 8001,
75013 Paris, France.
\texttt{stephane.crepey@lpsm.paris}
}\and Noufel~Frikha\thanks{
Universit\'e Paris 1 Panth\'eon-Sorbonne,
Centre d'Economie de la Sorbonne (CES),
106 Boulevard de l’H\^opital,
75642 Paris Cedex 13, France.
\texttt{noufel.frikha@univ-paris1.fr}
}\and Azar~Louzi\thanks{
Universit\'e Paris Cit\'e,
Laboratoire de Probabilit\'es, Statistique et Mod\'elisation (LPSM), CNRS UMR 8001,
75013 Paris, France.
\texttt{azar.louzi@lpsm.paris}
}
}
\date{August 27, 2025}

\begin{document}

\maketitle

\begin{center}
\begin{minipage}{.8\linewidth}
\small
{\bf Abstract.}
We propose a multilevel stochastic approximation (MLSA) scheme for the computation of the value-at-risk (VaR) and expected shortfall (ES) of a financial loss, which can only be computed via simulations conditionally on the realisation of future risk factors. Thus the problem of estimating its VaR and ES is nested in nature and can be viewed as an instance of stochastic approximation problems with biased innovations. In this framework, for a prescribed accuracy $\varepsilon $, the optimal complexity of a nested stochastic approximation algorithm is shown to be of the order $\varepsilon ^{-3}$. To estimate the VaR, our MLSA algorithm attains an optimal complexity of the order $\varepsilon ^{-2-\delta }$, where $\delta \in (0,1)$ is some parameter depending on the integrability degree of the loss, while to estimate the~ES, the algorithm achieves an optimal complexity of the order $\varepsilon ^{-2}|\ln {\varepsilon }|^{2}$. Numerical studies of the joint evolution of the error rate and the execution time demonstrate how our MLSA algorithm regains a significant amount of the performance lost due to the nested nature of the problem.
\medskip




\noindent
{\bf DOI.}
10.1007/s00780-025-00573-5
\quad
\noindent
{\bf arXiv.}
2304.01207
\quad
\noindent
{\bf HAL.}
04037328
\end{minipage}
\end{center}

\section{Introduction}

The post-great recession era has witnessed the implementation of multiple
risk measures with the goal of better controlling financial losses. In
a spirit of precaution and consistency, the Fundamental Review of the Trading
Book~\cite{BISFRTB} triggered a shift away from the value-at-risk (VaR)
towards the expected shortfall (ES, i.e., the average loss given it exceeds
the VaR) as a reference regulatory risk measure. The future valuation of
a position in financial derivatives is defined by a conditional expectation.
In the case of exotic products or nonlinear portfolio computations (e.g.~credit
or funding valuation adjustments), this future valuation can only be computed
by Monte Carlo simulation. VaR and ES computations on future portfolio
losses are thus nested in nature. A brute force nested Monte Carlo computational
approach \`a la Gordy and Juneja~\cite{Gordy} may be too heavy for practical
use. As a shortcut solution, a regression-based estimator for the inner
conditional expectation is implemented by Broadie et al.~\cite{BroadieDuMoallemi15},
but the resulting regression error is difficult to control. Barrera et
al.~\cite{barrera:hal-01710394} adopt the stochastic approximation (SA)
point of view of Bardou et al.~\cite{BardouFrikhaPages+2009+173+210} (see
also Bardou et al.~\cite{bardou2016cvar,10.1007/978-3-642-04107-5_11} and
Frikha~\cite{doi:10.1137/120903142}), itself building on top of the Rockafellar
and Uryasev~\cite{10.21314/JOR.2000.038} representation of the VaR and
ES as joint solutions to a convex optimisation problem. Barrera et al.~\cite{barrera:hal-01710394}
revisit the nested vs.~simulation/regression computational strategies of
both~\cite{Gordy} and~\cite{BroadieDuMoallemi15} under assumptions that are more realistic and
easier to check; but we clarify in this paper that in order to reach an
estimation accuracy of the order~$\varepsilon $, the complexity of the
resulting nested SA algorithm in~\cite{barrera:hal-01710394} is of the
order $\varepsilon ^{-3}$. As for the regression strategy in~\cite{barrera:hal-01710394},
its non-asymptotic error bounds involve some constants that are rather
large, which renders them of theoretical interest only.\\

In the present paper, we propose a multilevel SA (MLSA) algorithm for the
computation of the VaR and ES of a loss that has the form of a conditional
expectation. MLSA was introduced by Frikha~\cite{10.1214/15-AAP1109} as
an extension to the SA framework of multilevel Monte Carlo path simulation
of Giles~\cite{10.1287/opre.1070.0496}. It was then revisited by Dereich
and M\"uller-Gronbach~\cite{Dereich2019} from the different perspective
of approximating the objective function of an SA problem rather than its
solutions. We also refer to Frikha and Huang~\cite{FRIKHA20154066} for
a multi-step Richardson--Romberg stochastic approximation method. We stress
that the uniform mean-reversion assumption on the objective function discussed
by~\cite{10.1214/15-AAP1109} is generally not satisfied in a VaR--ES setup;
so the theoretical performance of the MLSA scheme applied to both risk
measures is not directly guaranteed by the results therein. We also refer
to Giles and Haji-Ali~\cite{doi:10.1137/18M1173186} who propose a multilevel
Monte Carlo estimator for the computation of the nested expectation
$\mathbb{P}[\mathbb{E}[X|Y]\geq 0]$ as well as a stochastic root-finding
algorithm for the computation of the VaR and~ES, but do not provide a theoretical
analysis of the latter. Our main contribution is to propose an optimised
MLSA algorithm that achieves sharp theoretical complexities when focusing
on estimating either the VaR or the ES. The VaR-focused estimation achieves
a complexity of the order $\varepsilon ^{-2-\delta}$, where
$\delta \in (0,1)$ is some specific parameter depending on the integrability
degree of the loss. The ES-focused estimation attains a complexity of the
order $\varepsilon ^{-2}|\ln{\varepsilon}|^{2}$. The theoretical analysis
is verified numerically on two financial case studies that show a considerable
gain with respect to the nested SA approach in terms of the computational
time needed to estimate the VaR and ES.\\

The paper is organised as follows. Section~\ref{sec:problem} presents the
problem of computing the VaR and ES by means of an SA scheme. Section~\ref{sec:nested-sa}
analyses the nested SA scheme for the computation of the VaR and ES. Section~\ref{sec:mlsa}
introduces and analyses the MLSA algorithm. The theoretical analysis and
the gain in performance achieved by our optimised MLSA scheme are illustrated
numerically in Sect.~\ref{s:numerical}. The proofs of the technical results
are given in Appendices~\ref{prf:thetah->theta0}--\ref{prf:thm:ml-variance-cv}.

\section{Baseline setup}
\label{sec:problem}

\subsection{Value-at-risk and expected shortfall}
\label{sec1.1}

Let $(\Omega ,\mathcal{A},\mathbb{P})$ be a probability space on which
is defined a real-valued financial loss
\begin{equation}
\label{eq:L}
X_{0}=\mathbb{E}[\varphi (Y,Z)\,|\,Y],
\end{equation}
where $Y\in \mathbb{R}^{d}$, $Z\in \mathbb{R}^{q}$ are two independent
random variables and
$\varphi :\mathbb{R}^{d}\times \mathbb{R}^{q}\to \mathbb{R}$ is a measurable
function such that $\varphi (Y,Z)\in L^{1}(\mathbb{P})$. Given a portfolio
and a time horizon $\delta >0$, $Y$ and $Z$ typically model the portfolio's
risk factors respectively up to and beyond $\delta $,
$\varphi (Y,Z)$ the portfolio's subsequent future cash flows and
$X_{0}$ the portfolio's future loss at time $\delta $ (see Sects.~\ref{ssec:numerical}
and~\ref{s:swap}). Hence if $\varphi (y,Z)\in L^{1}(\mathbb{P})$ for all
$y\in \mathbb{R}^{d}$, we can rewrite $X_{0}$ as
\begin{equation}
\label{eq:X0}
X_{0} =\mathbb{E}[\varphi (Y,Z)\,|\,Y] =\Phi (Y), \qquad \text{with }
\Phi (y):=\mathbb{E}[\varphi (y,Z)], y\in \mathbb{R}^{d}.
\end{equation}

We are interested in computing the VaR and ES of $X_{0}$ at a given confidence
level \text{${\alpha \in (0,1)}$,} denoted respectively by
$\VaR _{\alpha}(X_{0})$ and $\ES _{\alpha}(X_{0})$. As per F\"ollmer and
Schied~\cite{doi:10.1002/9780470061602.eqf15003} and Acerbi and Tasche~\cite{ACERBI20021487},
these risk measures are defined by
\begin{align*}
\VaR _{\alpha}(X_{0}) &:=\inf{\{\xi \in \mathbb{R}:\mathbb{P}[X_{0}
\leq \xi ]\geq \alpha \}},
\\
\ES _{\alpha}(X_{0}) &:=\frac{1}{1-\alpha}\int _{\alpha}^{1}\VaR _{a}(X_{0})
\mathrm{d}a.
\end{align*}

\subsection{Stochastic approximation}
\label{sec1.2}

As established by Rockafellar and Uryasev~\cite{10.21314/JOR.2000.038},
the VaR and ES are linked by a convex optimisation problem. To accurately
state this connection, we introduce the mapping
\begin{equation}
\label{eq:V0}
V_{0}(\xi ):=\xi +\frac{1}{1-\alpha}\mathbb{E}[(X_{0}-\xi )^{+}],
\qquad \xi \in \mathbb{R}.
\end{equation}
If the distribution function $F_{X_{0}}$ of $X_{0}$ is continuous, then
$V_{0}$ is continuously differentiable on $\mathbb{R}$ and
\begin{equation*}
V_{0}'(\xi )=\frac{1}{1-\alpha}\big(F_{X_{0}}(\xi )-\alpha \big),
\qquad \xi \in \mathbb{R}.
\end{equation*}
If $X_{0}$ admits a continuous density function $f_{X_{0}}$, then
$V_{0}$ is twice continuously differentiable on $\mathbb{R}$ with
\begin{equation*}
V_{0}''(\xi )=\frac{1}{1-\alpha}f_{X_{0}}(\xi ),\qquad \xi \in
\mathbb{R}.
\end{equation*}

\begin{lemma}[{\cite[Proposition~2.1]{10.1007/978-3-642-04107-5_11}}]
\label{lmm:stochopt}
Suppose $X_{0}\in L^{1}(\mathbb{P})$ and the distribution function of~$X_{0}$
is continuous. Then $V_{0}$ is convex and continuously differentiable,
we have $\lim _{|\xi |\rightarrow \infty}V_{0}(\xi )=\infty $ and
\begin{equation*}
\VaR _{\alpha}(X_{0})=\min \Argmin{V_{0}},
\end{equation*}
where
\begin{equation*}
\Argmin{V_{0}} =\{V_{0}'=0\} =\{\xi \in \mathbb{R}:\mathbb{P}[X_{0}
\leq \xi ]=\alpha \}
\end{equation*}
is a closed bounded non-empty interval of $\mathbb{R}$. Moreover,
\begin{equation}
\label{eq:H1}
V_{0}'(\xi )=\mathbb{E}[H_{1}(\xi ,X_{0})], \qquad \text{where } H_{1}(
\xi ,x)=1-\frac{1}{1-\alpha}\mathds1_{\{x \geq \xi \}}, \xi ,x\in
\mathbb{R},
\end{equation}
and
\begin{equation*}
\ES _{\alpha}(X_{0})=\min{V_{0}}.
\end{equation*}
\end{lemma}

In general, the set $\Argmin{V_{0}}$ does not reduce to a singleton. But
if the distribution function of $X_{0}$ is increasing, it does and
$\argmin{V_{0}}=\VaR _{\alpha}(X_{0})$.
\\

Let $\xi ^{0}_{\star}\in \Argmin{V_{0}}$ and
$\chi ^{0}_{\star}=\min{V_{0}}=V_{0}(\xi ^{0}_{\star})$. In the case of
an exactly computable function $\Phi $ in \eqref{eq:X0} and a strictly
increasing distribution function of~$X_{0}$, the approach initiated by
Bardou et al.~\cite{10.1007/978-3-642-04107-5_11} applies and allows to
approximate the unique couple $(\xi ^{0}_{\star},\chi ^{0}_{\star})$ by
a sequence $(\xi ^{0}_{n},\chi ^{0}_{n})_{n\geq 0}$ given by the two-time-scale
stochastic approximation (SA) dynamics
\begin{equation}
\label{sgd:algorithm}
\begin{cases}
\displaystyle \xi ^{0}_{n+1}=\xi ^{0}_{n}-\gamma _{n+1}H_{1}(\xi ^{0}_{n},X_{0}^{(n+1)}),
\\[1ex]
\displaystyle \chi ^{0}_{n+1}=\chi ^{0}_{n}-\frac{1}{n+1}H_{2}(\xi ^{0}_{n},
\chi ^{0}_{n},X_{0}^{(n+1)}),
\end{cases}
\end{equation}
where
\begin{equation*}
H_{2}(\xi ,\chi ,x)=\chi -\bigg(\xi +\frac{1}{1-\alpha}(x-\xi )^{+}
\bigg), \qquad \xi ,\chi ,x\in \mathbb{R}.
\end{equation*}
The sequence $(X_{0}^{(n)})_{n\geq 1}$ stands for i.i.d.~copies of
$X_{0}$ and $(\xi ^{0}_{0},\chi ^{0}_{0})$ is a random vector independent
of $(X_{0}^{(n)})_{n\geq 1}$ with
$\mathbb{E}[|\xi ^{0}_{0}|^{2}]+\mathbb{E}[|\chi ^{0}_{0}|^{2}]<
\infty $. The learning rate sequence $(\gamma _{n})_{n\geq 1}$ in
\eqref{sgd:algorithm} is deterministic, positive, nonincreasing and such
that
\begin{equation*}
\sum _{n\geq 1}\gamma _{n}=\infty , \qquad \sum _{n\geq 1}\gamma _{n}^{2}<
\infty .
\end{equation*}
While such behaviour is typical for learning rates in the stochastic approximation
literature, it is not required in our analyses as we allow
$\sum _{n\geq 1}\gamma _{n}^{2}=\infty $. See Algorithm~\ref{alg:robbins-monro}.

\begin{algorithm}[H]
\caption{SA for estimating (VaR, ES)}
\label{alg:robbins-monro}
\begin{algorithmic}[1]
\Require {$N\in \mathbb{N}$, a strictly positive nonincreasing sequence $(\gamma _{n})_{n\geq 1}$ such that $\sum_{n=1}^\infty\gamma_n=\infty$ and $\lim_{n\to\infty}\gamma_n=0$.}
\State {Choose $(\xi ^{0}_{0},\chi ^{0}_{0})$ such that $\mathbb{E}[|\xi ^{0}_{0}|^{2}]+\mathbb{E}[|\chi ^{0}_{0}|^{2}]<\infty $}
\For {$n=0, \dots , N-1$}
   \State {Simulate $X_{0}^{(n+1)}\sim X_{0}=\Phi (Y)$}
   \State {$\xi ^{0}_{n+1}\gets \xi ^{0}_{n}-\gamma _{n+1}H_{1}(\xi ^{0}_{n},X_{0}^{(n+1)})$}
   \State {$\chi ^{0}_{n+1}\gets \chi ^{0}_{n}-\frac{1}{n+1}H_{2}(\chi ^{0}_{n},\xi ^{0}_{n},X_{0}^{(n+1)})$}
\EndFor
\State \Return {$(\xi ^{0}_{N},\chi ^{0}_{N})$}
\end{algorithmic}
\end{algorithm}

Usually however, one does not have access to an exact simulator of
$X_{0}$ as given by~\eqref{eq:L}, because the law of $\varphi (Y,Z)$ conditionally
on $Y$ is not known and no analytical expression of $\Phi $ in
\eqref{eq:X0} is available.

\section{Nested stochastic approximation}
\label{sec:nested-sa}

The above discussion naturally suggests to replace the samples of
$X_{0}$ in the dynamics~\eqref{sgd:algorithm} by approximate samples. We
consider a bias parameter
\begin{equation*}
h=\frac{1}{K}\in \mathcal{H}:=\bigg\{\frac{1}{K}: K\in \mathbb{N}
\bigg\}.
\end{equation*}
We then approximate $X_{0}$ by the random variable $X_{h}$ defined by
\begin{equation}
\label{eq:Xh}
X_{h}=\frac{1}{K}\sum _{k=1}^{K}\varphi (Y,Z^{(k)}),
\end{equation}
where the sequence $(Z^{(k)})_{1\leq k\leq K}$ consists of i.i.d.~copies
of $Z$ independent of $Y$. In order to simulate $X_{h}$, it suffices to
sample $Y$, then independently sample $(Z^{(k)})_{1\leq k\leq K}$, to eventually
obtain $X_{h}$ as the sample mean of
$(\varphi (Y,Z^{(k)}))_{1\leq k\leq K}$.
\\

In the spirit of Sect.~\ref{sec:problem}, assuming that the distribution
function $F_{X_{h}}$ of $X_{h}$ is continuous, we define the approximating
optimisation problem
\begin{equation}
\label{eq:Vh}
\min _{\xi \in \mathbb{R}}{V_{h}(\xi )}, \qquad \text{where } V_{h}(
\xi ):=\xi +\frac{1}{1-\alpha}\mathbb{E}[(X_{h}-\xi )^{+}], \xi \in
\mathbb{R}.
\end{equation}
Because $F_{X_{h}}$ is continuous, $V_{h}$ is continuously differentiable
on $\mathbb{R}$ and we have \text{${V_{h}'(\xi )=\frac{1}{1-\alpha}(F_{X_{h}}(
\xi )-\alpha )}$} for $\xi \in \mathbb{R}$. If $X_{h}$ admits a continuous
density function~$f_{X_{h}}$, then $V_{h}$ is twice continuously differentiable
on $\mathbb{R}$ and \text{${V_{h}''(\xi )=\frac{1}{1-\alpha}f_{X_{h}}(
\xi )}$} for $\xi \in \mathbb{R}$. By Lemma~\ref{lmm:stochopt},
\begin{equation*}
\Argmin{V_{h}}=\{V_{h}'=0\}\neq \varnothing \qquad \text{and}\qquad
\chi ^{h}_{\star}=\min{V_{h}}=V_{h}(\xi ^{h}_{\star}).
\end{equation*}
Moreover, any minimiser $\xi ^{h}_{\star}$ of $V_{h}$ satisfies
$\mathbb{P}[X_{h}\leq \xi ^{h}_{\star}]=\alpha $. If the distribution function
of $X_{h}$ is strictly increasing, then
$\xi ^{h}_{\star}=\argmin{V_{h}}$ is uniquely defined.
\\

If the distribution function of $X_{h}$ is strictly increasing, in order
to approximate the unique couple
$(\xi ^{h}_{\star},\chi ^{h}_{\star})$, we devise the two-time-scale nested
SA (NSA) algorithm
\begin{equation}
\label{approximate:sgd:algorithm}
\begin{cases}
\displaystyle \xi ^{h}_{n+1} =\xi ^{h}_{n}-\gamma _{n+1}H_{1}(\xi ^{h}_{n},X_{h}^{(n+1)}),
\\[1ex]
\displaystyle \chi ^{h}_{n+1} =\chi ^{h}_{n}-\frac{1}{n+1}H_{2}(\xi ^{h}_{n},
\chi ^{h}_{n},X_{h}^{(n+1)}),
\end{cases}
\end{equation}
where $(X_{h}^{(n)})_{n\geq 1}$ is a sequence of i.i.d.~copies of
$X_{h}$ and $(\xi ^{h}_{0},\chi ^{h}_{0})$ is an $\mathbb{R}^{2}$-valued
random variable independent of $(X_{h}^{(n)})_{n\geq 1}$ with
$\mathbb{E}[|\xi ^{h}_{0}|^{2}]+\mathbb{E}[|\chi ^{h}_{0}|^{2}]<
\infty $. This scheme is nested in nature in the sense that the update
of the outer layer $(\xi ^{h}_{n+1},\chi ^{h}_{n+1})$ at step $n+1$ entails
simulating the inner layer $X_{h}^{(n+1)}$ by Monte Carlo methods as in~\eqref{eq:Xh}.
Besides, this scheme is biased as the target of
$(\xi ^{h}_{n},\chi ^{h}_{n})_{n\geq 0}$ is
$(\xi ^{h}_{\star},\chi ^{h}_{\star})$, which hopefully converges to
$(\xi ^{0}_{\star},\chi ^{0}_{\star})$ as
$\mathcal{H}\ni h\downarrow 0$.
\\

Algorithm~\ref{alg:nested-sa} summarises the NSA procedure for approximating
$(\xi ^{0}_{\star},\chi ^{0}_{\star})$.

\begin{algorithm}[H]
\caption{Nested SA for estimating (VaR, ES)}
\label{alg:nested-sa}
\begin{algorithmic}[1]
\Require {$K,N\in \mathbb{N}$, a strictly positive nonincreasing sequence $(\gamma _{n})_{n\geq 1}$ such that $\sum_{n=1}^\infty\gamma_n=\infty$ and $\lim_{n\to\infty}\gamma_n=0$.}
\State {Choose $(\xi ^{h}_{0},\chi ^{h}_{0})$ such that $\mathbb{E}[|\xi ^{h}_{0}|^{2}]+\mathbb{E}[|\chi ^{h}_{0}|^{2}]<\infty $}
\For {$n=0,\dots ,N-1$}
   \State {Simulate $Y^{(n+1)}\sim Y$ and $Z^{(n+1,1)},\dots ,Z^{(n+1,K)}\sim Z$ i.i.d.~and independent of $Y^{(n+1)}$}
   \State {$X_{h}^{(n+1)}\gets \frac{1}{K}\sum _{k=1}^{K}\varphi (Y^{(n+1)},Z^{(n+1,k)})$}
   \State {$\xi ^{h}_{n+1}\gets \xi ^{h}_{n}-\gamma _{n+1}H_{1}(\xi ^{h}_{n},X_{h}^{(n+1)})$}
   \State {$\chi ^{h}_{n+1}\gets \chi ^{h}_{n}-\frac{1}{n+1}H_{2}(\chi ^{h}_{n},\xi ^{h}_{n},X_{h}^{(n+1)})$}
\EndFor
\State \Return {$(\xi ^{h}_{N},\chi ^{h}_{N})$}
\end{algorithmic}
\end{algorithm}

\subsection{Convergence analysis}
\label{sect.2.1}

We study here the NSA scheme \eqref{approximate:sgd:algorithm}. We first
analyse its bias, then prove an \text{$L^{2}(\mathbb{P})$-error} bound on
the iterates, after that analyse its complexity and eventually establish
a tuning method for the number of iterations with respect to the bias parameter
$h\in \mathcal{H}$ in order to achieve some prescribed error.
\\

For $h\in \overline{\mathcal{H}}:=\mathcal{H}\cup \{0\}$, we set
$\Theta _{h}:=\Argmin{V_{h}}$, assuming that the distribution function
$F_{X_{h}}$ of $X_{h}$ is continuous.

\begin{lemma}
\label{lmm:thetah->theta0}
Suppose that $\varphi (Y,Z)\in L^{1}(\mathbb{P})$, that the distribution
function $F_{X_{h}}$ is continuous for any
$h\in \overline{\mathcal{H}}$ and that the sequence
$(X_{h})_{h\in \mathcal{H}}$ of random variables converges in distribution
to $X_{0}$ as $\mathcal{H}\ni h\downarrow 0$. Then
\begin{equation*}
\sup _{\xi \in \Theta _{h}} \dist (\xi ,\Theta _{0})\to 0
\qquad \text{as } \mathcal{H}\ni h\downarrow 0.
\end{equation*}
If additionally $(X_{h})_{h\in \mathcal{H}}$ converges to $X_{0}$ in
$L^{1}(\mathbb{P})$, then
\begin{equation*}
\chi ^{h}_{\star}\to \chi ^{0}_{\star }\qquad \text{as }
\mathcal{H}\ni h\downarrow 0.
\end{equation*}
\end{lemma}

\begin{proof}
See Appendix~\ref{prf:thetah->theta0}.
\end{proof}

We introduce the following set of assumptions on the sequence
$(X_{h})_{h\in \overline{\mathcal{H}}}$.

\begin{assumption}
\label{asp:Xh->X0}
\begin{enumerate}[\rm(i)]
\item
\label{asp:Xh->X0-ii}
For any $h\in \mathcal{H}$, $F_{X_{h}}$ admits the first-order Taylor expansion
\begin{equation}
\label{neugl1}
F_{X_{h}}(\xi )-F_{X_{0}}(\xi )=v(\xi )h+w(\xi ,h)h, \qquad \xi \in
\mathbb{R},
\end{equation}
for some functions $v,w(\,\cdot \,,h):\mathbb{R}\to \mathbb{R}$ satisfying,
for any $\xi ^{0}_{\star}\in \Theta _{0}$,
\begin{equation}
\label{neugl2}
\int ^{\infty}_{\xi ^{0}_{\star}}v(\xi )\mathrm{d}\xi <\infty ,
\qquad \lim _{\mathcal{H}\ni h\downarrow 0}w(\xi ^{0}_{\star},h)=
\lim _{\mathcal{H}\ni h\downarrow 0}\int ^{\infty}_{\xi ^{0}_{\star}}w(
\xi ,h)\mathrm{d}\xi =0.
\end{equation}
\item
\label{asp:Xh->X0-iii}
For any $h\in \overline{\mathcal{H}}$, the law of $X_{h}$ admits a continuous
density function $f_{X_{h}}$ with respect to Lebesgue measure. Moreover,
the sequence $(f_{X_{h}})_{h\in \mathcal{H}}$ of density functions converges
locally uniformly towards $f_{X_{0}}$.
\end{enumerate}
\end{assumption}

\begin{remark}
\label{rem2.3}
We refer to Giorgi et al.~\cite[Proposition~5.1(b)]{Giorgi2020} for a result
on the satisfaction of \eqref{neugl1}. The condition~\eqref{neugl2} reads
also
$w(\xi ^{0}_{\star},h)=\int _{\xi ^{0}_{\star}}^{\infty }w(\xi ,h)
\mathrm{d}\xi =\oo (1)$ as $\mathcal{H}\ni h\downarrow 0$, i.e.,
\begin{equation*}
\begin{aligned}
F_{X_{h}}(\xi ^{0}_{\star})-F_{X_{0}}(\xi ^{0}_{\star}) &=v(\xi ^{0}_{
\star})h+\oo (h),
\\
\int _{\xi ^{0}_{\star}}^{\infty}\big(F_{X_{h}}(\xi )-F_{X_{0}}(\xi )
\big)\mathrm{d}\xi &=h\int _{\xi ^{0}_{\star}}^{\infty }v(\xi )
\mathrm{d}\xi +\oo (h).
\end{aligned}
\end{equation*}
Assumption~\ref{asp:Xh->X0}(\ref{asp:Xh->X0-iii}) is a weakened postulate
in comparison with the second part of \cite[Proposition~5.1(a)]{Giorgi2020}.
\end{remark}

The result below quantifies the weak error implied by approximating the
unbiased problem with the biased one, in the form of a first-order expansion
in the bias parameter~$h$ of the error between
$(\xi ^{h}_{\star},\chi ^{h}_{\star})$ and
$(\xi ^{0}_{\star},\chi ^{0}_{\star})$.

\begin{proposition}
\label{prp:bias-cv}
Suppose that $\varphi (Y,Z)\in L^{1}(\mathbb{P})$, Assumption~\ref{asp:Xh->X0}
holds and the density function $f_{X_{0}}$ is strictly positive. Then as
$\mathcal{H}\ni h\downarrow 0$, we have for any
$\xi ^{h}_{\star}\in \Theta _{h}$ that
\begin{equation*}
\xi ^{h}_{\star}-\xi ^{0}_{\star }=-
\frac{v(\xi ^{0}_{\star})}{f_{X_{0}}(\xi ^{0}_{\star})}h+\oo (h) ,
\qquad \chi ^{h}_{\star}-\chi ^{0}_{\star }=-h\int _{\xi ^{0}_{\star}}^{
\infty }\frac{v(\xi )}{1-\alpha}\mathrm{d}\xi +\oo (h).
\end{equation*}
\end{proposition}

\begin{proof}
See Appendix~\ref{prf:bias-cv}.
\end{proof}

\begin{assumption}
\label{asp:misc}
\begin{enumerate}[\rm(i)]
\item
\label{asp:misc-iii}
For any $R>0$,
\begin{equation*}
\inf _{
\substack{h\in \overline{\mathcal{H}}\\\xi \in B(\xi ^{0}_{\star},R)}}{f_{X_{h}}(
\xi )}>0.
\end{equation*}

\item
\label{asp:misc-iv}
The density functions $(f_{X_{h}})_{h\in \overline{\mathcal{H}}}$ are uniformly
bounded and uniformly Lip\-schitz, i.e.,
\begin{equation*}
\sup _{h\in \overline{\mathcal{H}}}{(\|f_{X_{h}}\|_{\infty}+[f_{X_{h}}]_{
{\mathrm{Lip}}})}<\infty .
\end{equation*}
\end{enumerate}
\end{assumption}

\begin{remark}
\label{rem2.6}
Assumption~\ref{asp:misc}(\ref{asp:misc-iii}) is natural in view of Assumption~\ref{asp:Xh->X0}(\ref{asp:Xh->X0-iii})
if one assumes that the distribution function $F_{X_{0}}$ is strictly increasing.
Assumption~\ref{asp:misc}(\ref{asp:misc-iv}) is in line with Assumption~\ref{asp:Xh->X0}(\ref{asp:Xh->X0-iii}).
\end{remark}

Henceforth, we set $k_{\alpha}:=1\vee \frac{\alpha}{1-\alpha}$. Using the
definition of $(\bar\lambda _{h,q})_{h\in \overline{\mathcal{H}}}$ in Lemma~\ref{lmm:lyapunov}(\ref{lmm:lyapunov-ii}),
we let, for any positive integer $q$,
\begin{equation}
\label{eq:lambda}
\bar\lambda _{q}:=\inf _{h\in \overline{\mathcal{H}}}{\bar\lambda _{h,q}}.
\end{equation}

\begin{theorem}
\label{thm:variance-cv}
\begin{enumerate}[\rm(i)]
\item\label{thm:variance-cv:i}
Suppose that $\varphi (Y,Z)\in L^{2}(\mathbb{P})$, Assumptions~\ref{asp:Xh->X0}
and~\ref{asp:misc} hold and
\begin{equation*}
\sup _{h\in \overline{\mathcal{H}}}{\mathbb{E}\bigg[|\xi ^{h}_{0}|^{2}
\exp \bigg(\frac{2}{1-\alpha}k_{\alpha }\sup _{h'\in
\overline{\mathcal{H}}}\|f_{X_{h'}}\|_{\infty}|\xi _{0}^{h}|\bigg)
\bigg]}+\sup _{h\in \overline{\mathcal{H}}}\mathbb{E}[|\chi _{0}^{h}|]<
\infty .
\end{equation*}
If $\gamma _{n}=\gamma _{1}n^{-\beta}$, $\beta \in (0,1]$, with
$\bar\lambda _{1}\gamma _{1}>1 $ if $\beta =1$, then for any
$h\in \mathcal{H}$ and any positive integer $n$,
\begin{equation}
\label{uniform:L2:bound:var:alg}
\mathbb{E}[(\xi ^{h}_{n}-\xi ^{h}_{\star})^{2}]\leq \bar{K}^{\beta}_{h,2}
\gamma _{n}
\end{equation}
for some constants $(\bar{K}^{\beta}_{h,2})_{h\in \mathcal{H}}$ such that
$\sup _{h\in \mathcal{H}}{\bar{K}_{h,2}^{\beta}}<\infty $. Moreover,
\begin{equation*}
\mathbb{E}[|\chi ^{h}_{n}-\chi ^{h}_{\star}|]\leq
\frac{C_{h}^{\beta}}{n^{\frac{1}{2}\wedge \beta}}
\end{equation*}
for some constants $(C_{h}^{\beta})_{h\in \mathcal{H}}$ such that
$\sup _{h\in \mathcal{H}}{C_{h}^{\beta}}<\infty $.
\item
\label{thm:variance-cv:ii}
Assume further that
\begin{equation*}
\sup _{h\in \overline{\mathcal{H}}}{\mathbb{E}\bigg[|\xi ^{h}_{0}|^{4}
\exp \bigg(\frac{4}{1-\alpha}k_{\alpha }\sup _{h'\in
\overline{\mathcal{H}}}\|f_{X_{h'}}\|_{\infty}|\xi _{0}^{h}|\bigg)
\bigg]}+\sup _{h\in \overline{\mathcal{H}}}\mathbb{E}[(\chi _{0}^{h})^{2}]<
\infty .
\end{equation*}
If $\gamma _{n}=\gamma _{1}n^{-\beta}$, $\beta \in (0,1]$, with
$\bar\lambda _{2}\gamma _{1}>2 $ when $\beta =1$, then for any
$h\in \mathcal{H}$ and any positive integer $n$,
\begin{equation}
\label{uniform:L4:bound:var:alg}
\mathbb{E}[(\xi ^{h}_{n}-\xi ^{h}_{\star})^{4}]\leq \bar{K}^{\beta}_{h,4}
\gamma _{n}^{2}
\end{equation}
for some constants $(\bar{K}^{\beta}_{h,4})_{h\in \mathcal{H}}$ satisfying
$\sup _{h\in \mathcal{H}}{\bar{K}_{h,4}^{\beta}}<\infty $. Moreover,
\begin{equation*}
\mathbb{E}[(\chi ^{h}_{n}-\chi ^{h}_{\star})^{2}]\leq
\frac{\bar{C}_{h}^{\beta}}{n^{1\wedge 2\beta}}
\end{equation*}
for some constants $(\bar{C}_{h}^{\beta})_{h\in \mathcal{H}}$ such that
$\sup _{h\in \mathcal{H}}{\bar{C}_{h}^{\beta}}<\infty $.
\end{enumerate}
\end{theorem}

\begin{proof}
See Appendix~\ref{prf:variance-cv}.
\end{proof}

\subsection{Complexity analysis}
\label{sec2.2}

Proposition~\ref{prp:bias-cv} and Theorem~\ref{thm:variance-cv}(\ref{thm:variance-cv:i})
suggest a heuristic to balance the bias parameter $h\in \mathcal{H}$ with
respect to the number $n$ of steps in
\eqref{approximate:sgd:algorithm} in order to achieve a prescribed accuracy
$\varepsilon \in (0,1)$.

\begin{theorem}
\label{thm:cost:nsa}
Let $\varepsilon \in (0,1)$ be a prescribed accuracy. Within the framework
of Theorem~\ref{thm:variance-cv}(\ref{thm:variance-cv:i}), setting
\begin{equation*}
h=\frac{1}{\lceil \varepsilon ^{-1}\rceil} \qquad \text{and}\qquad n=
\lceil \varepsilon ^{-\frac{2}{\beta}} \rceil
\end{equation*}
yields a global $L^{1}(\mathbb{P})$-error of the order
$\varepsilon $. The corresponding computational cost satisfies
\begin{equation*}
\Cost _{{\mathrm{NSA}}}\leq Cnh^{-1}\leq C\varepsilon ^{-
\frac{2}{\beta}-1} \qquad \text{as } \varepsilon \downarrow 0
\end{equation*}
for some constant $C$ that may change from one occurrence to the next,
but is independent of~$\varepsilon $. The minimal computational cost is
of the order~$\varepsilon ^{-3}$ and is attained for $\beta = 1$, i.e.,
$\gamma _{n} = \gamma _{1} n^{-1}$, under the restriction
$\bar{\lambda}_{1}\gamma _{1} > 1$.
\end{theorem}

\begin{proof}
Following Proposition~\ref{prp:bias-cv} and Theorem~\ref{thm:variance-cv}(\ref{thm:variance-cv:i}),
if $\gamma _{n}=\gamma _{1}n^{-\beta}$, $\beta \in (0,1]$, with
$\bar\lambda _{1}\gamma _{1}>1 $ if $\beta =1$, then there exists
$C>0$ such that for any $h\in \mathcal{H}$ and any positive integer
$n$,
\begin{equation*}
\mathbb{E}[|\xi ^{h}_{n}-\xi ^{0}_{\star}|]\leq C(h+n^{-
\frac{\beta}{2}}) ,\qquad \mathbb{E}[|\chi ^{h}_{n}-\chi ^{0}_{\star}|]
\leq C(h+n^{-\frac{1}{2}\wedge \beta}).
\end{equation*}
The assertions follow easily from these upper bounds.
\end{proof}

\begin{remark}
\label{rem2.9}
Barrera et al.~\cite[Algorithm~1]{barrera:hal-01710394} differs from our
NSA scheme in two ways. First, they use a single-time-scale scheme for
both the VaR and ES. We use in contrast a two-time-scale scheme, with a
slower and more precise VaR component (independent from the ES) and a faster
ES component. Second, the $n$th iteration of their NSA algorithm uses a
bias parameter $h_{n}$ dependent upon $n$, where $(h_{n})_{n\geq 1}$ tends
to~$0$ as $n\to \infty $, hence building iterates such that
$(\xi ^{h_{n}}_{n},\chi ^{h_{n}}_{n})\to (\xi ^{0}_{\star},\chi ^{0}_{
\star})$ as $n\to \infty $. However, their approach results in a substantially
increased complexity. Indeed, in view of~\cite[Theorem~3.2]{barrera:hal-01710394},
the error is of the order
$\sqrt{\gamma _{n}}=\gamma _{1}^{1/2}n^{-\beta /2}$,
$\beta \in (\frac{1}{2},1]$, $n\geq 1$. For a prescribed accuracy
$\varepsilon >0$, one has to choose
$n=\lceil \varepsilon ^{-2/\beta}\rceil $. In~\cite[Sect.~3.2]{barrera:hal-01710394},
the authors recommend selecting $h_{n}=n^{-\beta '}$,
$\beta '>\beta $, $n\geq 1$. This results in a complexity of
$\sum _{k=1}^{n}h_{k}^{-1}\leq C\varepsilon ^{-2(1+\beta ')/\beta}$, which
is optimal when $\beta '\to \beta =1$, yielding a complexity of the order
$\varepsilon ^{-4}$. Our NSA algorithm
\eqref{approximate:sgd:algorithm}, in contrast, has an optimal complexity
of the order $\varepsilon ^{-3}$.
\end{remark}

The next section is devoted to the MLSA scheme which combines multiple
paired NSA estimators obtained through a geometric progression of bias
parameters in order to reduce the complexity.

\section{Multilevel stochastic approximation}
\label{sec:mlsa}

We first explain the construction of the multilevel SA algorithm to approximate
the VaR and ES; then we provide an analysis thereof from both convergence
and complexity viewpoints.
\\

Let $h_{0}=\frac{1}{K}\in \mathcal{H}$ be a bias parameter and
$M \geq 2$, $L$ two positive integers. $L$~is called the number of levels.
Consider the geometric progression of bias parameters
\begin{equation*}
h_{\ell}=\frac{h_{0}}{M^{\ell}}=\frac{1}{KM^{\ell}}, \qquad \ell \in
\{0,\dots ,L\}.
\end{equation*}
Assuming that the distribution function of $X_{h_{\ell}}$ is continuous
and strictly increasing for any \text{${\ell \in \{0,\dots ,L\}}$,} we let
$(\xi ^{h_{\ell}}_{\star},\chi ^{h_{\ell}}_{\star})=(\argmin{V_{h_{
\ell}}},\min{V_{h_{\ell}}})$, $\ell \in \{0,\dots ,L\}$. We thus have the
telescopic decompositions
\begin{equation*}
\xi ^{h_{L}}_{\star }=\xi ^{h_{0}}_{\star}+\sum _{\ell =1}^{L}(\xi ^{h_{
\ell}}_{\star}-\xi ^{h_{\ell -1}}_{\star}), \qquad \chi ^{h_{L}}_{
\star }=\chi ^{h_{0}}_{\star}+\sum _{\ell =1}^{L}(\chi ^{h_{\ell}}_{
\star}-\chi ^{h_{\ell -1}}_{\star}).
\end{equation*}

While the idea of NSA is to put the entirety of the available computational
power behind simulating the quantities on the left-hand sides above, MLSA
rather simulates the telescopic decompositions on the right-hand sides.
It starts by simulating the level-$0$ terms $\xi ^{h_{0}}_{\star}$ and
$\chi ^{h_{0}}_{\star}$ which are, although of high bias, very fast to
estimate. It then simulates independent incremental error corrections
$(\xi ^{h_{\ell}}_{\star}-\xi ^{h_{\ell -1}}_{\star},\chi ^{h_{\ell}}_{
\star}-\chi ^{h_{\ell -1}}_{\star})_{1\leq \ell \leq L}$ to the initial
simulations.
\\

In Sect.~\ref{sec:nested-sa}, we saw how NSA can approximate each pair
$(\xi ^{h_{\ell}}_{\star},\chi ^{h_{\ell}}_{\star})$,
$\ell \in \{0,\dots ,L\}$. Let
$\mathbf{N}=(N_{0},\dots ,N_{L})\in \mathbb{N}^{L+1}$. Following Frikha~\cite{10.1214/15-AAP1109},
we define the multilevel~SA estimators
$\xi ^{{\mathrm{ML}}}_{\mathbf{N}}$ and
$\chi ^{{\mathrm{ML}}}_{\mathbf{N}}$ of $\xi ^{h_{L}}_{\star}$ and
$\chi ^{h_{L}}_{\star}$ by
\begin{equation}
\label{eq:ml}
\begin{cases}
\displaystyle \xi ^{{\mathrm{ML}}}_{\mathbf{N}} =\xi ^{h_{0}}_{N_{0}}+
\sum _{\ell =1}^{L}(\xi ^{h_{\ell}}_{N_{\ell}}-\xi ^{h_{\ell -1}}_{N_{
\ell}}),
\\
\displaystyle \chi ^{{\mathrm{ML}}}_{\mathbf{N}} =\chi ^{h_{0}}_{N_{0}}+
\sum _{\ell =1}^{L}(\chi ^{h_{\ell}}_{N_{\ell}}-\chi ^{h_{\ell -1}}_{N_{
\ell}}),
\end{cases}
\end{equation}
where at any level $\ell \in \{0,\dots ,L\}$, the initialisations
$(\xi ^{h_{\ell}}_{0},\chi ^{h_{\ell}}_{0})$ are generated such that
$\mathbb{E}[|\xi ^{h_{\ell}}_{0}|^{2}]+\mathbb{E}[|\chi ^{h_{\ell}}_{0}|^{2}]<
\infty $ and the iterates
$(\xi ^{h_{\ell}}_{n},\chi ^{h_{\ell}}_{n})_{n\geq 1}$ are computed using
the NSA~scheme \eqref{approximate:sgd:algorithm}. We stress that for any
fixed level $\ell \in \{1,\dots ,L\}$, the random variables
$(X^{(n)}_{h_{\ell -1}},X^{(n)}_{h_{\ell}})_{1\leq n\leq N_{\ell}}$ used
to compute
$(\xi ^{h_{\ell -1}}_{N_{\ell}},\xi ^{h_{\ell}}_{N_{\ell}})$ and
$(\chi _{N_{\ell}}^{h_{\ell -1}}, \chi _{N_{\ell}}^{h_{\ell}}) $ are i.i.d.~with
the same law as $(X_{h_{\ell -1}},X_{h_{\ell}})$, where
\begin{equation*}
X_{h_{\ell -1}}=\frac{1}{KM^{\ell -1}}\sum _{k=1}^{KM^{\ell -1}}
\varphi (Y,Z^{(k)})
\end{equation*}
and $X_{h_{\ell}}$ is obtained from $X_{h_{\ell -1}}$ via
\begin{equation*}
X_{h_{\ell}}=\frac{1}{M}X_{h_{\ell -1}}+\frac{1}{KM^{\ell}}\sum _{k=KM^{
\ell -1}+1}^{KM^{\ell}}\varphi (Y,Z^{(k)}).
\end{equation*}

Algorithm~\ref{alg:mlsa} summarises this process.

\begin{algorithm}[H]
\caption{Multilevel SA for estimating $(\mathrm{VaR},\mathrm{ES})$}
\label{alg:mlsa}
\begin{algorithmic}[1]
\Require {A number $L\geq 1$ of levels, a bias parameter $h_{0}=\frac{1}{K}\in \mathcal{H}$, a geometric step size $M\geq 2$, strictly positive integers $N_{0},\dots ,N_{L}$, a strictly positive nonincreasing sequence
$(\gamma _{n})_{n\geq 1}$  such that $\sum_{n=1}^\infty\gamma_n=\infty$ and $\lim_{n\to\infty}\gamma_n=0$.}
\For {$\ell =0, \dots , L$}
   \State {Set $h_{\ell}\gets \frac{h_{0}}{M^{\ell}}$}
   \For {$j=(\ell -1)^{+}, \dots , \ell $}
      \State {Choose $(\xi ^{h_{j}}_{0},\chi ^{h_{j}}_{0})$ such that $\mathbb{E}[|\xi ^{h_{j}}_{0}|^{2}]+\mathbb{E}[|\chi ^{h_{j}}_{0}|^{2}]<\infty $}
   \EndFor
   \For {$n=0, \dots , N_{\ell}-1$}
      \State {Simulate $Y^{(n+1)}\sim Y$}
      \State {Simulate $Z^{(n+1,1)},\dots ,Z^{(n+1,KM^{\ell})} \sim Z$  i.i.d.~and independent from $Y^{(n+1)}$}
      \If {$\ell =0$}
         \State {$X_{h_{0}}^{(n+1)}\gets \frac{1}{K}\sum _{k=1}^{K}\varphi (Y^{(n+1)},Z^{(n+1,k)})$}
      \Else
         \State {$X_{h_{\ell -1}}^{(n+1)}\gets \frac{1}{KM^{\ell -1}}\sum _{k=1}^{KM^{\ell -1}}\varphi (Y^{(n+1)},Z^{(n+1,k)})$}
         \State {$X_{h_{\ell}}^{(n+1)}\gets \frac{1}{M}X^{(n+1)}_{h_{\ell -1}}+\frac{1}{KM^{\ell}}\sum _{k=KM^{\ell -1}+1}^{KM^{\ell}}\varphi (Y^{(n+1)},Z^{(n+1,k)})$}
      \EndIf
      \For {$j=(\ell -1)^{+}, \dots ,\ell $}
         \State {$\xi ^{h_{j}}_{n+1}\gets \xi ^{h_{j}}_{n}-\gamma _{n+1}H_{1}(\xi ^{h_{j}}_{n},X_{h_{j}}^{(n+1)})$}
         \State {$\chi ^{h_{j}}_{n+1}\gets \chi ^{h_{j}}_{n}-\frac{1}{n+1}H_{2}(\chi ^{h_{j}}_{n},\xi ^{h_{j}}_{n},X_{h_{j}}^{(n+1)})$}
      \EndFor
   \EndFor
\EndFor
\State {$\xi ^{\text{\rm ML}}_{\mathbf{N}}\gets \xi ^{h_{0}}_{N_{0}}+\sum _{\ell =1}^{L} (\xi ^{h_{\ell}}_{N_{\ell}}-\xi ^{h_{\ell -1}}_{N_{\ell}})$}
\State {$\chi ^{\text{\rm ML}}_{\mathbf{N}}\gets \chi ^{h_{0}}_{N_{0}}+\sum _{\ell =1}^{L} (\chi ^{h_{\ell}}_{N_{\ell}}-\chi ^{h_{\ell -1}}_{N_{\ell}})$}
\State \Return {$(\xi ^{\text{\rm ML}}_{\mathbf{N}},\chi ^{\text{\rm ML}}_{\mathbf{N}})$}
\end{algorithmic}
\end{algorithm}

\begin{remark}
\label{rem3.1}
Intuitively, the larger $\ell $, the closer the random variables
$X_{h_{\ell -1}}$ and $X_{h_{\ell}}$ are to $X_{0}$ and the iterations
$N_{\ell}$ are required at the level $\ell $ of the MLSA scheme~\eqref{eq:ml}
to achieve a high accuracy. Hence $N_{0}\geq \cdots \geq N_{L}$.
\end{remark}

\subsection{Convergence analysis}
\label{sect.3.1}

The global error between the multilevel estimator
$(\xi ^{{\mathrm{ML}}}_{\mathbf{N}},\chi ^{{\mathrm{ML}}}_{\mathbf{N}})$
and its target $(\xi ^{0}_{\star},\chi ^{0}_{\star})$ can be decomposed
into a sum of a statistical error and a bias error via
\begin{align}
\xi ^{{\mathrm{ML}}}_{\mathbf{N}}-\xi ^{0}_{\star }&=(\xi ^{
{\mathrm{ML}}}_{\mathbf{N}}-\xi ^{h_{L}}_{\star})+(\xi ^{h_{L}}_{\star}-
\xi ^{0}_{\star}),
\label{var:error}
\\
\chi ^{{\mathrm{ML}}}_{\mathbf{N}}-\chi ^{0}_{\star }&=(\chi ^{
{\mathrm{ML}}}_{\mathbf{N}}-\chi ^{h_{L}}_{\star})+(\chi ^{h_{L}}_{
\star}-\chi ^{0}_{\star}).
\label{es:error}
\end{align}

In the ensuing analysis, we quantify each error appearing in the decompositions
above in terms of the parameters of \eqref{eq:ml}. We then propose an optimal
choice for $L$ and $N_{0}, \dots , N_{L}$ to achieve a given prescribed
accuracy.

\begin{proposition}
\label{prop:local:strong:error:indicator:func}
\begin{enumerate}[\rm(i)]
\item
\label{prop:local:strong:error:indicator:func-i}
Assume that the real-valued random variables $X_{h}$ admit density functions
$f_{X_{h}}$ that are bounded uniformly in
$h\in \overline{\mathcal{H}}$.
\begin{enumerate}[\rm a.]
\item
\label{prop:local:strong:error:indicator:func-ia}
If
\begin{equation}
\label{assumption:finite:Lp:moment}
\mathbb{E}\big[\big|\varphi (Y,Z)-\mathbb{E}[\varphi (Y,Z)\,|\,Y]
\big|^{p_{\star}}\big]<\infty \qquad
\text{holds for some }p_{\star}>1,
\end{equation}
then for any $h,h'\in \overline{\mathcal{H}}$ such that
$0\leq h\leq h'$ and any $\xi \in \mathbb{R}$,
\begin{equation*}
\mathbb{E}[|\mathds1_{\{X_{h}>\xi \}}-\mathds1_{\{X_{h'}>\xi \}}|]
\leq C(h'-h)^{\frac{p_{\star}}{2(p_{\star}+1)}},
\end{equation*}
where
\begin{equation*}
C:= B_{p_{\star}}\mathbb{E}\big[\big|\varphi (Y, Z)-\mathbb{E}[
\varphi (Y, Z)\,|\,Y]\big|^{p_{\star}}\big]^{\frac{1}{p_{\star}+1}}
\Big(\sup _{h\in \overline{\mathcal{H}}}\|f_{X_{h}}\|_{\infty}\Big)^{
\frac{p_{\star}}{p_{\star}+1}}
\end{equation*}
with a strictly positive constant $B_{p_{\star}}$ that depends only upon
$p_{\star}$.
\item
\label{prop:local:strong:error:indicator:func-ib}
Assume there exists a nonnegative constant $C_{g}<\infty $ such that for
all~$u\in \mathbb{R}$,
\begin{equation}
\label{assumption:conditional:gaussian:concentration}
\mathbb{E}\Big[\exp \Big(u\big(\varphi (Y,Z)-\mathbb{E}[\varphi (Y,Z)
\,|\,Y]\big)\Big)\,\Big|\,Y\Big]\leq \mathrm{e} ^{C_{g}u^{2}}\qquad \Pas
\end{equation}
Then for any $h,h'\in \overline{\mathcal{H}}$ such that $0\leq h<h'$ and
any $\xi \in \mathbb{R}$,
\begin{equation*}
\mathbb{E}[|\mathds1_{\{X_{h}>\xi \}}-\mathds1_{\{X_{h'}>\xi \}}|]
\leq 2\sqrt{C_{g}(h'-h)}\Big(1+\sup _{h\in \overline{\mathcal{H}}}\|f_{X_{h}}
\|_{\infty}\sqrt{2\big|\ln{\big(C_{g}(h'-h)\big)}\big|}\Big).
\end{equation*}
\end{enumerate}
\item
\label{prop:local:strong:error:indicator:func-ii}
Let
$G_{\ell}:= h_{\ell}^{-\frac{1}{2}}(X_{h_{\ell}}-X_{h_{\ell -1}})$ and
define
\begin{equation*}
F_{X_{h_{\ell -1}}|G_{\ell}=g} (x) := \mathbb{P}[X_{h_{\ell -1}}\leq x
\,|\, G_{\ell}=g],
\end{equation*}
for $g\in \supp (\mathbb{P}_{G_{\ell}})$, $\ell \geq 1$. Consider the sequence
$(K_{\ell})_{\ell \geq 1}$ of random variables given by~$K_{\ell}:= K_{
\ell}(G_{\ell})$, where
\begin{equation*}
K_{\ell}(g):=\sup _{x\neq y}
\frac{|F_{X_{h_{\ell -1}}|G_{\ell}=g}(x)-F_{X_{h_{\ell -1}}|G_{\ell}=g}(y)|}{|x-y|},
\qquad g\in \supp (\mathbb{P}_{G_{\ell}}), \ell \geq 1.
\end{equation*}
Assume that $(K_{\ell})_{\ell \geq 1}$ satisfies
\begin{equation}
\label{assump:unif:lipschitz:integrability:conditional:cdf}
\sup _{\ell \geq 1}\mathbb{E}[|G_{\ell}|K_{\ell}]<\infty .
\end{equation}
Then
\begin{equation*}
\sup _{\ell \geq 1,\xi \in \mathbb{R}} h_{\ell}^{-\frac{1}{2}}
\mathbb{E}[|\mathds1_{\{X_{h_{\ell}}>\xi \}}-\mathds1_{\{X_{h_{\ell -1}}>
\xi \}}|]<\infty .
\end{equation*}
\end{enumerate}
\end{proposition}

\begin{proof}
See Appendix~\ref{prf:local:strong:error:indicator:func}.
\end{proof}

\begin{remark}
\label{rem3.3}
The scenarios formulated in Proposition~\ref{prop:local:strong:error:indicator:func}
are ordered by strength. The scenario
\eqref{assumption:finite:Lp:moment} is already described in Barrera et
al.~\cite[Assumption~H7]{barrera:hal-01710394}, although for~$p_{
\star}>2$; our formulation also includes heavy-tailed setups where
$p_{\star}<2$. A~variant of the Gaussian concentration scenario
\eqref{assumption:conditional:gaussian:concentration} is stated in~\cite[Assumption~H6]{barrera:hal-01710394}.
According to Frikha and Menozzi~\cite[Sect.~1]{10.1214/ECP.v17-1952}, it
suffices that $\mathbb{E}[\exp (u_{0}\varphi (Y,Z)^{2})\,|\,Y]$ be bounded
for some $u_{0}>0$ for
\eqref{assumption:conditional:gaussian:concentration} to hold. Observe
that \eqref{assumption:finite:Lp:moment} follows from
\eqref{assumption:conditional:gaussian:concentration} for any~$p_{
\star}>1$ via a power series expansion of the exponential. The scenario
\eqref{assump:unif:lipschitz:integrability:conditional:cdf} is a weakening
of Haji-Ali et al.~\cite[Assumption~2.5]{doi:10.1137/21M1447064} (and Gordy
and Juneja~\cite[Assumption~1]{Gordy}). As shown in Sect.~\ref{ssec:mlsa-cost},
it describes a setup that is computationally optimal. Note that if
$(g,\ell )\mapsto F_{X_{h_{\ell -1}}|G_{\ell}=g}$ is uniformly Lipschitz,
then by \eqref{strong:error:diff:Xh} with $h=h_{\ell}$ and
$h'=h_{\ell -1}$,
\eqref{assump:unif:lipschitz:integrability:conditional:cdf} holds.
\end{remark}

\begin{assumption}
\label{asp:fh-f0}
There exist $C,\delta >0$ such that for any $h\in \mathcal{H}$ and any
compact set $\mathcal{K}\subseteq \mathbb{R}$,
\begin{equation*}
\sup _{\xi \in \mathcal{K}}{|f_{X_{h}}(\xi )-f_{X_{0}}(\xi )|}\leq Ch^{
\frac{1}{4}+\delta}.
\end{equation*}
\end{assumption}

\begin{remark}
\label{rem3.5}
Similarly to the second part of Assumption~\ref{asp:Xh->X0}(\ref{asp:Xh->X0-iii}),
Assumption~\ref{asp:fh-f0} is a weakened postulate in comparison with Giorgi
et al.~\cite[Proposition~5.1(a)]{Giorgi2020}.
\end{remark}

We state below our main result regarding the non-asymptotic squared statistical
error of Algorithm~\ref{alg:mlsa}.

\begin{theorem}
\label{thm:ml-variance-cv}
Suppose that $\varphi (Y,Z)\in L^{2}(\mathbb{P})$ and Assumptions~\ref{asp:Xh->X0},
\ref{asp:misc} and~\ref{asp:fh-f0} hold. Then if
$\gamma _{n}=\gamma _{1}n^{-\beta}$, $\beta \in (0,1]$, with
$\bar\lambda _{2}\gamma _{1}>2 $ if $\beta =1$, and if
\begin{equation*}
\sup _{h\in \overline{\mathcal{H}}}{\mathbb{E}\bigg[|\xi ^{h}_{0}|^{4}
\exp \bigg(\frac{4}{1-\alpha}k_{\alpha }\sup _{h'\in
\overline{\mathcal{H}}}\|f_{X_{h'}}\|_{\infty}|\xi _{0}^{h}|\bigg)
\bigg]}+\sup _{h\in \overline{\mathcal{H}}}\mathbb{E}[(\chi _{0}^{h})^{2}]<
\infty ,
\end{equation*}
there exists some constant $\bar{C}<\infty $ such that, for any positive integer
$L$ and any $\mathbf{N}=(N_{0},\dots ,N_{L})\in \mathbb{N}^{L+1}$, we have
\begin{equation}
\label{L2:norm:ML:VaR}
\mathbb{E}[(\xi ^{{\mathrm{ML}}}_{\mathbf{N}}-\xi _{\star}^{h_{L}})^{2}]
\leq \bar{C}\bigg(\gamma _{N_{0}}+\Big(\sum _{\ell =1}^{L}\gamma _{N_{\ell}}
\Big)^{2}+\sum _{\ell =1}^{L}\gamma _{N_{\ell}}^{\frac{3}{2}}+\sum _{
\ell =1}^{L}\gamma _{N_{\ell}}\epsilon (h_{\ell})\bigg)
\end{equation}
and
\begin{eqnarray}
\label{L2:norm:ML:CVaR}
&&\mathbb{E}[(\chi ^{{\mathrm{ML}}}_{\mathbf{N}}-\chi _{\star}^{h_{L}})^{2}]
\nonumber
\\
&&\leq \bar{C}\bigg(\frac{1}{N_{0}^{1\wedge 2\beta}}+\sum _{\ell =1}^{L}
\frac{h_{\ell}}{N_{\ell}}+\sum _{\ell =1}^{L}
\frac{\bar{\gamma}_{N_{\ell}}}{N_{\ell}}\epsilon (h_{\ell})+\sum _{
\ell =1}^{L}\frac{1}{N_{\ell}^{2}}\sum _{k=1}^{N_{\ell}}\gamma _{k}^{
\frac{3}{2}}+\Big(\sum _{\ell =1}^{L}\bar{\gamma}_{N_{\ell}}\Big)^{2}
\bigg),
\end{eqnarray}
where $\bar{\gamma}_{n}=\frac{1}{n}\sum _{k=1}^{n}\gamma _{k}$ and for
$h\in \mathcal{H}$,
\begin{equation}
\label{eq:eps(hl)}
\epsilon (h):=
\begin{cases}
h^{\frac{p_{*}}{2(1+p_{*})}} &\quad
\text{if \eqref{assumption:finite:Lp:moment} is satisfied,}
\\
h^{\frac{1}{2}}|\ln{h}|^{\frac{1}{2}} &\quad
\text{if \eqref{assumption:conditional:gaussian:concentration} is satisfied,}
\\
h^{\frac{1}{2}} &\quad
\text{if \eqref{assump:unif:lipschitz:integrability:conditional:cdf} is
satisfied.}
\end{cases}
\end{equation}
\end{theorem}

\begin{proof}
See Appendix~\ref{prf:thm:ml-variance-cv}.
\end{proof}

\subsection{Complexity analysis}
\label{ssec:mlsa-cost}

Considering the error decompositions \eqref{var:error} and
\eqref{es:error} for $h_{0}\in \mathcal{H}$ and $M\geq 2$, we rely on Theorem~\ref{thm:ml-variance-cv}
to find the optimal number $L\geq 1$ of levels and number
$N_{\ell}\geq 1$ of steps at each level $\ell \in \{0,\dots ,L\}$ of the
MLSA scheme \eqref{eq:ml} in order to achieve a prescribed accuracy
$\varepsilon \in (0,1)$.

\begin{proposition}
\label{prop3.7}
Under Assumption~\ref{asp:misc}(\ref{asp:misc-iii}), let
$\varepsilon \in (0,1)$ be some prescribed accuracy. If
$h_{0}>\varepsilon $, then setting
\begin{equation*}
L=\bigg\lceil \frac{\ln ({h_{0}\varepsilon ^{-1}})}{\ln{M}}
\bigg\rceil
\end{equation*}
achieves a bias error on the estimation of
$(\xi ^{0}_{\star},\chi ^{0}_{\star})$ of the order $\varepsilon $.
\end{proposition}

\begin{proof}
In view of Proposition~\ref{prp:bias-cv}, the bias error of the couple
of estimators
$(\xi ^{{\mathrm{ML}}}_{\mathbf{N}},\chi ^{{\mathrm{ML}}}_{\mathbf{N}})$
is of the order $h_{L}$. We thus select $L\geq 1$ such that
$h_{L}=\frac{h_{0}}{M^{L}}\leq \varepsilon $.
\end{proof}

\begin{lemma}
\label{lem3.8}
The computational cost of MLSA satisfies
\begin{equation*}
\Cost _{{\mathrm{MLSA}}}\leq C\sum _{\ell =0}^{L}
\frac{N_{\ell}}{h_{\ell}}.
\end{equation*}
\end{lemma}

\begin{proof}
This follows directly from the scheme \eqref{eq:ml}.
\end{proof}

\begin{theorem}
\label{thm:cost:mlsa}
Let $\varepsilon \in (0,1)$ be some prescribed accuracy. Then under the
same assumptions as in Theorem~\ref{thm:ml-variance-cv}, we have:
\begin{enumerate}[\rm(i)]
\item
\label{thm:cost:mlsa:var}
\textup{(VaR-focused parametrisation)} Set
\begin{equation*}
N_{\ell}=\bigg\lceil (\bar{C}\gamma _{1})^{\frac{1}{\beta}}\varepsilon ^{-
\frac{2}{\beta}}\bigg(\sum _{\ell '=0}^{L}h_{\ell '}^{-
\frac{\beta}{1+\beta}}\epsilon (h_{\ell '})^{\frac{1}{1+\beta}}\bigg)^{
\frac{1}{\beta }}h_{\ell}^{\frac{1}{1+\beta}}\epsilon (h_{\ell})^{
\frac{1}{1+\beta}}\bigg\rceil , \qquad 0\leq \ell \leq L,
\end{equation*}
i.e.,
\begin{eqnarray*}
N_{\ell }=
\begin{cases}
\lceil (\bar{C}\gamma _{1})^{\frac{1}{\beta}}\varepsilon ^{-\frac{2}{\beta}}h_{
\ell}^{\frac{1}{1+\beta}(1+\frac{p_{*}}{2(1+p_{*})})} (\sum _{\ell '=0}^{L}h_{
\ell '}^{\frac{1}{1+\beta}(-\beta +\frac{p_{*}}{2(1+p_{*})})})^{
\frac{1}{\beta}}\rceil &\!\!\!\!
\text{under \eqref{assumption:finite:Lp:moment},}
\\
\lceil (\bar{C}\gamma _{1})^{\frac{1}{\beta }}\varepsilon ^{-
\frac{2}{\beta}}h_{\ell}^{\frac{3}{2(1+\beta )}}|\ln{h_{\ell}}|^{
\frac{1}{2(1+\beta )}} (\sum _{\ell '=0}^{L}h_{\ell '}^{
\frac{1-2\beta}{2(1+\beta )}}|\ln{h_{\ell '}}|^{\frac{1}{2(1+\beta )}})^{
\frac{1}{\beta}}\rceil &\!\!\!\!
\text{under \eqref{assumption:conditional:gaussian:concentration},}
\\
\lceil (\bar{C}\gamma _{1})^{\frac{1}{\beta}}\varepsilon ^{-\frac{2}{\beta}}h_{
\ell}^{\frac{3}{2(1+\beta )}} (\sum _{\ell '=0}^{L}h_{\ell '}^{
\frac{1-2\beta}{2(1+\beta )}})^{\frac{1}{\beta}}\rceil &\!\!\!\!
\text{under
\eqref{assump:unif:lipschitz:integrability:conditional:cdf},}
\end{cases}
\end{eqnarray*}
where $\bar{C}$ is the constant appearing in \eqref{L2:norm:ML:VaR}. Then MLSA
achieves a statistical error on the estimation of $\xi ^{0}_{\star}$ of
the order $\varepsilon $. The optimal computational cost of MLSA is achieved
when $\beta =1$, under the constraint
$\bar\lambda _{2}\gamma _{1}>2$, in which case
\begin{equation*}
\Cost _{{\mathrm{MLSA}}}^{{\mathrm{VaR}}} \leq C
\begin{cases}
\varepsilon ^{-3 +\frac{p_{*}}{2(1+p_{*})}} &\quad
\text{if \eqref{assumption:finite:Lp:moment} is satisfied,}
\\
\varepsilon ^{-\frac{5}{2}}|\ln{\varepsilon}|^{\frac{1}{2}} &\quad
\text{if \eqref{assumption:conditional:gaussian:concentration} is satisfied,}
\\
\varepsilon ^{-\frac{5}{2}} &\quad
\text{if \eqref{assump:unif:lipschitz:integrability:conditional:cdf} is
satisfied.}
\end{cases}
\end{equation*}

\item
\label{thm:cost:mlsa:es}
\textup{(ES-focused parametrisation)} Under the constraint
$\bar\lambda _{2}\gamma _{1}>2$, set $\beta =1$ and
\begin{equation*}
N_{\ell}=\lceil \bar{C}\varepsilon ^{-2}Lh_{\ell}\rceil
=\bigg\lceil \bar{C}\varepsilon ^{-2}\Big\lceil\frac{\ln ({h_{0}\varepsilon ^{-1}})}{\ln{M}}\Big\rceil h_{\ell}\bigg\rceil ,
\qquad 0\leq \ell \leq L,
\end{equation*}
where $\bar{C}$ is the constant appearing in \eqref{L2:norm:ML:CVaR}.
Then MLSA achieves a statistical error on the estimation of
$\chi ^{0}_{\star}$ of the order $\varepsilon $. The corresponding computational
cost satisfies

\begin{equation*}
\Cost _{{\mathrm{MLSA}}}^{{\mathrm{ES}}}\leq C\varepsilon ^{-2}|\ln{
\varepsilon}|^{2}.
\end{equation*}
\end{enumerate}
\end{theorem}

\begin{proof}
(\ref{thm:cost:mlsa:var})\
Similarly to the multilevel Monte Carlo algorithm
for the computation of the probability $\mathbb{P}[X_{0}>0]$ by Giles and
Haji-Ali~\cite{doi:10.1137/18M1173186}, it is expected that the leading
term in the global $L^{2}$-error \eqref{L2:norm:ML:VaR} is the last one,
namely $\sum _{\ell =1}^{L}\gamma _{N_{\ell}}\epsilon (h_{\ell})$. Following
this heuristic, in order to obtain the optimal values for
$N_{0},\dots ,N_{L}$, we minimise the complexity under a mean-squared-error
constraint, i.e., we solve
\begin{equation*}
\begin{cases}
\minimize &\sum _{\ell =0}^{L}N_{\ell }h_{\ell}^{-1}
\qquad \text{over ${N_{0},\dots ,N_{L}>0}$}
\\[3ex]
\subjectTo &\sum _{\ell =0}^{L}\gamma _{N_{\ell}}
\epsilon (h_{\ell})=\bar{C}^{-1}\varepsilon ^{2},
\end{cases}
\end{equation*}
where $\bar{C}$ is the constant appearing in \eqref{L2:norm:ML:VaR}.
We can easily check that with the above choices of $L$ and $N_{0},\dots ,N_{L}$, the
first terms
$\gamma _{N_{0}}+(\sum _{\ell =1}^{L}\gamma _{N_{\ell}})^{2}+\sum _{
\ell =1}^{L}\gamma _{N_{\ell}}^{\frac{3}{2}}$ in
\eqref{L2:norm:ML:VaR} are $\OO (\varepsilon ^{2})$ as
$\varepsilon \downarrow 0$. The related complexity computations are standard
and hence omitted.
\\

(\ref{thm:cost:mlsa:es})\
To parametrise the amounts
$N_{0},\dots ,N_{L}$ of iterations optimally to compute the~ES, we minimise
the complexity of MLSA while constraining the global $L^{2}$-error of the
multilevel ES estimator \eqref{L2:norm:ML:CVaR} to an order of
$\varepsilon ^{2}$. We presume that the leading term of the upper estimate
\eqref{L2:norm:ML:CVaR} is
$\sum _{\ell =1}^{L}\frac{h_{\ell}}{N_{\ell}}$, which we check a posteriori.
Hence we~solve
\begin{equation*}
\begin{cases}
\minimize &\displaystyle \sum _{\ell =0}^{L}N_{\ell }h_{\ell}^{-1}
\qquad \text{over ${N_{0},\dots ,N_{L}>0}$}
\\[3ex]
\subjectTo &\displaystyle \sum _{\ell =0}^{L} h_{\ell }N_{\ell}^{-1}=\bar{C}^{-1}
\varepsilon ^{2}.
\end{cases}
\end{equation*}

Under this parametrisation, we check that
$(\sum _{\ell =0}^{L}\bar{\gamma}_{N_{\ell}})^{2}=\OO (\varepsilon ^{2
\beta}|\ln{\varepsilon}|^{2(\mathds1_{\{\beta =1\}}-\beta )})$ as~$
\varepsilon \downarrow 0$. Hence in order to achieve a mean-squared error
of the order $\varepsilon ^{2}$, one has to choose $\beta =1$. One can
also verify that the remaining terms of the upper estimate of~\eqref{L2:norm:ML:CVaR},
namely
$\frac{1}{N_{0}}+\sum _{\ell =1}^{L}
\frac{\bar{\gamma}_{N_{\ell}}}{N_{\ell}}\epsilon (h_{\ell})+\sum _{
\ell =1}^{L}\frac{1}{N_{\ell}^{2}}\sum _{k=1}^{N_{\ell}}\gamma _{k}^{
\frac{3}{2}}$, are of the order~$\varepsilon ^{2}$.
\end{proof}

For the VaR estimation, in the best case scenario where
\eqref{assump:unif:lipschitz:integrability:conditional:cdf} is satisfied,
it is possible for a prescribed accuracy $\varepsilon $ to estimate the
VaR with a complexity of $\OO (\varepsilon ^{-\frac{5}{2}})$ by using the
multilevel stochastic approximation approach. This is an order of magnitude
lower than the optimal complexity of $\OO (\varepsilon ^{-3})$ of the nested
stochastic approximation approach obtained in Theorem~\ref{thm:cost:nsa}.

As for the ES estimation, note that the resulting complexity coincides
exactly with the optimal complexity of the standard multilevel Monte Carlo
algorithm derived by Giles~\cite{10.1287/opre.1070.0496}.

\section{Financial case studies}
\label{s:numerical}

We now assess the convergence results of NSA and MLSA numerically. This
is done through stylised financial setups, where the probability measure
$\mathbb{P}$ in the paper is used for all pricing and risk computations.
For comparative and assessment purposes, the numerical settings we are
interested in allow to retrieve the VaR and ES either ana\-lytically or
through Algorithms~\ref{alg:robbins-monro}, \ref{alg:nested-sa} or~\ref{alg:mlsa}.
Recall that real-life scenarios nonetheless lack analytical formulas and
prohibit the recourse to Algorithm~\ref{alg:robbins-monro}. The code producing
the results exhibited next is available at \texttt{github.com/azarlouzi/mlsa}.

\subsection{European option}
\label{ssec:numerical}

We endorse the stylised financial setup of Giles and Haji-Ali~\cite[Sect.~3]{doi:10.1137/18M1173186}
and consider an option with payoff $-W_{T}^{2}$ at maturity $T=1$, where
$W$ is a standard Brownian motion. Assuming zero interest rates, the value
$v(t,y)$ of the option at time $t\in [0,1]$ is given by
\begin{equation*}
v(t,y)=\mathbb{E}[-W_{1}^{2}\,|\,W_{t}=y].
\end{equation*}
Let $\delta \in (0,1)$ be a time horizon. We define the loss $X_{0}$ of
the option by
\begin{equation*}
X_{0}=v(0,0)-v(\delta ,W_{\delta}).
\end{equation*}

Let $\varphi :\mathbb{R}^{2}\to \mathbb{R}$ be given by
\begin{align*}
\varphi (y,z) &:=-(\sqrt{\delta}\, y+\sqrt{1-\delta}\,z)^{2}
\\
&
\hphantom{:}
=-\delta y^{2}-2\sqrt{\delta (1-\delta )}\,yz-(1-\delta )z^{2},
\qquad y,z\in \mathbb{R}.
\end{align*}
Then
\begin{equation}
\label{eq:X0=f(Y,Z)}
X_{0} \stackrel{\mathcal{L}}{=}\mathbb{E}[\varphi ({Y},Z)]-\mathbb{E}[
\varphi (Y,Z)\,|\,Y] =-1-\mathbb{E}[\varphi (Y,Z)\,|\,Y] =\delta (Y^{2}-1),
\end{equation}
where $Y$ and $Z$ are independent and $\sim \mathcal{N}(0,1)$.\\

The VaR $\xi ^{0}_{\star}$ at level $\alpha \in (0,1)$ of the loss
$X_{0}$ can be obtained analytically. It satisfies
\begin{equation*}
1-\alpha =\mathbb{P}[X_{0}>\xi ^{0}_{\star}] =\mathbb{P}\bigg[Y^{2}>1+
\frac{\xi ^{0}_{\star}}{\delta}\bigg]
=2F\bigg(-\Big(1+\frac{\xi ^{0}_{\star}}{\delta}\Big)^{\frac{1}{2}}
\bigg),
\end{equation*}
where $F$ is the standard Gaussian distribution function. Hence
\begin{equation}
\label{eq:theta*}
\xi ^{0}_{\star}=\delta \bigg(\bigg(F^{-1}\Big(\frac{1-\alpha}{2}
\Big)\bigg)^{2}-1\bigg).
\end{equation}
We can also get an analytical formula for the ES $\chi ^{0}_{\star}$ at
level $\alpha $. Indeed, using the symmetry of the Gaussian distribution,
\begin{align*}
\chi ^{0}_{\star }&=\mathbb{E}[X_{0}\,|\,X_{0}>\xi ^{0}_{\star}]
\\
&=\frac{\delta}{1-\alpha}\big(2\mathbb{E}[Y^{2}\mathds1_{\{Y>\mu \}}]-(1-
\alpha )\big), \qquad \text{where } \mu :=\bigg(1+
\frac{\xi ^{0}_{\star}}{\delta}\bigg)^{\frac{1}{2}}.
\end{align*}
Integrating by parts,
\begin{equation*}
\mathbb{E}[Y^{2}\mathds1_{\{Y>\mu \}}]=\mu f(\mu )+F(-\mu ),
\end{equation*}
where $f$ denotes the standard Gaussian density function. Hence
\begin{equation}
\label{eq:c*}
\chi ^{0}_{\star}=\frac{2\delta}{1-\alpha}\bigg(\mu f(\mu )+F(-\mu )-
\frac{1-\alpha}{2}\bigg).
\end{equation}

The loss $X_{0}\stackrel{\mathcal{L}}{=}\delta (Y^{2}-1)$ (cf.~\eqref{eq:X0=f(Y,Z)})
can be simulated directly so that Algorithm~\ref{alg:robbins-monro} is
applicable. But the goal here is to estimate the values of
$\xi ^{0}_{\star}$ and $\chi ^{0}_{\star}$ using Algorithms~\ref{alg:nested-sa}
and~\ref{alg:mlsa}, and to study empirically the influence of the parametrisation
of each algorithm on their performances. The exact values obtained with
\eqref{eq:theta*} and \eqref{eq:c*} serve as convergence benchmarks.
\\

The loss
$X_{0}\stackrel{\mathcal{L}}{=}-1-\mathbb{E}[\varphi (Y,Z)\,|\,Y]$ can
alternatively be approximated, for a given bias parameter
$h=\frac{1}{K}\in \mathcal{H}$, by
\begin{equation*}
X_{h}=-1-\frac{1}{K}\sum _{k=1}^{K}\varphi (Y,Z^{(k)}),
\end{equation*}
where
$\text{$Y,Z^{(1)},\dots ,Z^{(K)}$ are~{i.i.d.}} \sim \mathcal{N}(0,1)$.
We can then apply Algorithm~\ref{alg:nested-sa} or~\ref{alg:mlsa} on this
basis.

\subsubsection{Numerical results}
\label{sec4.1.1}

In the following applications, the confidence level is taken as
$\alpha =97.5\%$, and the time horizon is set to $\delta =0.5$. This setup
yields $\xi ^{0}_{\star}\approx 2.012$ and
$\chi ^{0}_{\star}\approx 2.901$ (computed by using the explicit formulas
\eqref{eq:theta*} and \eqref{eq:c*}). The value of $\beta =1$ is used for
the step sequences $(\gamma _{n})_{n\geq 1}$ since according to Theorems~\ref{thm:cost:nsa}
and \ref{thm:cost:mlsa}, it leads to the optimal complexities for Algorithms~\ref{alg:nested-sa}
and~\ref{alg:mlsa}.
\\

To showcase the results of Proposition~\ref{prp:bias-cv}, we study the
linearity of the errors $\xi ^{h}_{\star}-\xi ^{0}_{\star}$ and
$\chi ^{h}_{\star}-\chi ^{0}_{\star}$ and the stability of the renormalised
errors $h^{-1}(\xi ^{h}_{\star}-\xi ^{0}_{\star})$ and
$h^{-1}(\chi ^{h}_{\star}-\chi ^{0}_{\star})$ as
$\mathcal{H}\ni h\downarrow 0$. We run Algorithm~\ref{alg:nested-sa} for
multiple values of \text{${h\in \{\frac{1}{10},\frac{1}{20},
\frac{1}{50},\frac{1}{100},\frac{1}{200}\}}$} and for a very large number
$N=10^{6}$ of iterations. We adopt the step size
$\gamma _{n}=0.1/(10^{4}+n)$, with a smoothing applied to the denominator
to avoid instability at early iterations. We average $200$ outcomes
$(\xi ^{h}_{10^{6}},\chi ^{h}_{10^{6}})$ to form an estimate
$(\bar\xi ^{h}_{10^{6}},\bar\chi ^{h}_{10^{6}})$ of
$(\xi ^{h}_{\star},\chi ^{h}_{\star})$. The left panel of Fig.~\ref{fig:approx-h}
shows the evolution of the quantities
$\bar\xi ^{h}_{10^{6}}-\xi ^{0}_{\star}$ and
$\bar\chi ^{h}_{10^{6}}-\chi ^{0}_{\star}$ as $h$ runs through
$\{\frac{1}{10},\frac{1}{20},\frac{1}{50},\frac{1}{100},\frac{1}{200}
\}$, while the right panel shows the evolution of the corresponding quantities
$h^{-1}(\bar\xi ^{h}_{10^{6}}-\xi ^{0}_{\star})$ and
$h^{-1}(\bar\chi ^{h}_{10^{6}}-\chi ^{0}_{\star})$.

\begin{figure}[H]
\centerline{
\begin{subfigure}{.5\linewidth}
%\centering
\flushleft
\includegraphics[width=.975\linewidth]{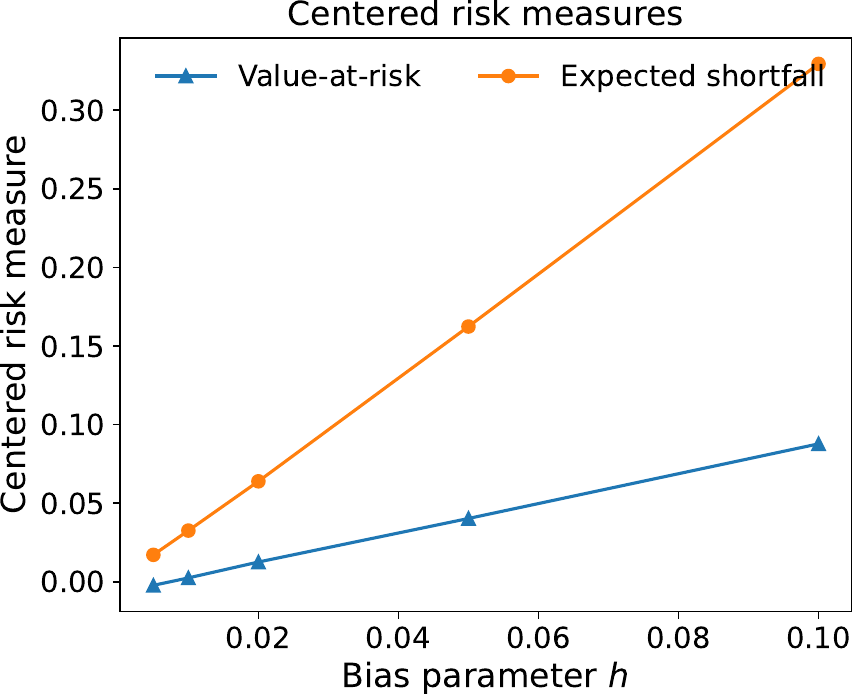}
\end{subfigure}%
\begin{subfigure}{.5\linewidth}
%\centering
\flushright
\includegraphics[width=.975\linewidth]{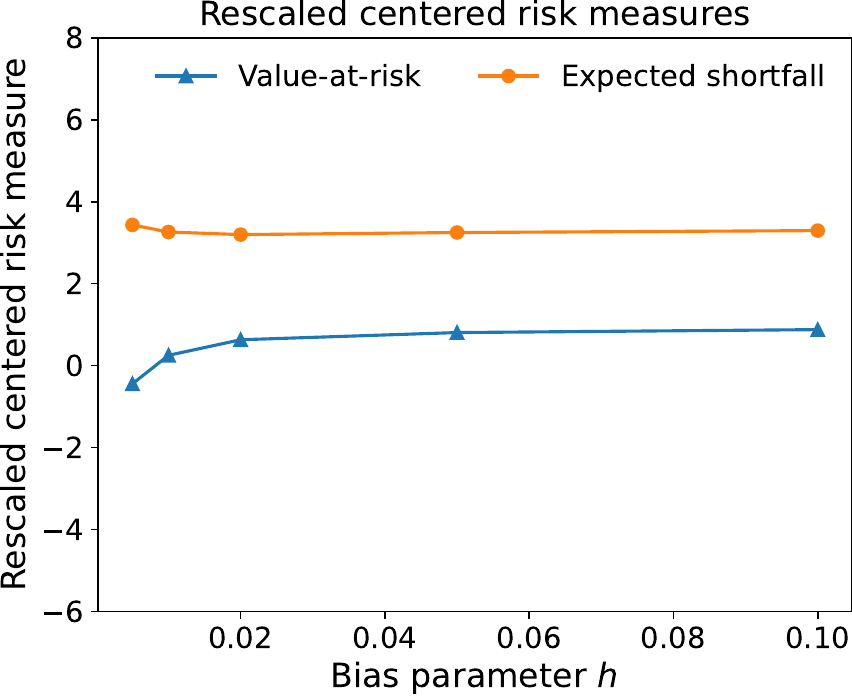}
\end{subfigure}
}
\caption{Centered and rescaled risk measures for a bias parameter $h$ tending to $0$.}
\label{fig:approx-h}
\end{figure}

The left-panel plot suggests that asymptotically as
$\mathcal{H}\ni h\downarrow 0$, the quantities
$\xi ^{h}_{\star}-\xi ^{0}_{\star}$ and
$\chi ^{h}_{\star}-\chi ^{0}_{\star}$ are linear in $h$. The right-panel
plot strengthens this observation, demonstrating that the quantities
$h^{-1}(\xi ^{h}_{\star}-\xi ^{0}_{\star})$ and
$h^{-1}(\chi ^{h}_{\star}-\chi ^{0}_{\star})$ are approximately constant
in a neighborhood of $0$. This illustrates empirically the validity of
Proposition~\ref{prp:bias-cv} on the weak error expansion.\\

We next aim to compare the performances of Algorithms~\ref{alg:robbins-monro},
\ref{alg:nested-sa} and~\ref{alg:mlsa}. For Algorithms~\ref{alg:robbins-monro}
and~\ref{alg:nested-sa}, we set the step size
$\gamma _{n}=1/(100+n)$ for the VaR simulation and
$\gamma _{n}=0.1/(2.5\times 10^{4}+n)$ for the ES simulation. We work under
the scenario~\eqref{assumption:finite:Lp:moment} for Algorithm~\ref{alg:mlsa},
that clearly holds for any $p_{\star}>1$, with the particular value
$p_{\star}=11$. We choose $M=2$. The VaR and ES are computed using their
respective optimal iteration amounts described in Theorem~\ref{thm:cost:mlsa}.
The remaining hyperparameter setup is provided in Table~\ref{tbl:option}.
The $(\gamma _{n})_{n\geq 1}$ given there for the VaR estimation were obtained
through a full grid search on their parameters.

For each algorithm and each risk measure, we plot on a logarithmic scale
the average run times over $200$ runs against the achieved RMSEs, for a
prescribed accuracy $\varepsilon $ that loops through
$\{\frac{1}{32},\frac{1}{64},\frac{1}{128},\frac{1}{256},
\frac{1}{512}\}$. We also plot the average run times against the prescribed
accuracies $\varepsilon $ themselves; see Fig.~\ref{fig:comparison}. Table~\ref{tbl:option:slopes}
reports the slopes fitted on these curves.

\begin{table}[H]
\tabcolsep=1em
\caption{Algorithm~\ref{alg:mlsa} parametrisation by prescribed accuracy}
\label{tbl:option}
\centerline{
\begin{tabular}{c|ccc|ccc}
\hline
\multirow{2}{20mm}{\centering Prescribed accuracy $\varepsilon$}
& \multicolumn{3}{c|}{VaR estimation}
& \multicolumn{3}{c}{ES estimation} \\
\cline{2-7}
& $h_0$ & $L$ & $\gamma_n$
& $h_0$ & $L$ & $\gamma_n$ \\
\specialrule{0.1em}{0em}{0em}
$\frac{1}{32}$
& $\frac{1}{16}$ & $1$ & $\frac{2}{2.5\times10^3+n}$
& $\frac{1}{16}$ & $1$ & $\frac{0.1}{10^4+n}$ \\
\hline
$\frac{1}{64}$
& $\frac{1}{32}$ & $1$ & $\frac{2}{4\times10^3+n}$
& $\frac{1}{32}$ & $1$ & $\frac{0.1}{10^4+n}$ \\
\hline
$\frac{1}{128}$
& $\frac{1}{32}$ & $2$ & $\frac{0.75}{9\times10^3+n}$
& $\frac{1}{32}$ & $2$ & $\frac{0.1}{10^4+n}$ \\
\hline
$\frac{1}{256}$
& $\frac{1}{32}$ & $3$ & $\frac{0.25}{10^4+n}$
& $\frac{1}{32}$ & $3$ & $\frac{0.1}{2\times10^4+n}$ \\
\hline
$\frac{1}{512}$
& $\frac{1}{32}$ & $4$ & $\frac{0.09}{10^4+n}$
& $\frac{1}{32}$ & $4$ & $\frac{0.1}{2.5\times10^4+n}$ \\
\hline
\end{tabular}
}
\end{table}

\begin{figure}[H]
\centerline{
\begin{subfigure}{.5\linewidth}
%\centering
\flushleft
\includegraphics[width=.975\textwidth]{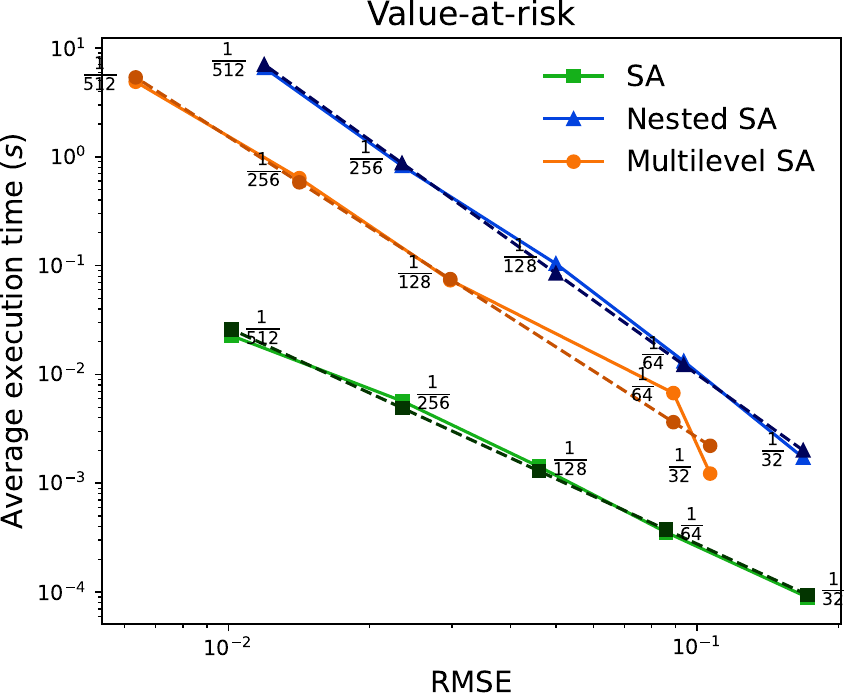}
\end{subfigure}%
\begin{subfigure}{.5\linewidth}
%\centering
%\hspace{-0.5cm}
\flushright
\includegraphics[width=.975\textwidth]{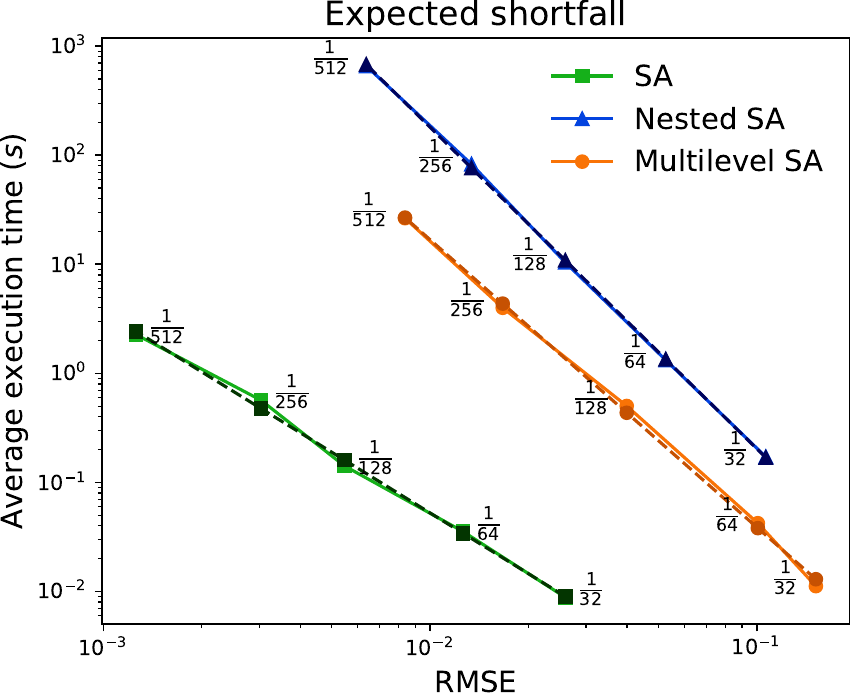}
\end{subfigure}
}
\medskip
\centerline{
\begin{subfigure}{.5\linewidth}
%\centering
\flushleft
\includegraphics[width=.975\textwidth]{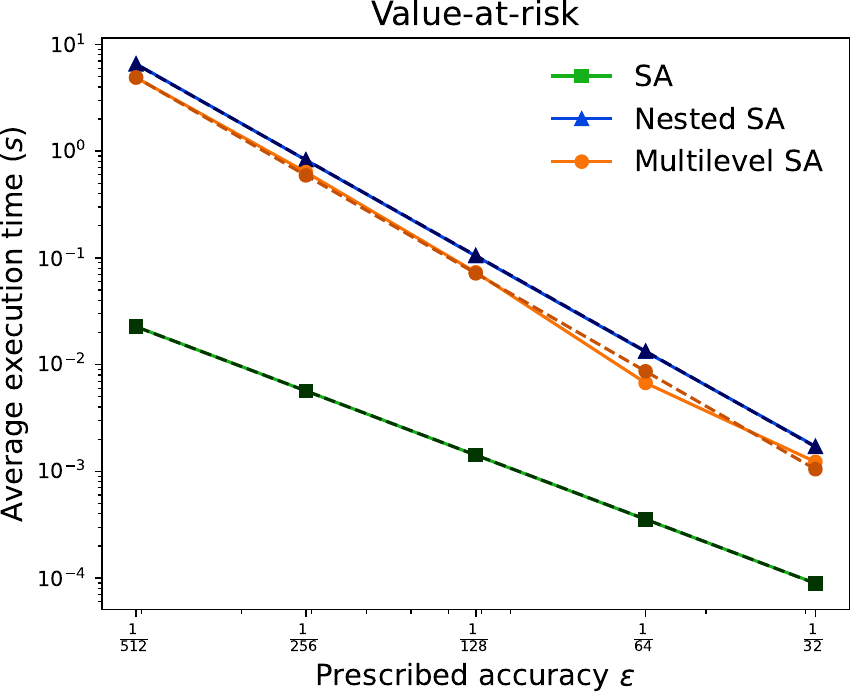}
\end{subfigure}%
\begin{subfigure}{.5\linewidth}
%\centering
\flushright
\includegraphics[width=.975\textwidth]{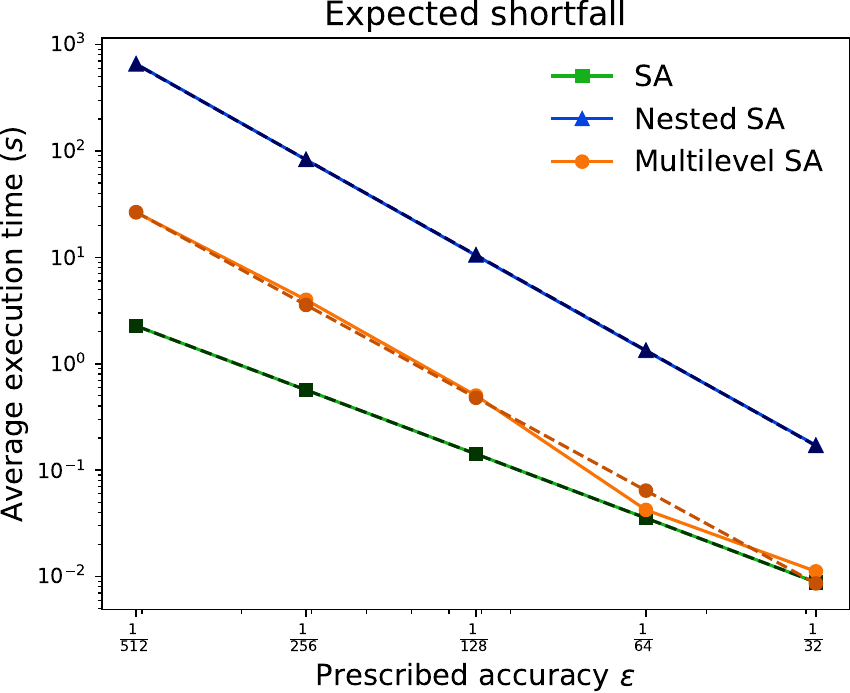}
\end{subfigure}
}
\caption{Performance comparison of Algorithms~\ref{alg:robbins-monro},~\ref{alg:nested-sa} and~\ref{alg:mlsa}.}
\label{fig:comparison}
\end{figure}

\begin{table}[H]
\tabcolsep=1em
\caption{Reported slopes in Fig.~\ref{fig:comparison}}
\label{tbl:option:slopes}
\centerline{
\begin{tabular}{l|cc|cc}
\hline
\multirow{2}{20mm}{SA scheme}
& \multicolumn{2}{c|}{VaR estimation}
& \multicolumn{2}{c}{ES estimation} \\
\cline{2-5}
& RMSE
& Accuracy
& RMSE
& Accuracy \\
\specialrule{0.1em}{0em}{0em}
Nested SA (Alg.~\ref{alg:nested-sa}) & ${-3.08}$ & ${-2.98}$ & ${-2.95}$ & ${-2.98}$ \\\hline
Multilevel SA (Alg.~\ref{alg:mlsa}) & ${-2.76}$ & ${-3.05}$ & ${-2.64}$ & ${-2.90}$ \\\hline
SA (Alg.~\ref{alg:robbins-monro}) & ${-1.99}$ & ${-2.00}$ & ${-1.85}$ & ${-2.00}$ \\
\hline
\end{tabular}
}
\end{table}

For the VaR as for the ES, for any given scored RMSE, Algorithm~\ref{alg:mlsa}
displays an execution time that is orders of magnitude lower than Algorithm~\ref{alg:nested-sa}.
Take an RMSE of the order $5\times 10^{-2}$. MLSA computes the VaR in
$10^{-2}s$, scoring a $10$-fold speedup over NSA which computes it in
$10^{-1}s$. Likewise, to compute the ES, MLSA requires~$2\times 10^{-1}s$,
a $10$-fold speedup over NSA which requires $2s$. This corroborates the
theoretical optimal complexity of Algorithm~\ref{alg:mlsa} that is significantly
lower than that of Algorithm~\ref{alg:nested-sa} (cf.~Theorems~\ref{thm:cost:nsa}
and~\ref{thm:cost:mlsa}). Finally, as one would expect, Algorithm~\ref{alg:robbins-monro}
that simulates $X_{0}$ directly considerably outperforms both Algorithms~\ref{alg:nested-sa}
and~\ref{alg:mlsa}. The spacing between the SA and NSA curves illustrates
the performance gap occurring when using, in a Robbins--Monro instance,
a nested simulation instead of a direct simulation for the loss. The placement
of the MLSA curve in between these two demonstrates the partial yet significant
regain on this performance gap, which is achieved by shifting from the
nested to the multilevel scheme. Note, however, that the multilevel VaR
estimation turned out to be rather unstable numerically. Small changes
to the step-size parametrisation may substantially affect the results.
To counter this, a thorough grid search on the parameters of the step size
was performed, adding a nonnegligible tuning phase to the total execution
time. The multilevel ES estimation, by contrast, is quite robust and requires
only minor hyperparameter tuning. Up to a variation range, the slopes of
the RMSE curves in Table~\ref{tbl:option:slopes} reflect the theoretical
exponents in the complexity theorems. The curves of the average execution
time as a function of the prescribed accuracy feature a log-linear behaviour,
whose slopes in Table~\ref{tbl:option:slopes} roughly match the theoretical
complexity exponents in Theorems~\ref{thm:cost:nsa} and~\ref{thm:cost:mlsa}.

\subsection{Swap on a rate}
\label{s:swap}

We adapt the setup of Albanese et al.~\cite[Sect.~A.1]{AlbaneseArmentiCrepey+2020+25+53}.
Consider a stylised swap of nominal~$\bar{N}$, strike $\bar{S}$ and maturity
$\bar{T}$ on some underlying (FX or interest) rate. This rate follows a
Black--Scholes model $(S_{t})_{t\geq 0}$ of risk neutral drift
$\bar\kappa $ and constant volatility~$\bar\sigma $, or in other words,
$(\widehat{S}_{t}:=\mathrm{e} ^{-\bar\kappa t}S_{t})_{t\geq 0}$ follows a Black
martingale model of volatility~$\bar\sigma $. Given coupon dates
$0<T_{1}<\cdots <T_{d}=\bar{T}$, the swap pays at $ T_{i} $ the cash flow
$\Delta _{i}(S_{T_{i-1 }}-\bar{S})$, where
$\Delta _{i}=T_{i}-T_{i-1}$. For $t\in [0,\bar{T}]$, let
$\rho _{t}=e^{-rt}$ for some constant risk-free rate $r$, and denote by
$i_{t}\in \{1,\dots ,d+1\}$ the integer that satisfies
$T_{i_{t}-1}\leq t<T_{i_{t}}$. The value $P_{t}$ of the swap at time
$t\in [0,\bar{T}]$ is thus expressed as
\begin{align}
P_{t} &=\bar{N}\mathbb E\bigg[\rho _{t}^{-1}\rho _{T_{i_{t}}}\Delta _{i_{t}}(S_{T_{i_{t}-1}}-
\bar{S}) +\sum _{i=i_{t}+1}^{d}\rho _{t}^{-1}\rho _{T_{i}}\Delta _{i}(S_{T_{i-1}}-
\bar{S})\,\bigg|\,S_{t}\bigg]
\notag
\\
&=\bar{N}\mathbb E\bigg[\rho _{t}^{-1}\rho _{T_{i_{t}}}\Delta _{i_{t}}(
\mathrm{e} ^{\bar\kappa T_{i_{t}-1}}\widehat{S}_{T_{i_{t}-1}}-\bar{S})
\nonumber
\\
&
\hphantom{=:\bar{N}\mathbb E\bigg[}
+\sum _{i=i_{t}+1}^{d}\rho _{t}^{-1}\rho _{T_{i}}\Delta _{i}(\mathrm{e} ^{
\bar\kappa T_{i-1}}\widehat{S}_{T_{i-1}}-\bar{S})\,\bigg|\,
\widehat{S}_{t}\bigg]
\label{eq:MtM-cond}
\\
&=\bar{N}\bigg(\rho _{t}^{-1}\rho _{T_{i_{t}}}\Delta _{i_{t}}(\mathrm{e} ^{
\bar\kappa T_{i_{t}-1}}\widehat{S}_{T_{i_{t}-1}}-\bar{S}) +\sum _{i=i_{t}+1}^{d}
\rho _{t}^{-1}\rho _{T_{i}}\Delta _{i}(\mathrm{e} ^{\bar\kappa T_{i-1}}
\widehat{S}_{t}-\bar{S})\bigg).
\label{eq:MtM}
\end{align}
We assume the swap is issued at par at time $0$, i.e., $P_{0}=0$. Hence
we get
\begin{equation}
\label{eq:strike}
\bar{S}=
\frac{\sum _{i=1}^{d}\rho _{T_{i}}\Delta _{i}\mathrm{e} ^{\bar\kappa T_{i-1}}}{\sum _{i=1}^{d}\rho _{T_{i}}\Delta _{i}}S_{0}.
\end{equation}
Given some confidence level $\alpha \in (0,1)$, we are interested in computing
the risk measures
$\xi ^{0}_{\star}:=\VaR _{\alpha}(\rho _{\delta }P_{\delta})$ and
$\chi ^{0}_{\star}:=\ES _{\alpha}(\rho _{\delta }P_{\delta})$ of a short
position on the swap at some time horizon
$\delta <\min _{i}\Delta _{i}$.
\\

In this lognormal setup, the values of $\xi ^{0}_{\star}$ and
$\chi ^{0}_{\star}$ can be obtained analytically. On the one hand, observing
that $i_{\delta}=1$ and using \eqref{eq:MtM} and \eqref{eq:strike},
\begin{equation}
\label{eq:bMtM}
\rho _{\delta }P_{\delta}=\bar{N}A (\widehat{S}_{\delta}-S_{0} ),
\qquad \text{where } A:=\sum _{i=2}^{d}\rho _{T_{i}}\Delta _{i}\mathrm{e} ^{
\bar\kappa T_{i-1}}.
\end{equation}
Hence
\begin{equation*}
1-\alpha =\mathbb{P}[\rho _{\delta }P_{\delta}>\xi ^{0}_{\star}] =
\mathbb{P}\bigg[\widehat{S}_{\delta}>{S_{0}}+
\frac{\xi ^{0}_{\star}}{\bar{N}A}\bigg],
\end{equation*}
and thus
\begin{equation}
\label{eq:swap:xi}
\xi ^{0}_{\star}=\bar{N}AS_{0}\bigg(\exp \Big(F^{-1}(\alpha )
\bar\sigma \sqrt \delta -\frac{\bar\sigma ^{2}}{2}\delta \Big)-1
\bigg),
\end{equation}
where $F$ is the standard Gaussian distribution function. On the other
hand, with
\begin{equation*}
\omega := S_{0}+\frac{\xi ^{0}_{\star}}{\bar{N}A}, \quad \quad \eta _{
\pm}:=\frac{1}{\bar\sigma \sqrt \delta}\bigg(\ln{\frac{\omega}{S_{0}}}
\pm \frac{\bar\sigma ^{2}\delta}{2}\bigg),
\end{equation*}
we obtain
\begin{equation*}
\chi ^{0}_{\star }=\mathbb{E}[\rho _{\delta }P_{\delta }\,|\, \rho _{
\delta }P_{\delta}>\xi ^{0}_{\star}] =\frac{\bar{N}A}{1-\alpha}\big(
\mathbb{E}[\widehat{S}_{\delta}\mathds1_{\{\widehat{S}_{\delta}>
\omega \}}]-(1-\alpha )S_{0}\big),
\end{equation*}
and hence
\begin{equation*}
\mathbb{E}[\widehat{S}_{\delta}\mathds1_{\{\widehat{S}_{\delta}>
\omega \}}] =S_{0}\mathbb{E}\bigg[\exp \bigg(\bar\sigma \sqrt \delta U-
\frac{\bar\sigma ^{2}}{2}\delta \bigg)\mathds1_{\{U>\eta _{+}\}}
\bigg] =S_{0}\big(1-F(\eta _{-})\big),
\end{equation*}
where $U\sim \mathcal{N}(0,1)$. Therefore
\begin{equation}
\label{eq:swap:C}
\chi ^{0}_{\star}=\bar{N}AS_{0}\frac{\alpha -F(\eta _{-})}{1-\alpha}.
\end{equation}

Define $X_{0}:=\rho _{\delta }P_{\delta}$. Algorithm~\ref{alg:robbins-monro}
can be applied to approximate the values of~$\xi ^{0}_{\star}$ and
$\chi ^{0}_{\star}$. Indeed, according to \eqref{eq:bMtM},
\begin{equation*}
X_{0}\stackrel{\mathcal{L}}{=}\bar{N} AS_{0}\bigg(\exp \Big(-
\frac{\bar\sigma ^{2}}{2}\delta +\bar\sigma \sqrt \delta U\Big)-1
\bigg),
\end{equation*}
where $U\sim \mathcal{N}(0,1)$. Algorithms~\ref{alg:nested-sa} and~\ref{alg:mlsa}
are also applicable to approximate the values of~$\xi ^{0}_{\star}$ and
$\chi ^{0}_{\star}$. By \eqref{eq:MtM-cond} and \eqref{eq:strike},
\begin{equation}
\label{eq:swap:nested}
X_{0} =\bar{N}\mathbb{E}\bigg[\sum _{i=2}^{d}\rho _{T_{i}}\Delta _{i}
\mathrm{e} ^{\bar\kappa T_{i-1}}(\widehat{S}_{T_{i-1}}-S_{0})\,\bigg|\,
\widehat{S}_{\delta}\bigg] \stackrel{\mathcal{L}}{=}\mathbb{E}[
\varphi (Y,Z)\,|\,Y],
\end{equation}
where $Y\in \mathbb{R}$ and
$Z=(Z_{1},\dots ,Z_{d-1})\in \mathbb{R}^{d-1}$ are independent, with
\begin{equation*}
\begin{aligned}
Y &:=\exp \bigg(-\frac{\bar\sigma ^{2}}{2}\delta +\bar\sigma \sqrt
\delta U_{0}\bigg) \stackrel{\mathcal{L}}{=}
\frac{\widehat{S}_{\delta}}{S_{0}},
\\
Z_{1} &:=\exp \bigg(-\frac{\bar\sigma ^{2}}{2}(T_{1}-\delta )+
\bar\sigma \sqrt{T_{1}-\delta}U_{1}\bigg) \stackrel{\mathcal{L}}{=}
\frac{\widehat{S}_{T_{1}}}{\widehat{S}_{\delta}},
\\
Z_{i} &:=\exp \bigg(-\frac{\bar\sigma ^{2}}{2}\Delta _{i}+\bar\sigma
\sqrt{\Delta _{i}}U_{i}\bigg) \stackrel{\mathcal{L}}{=}
\frac{\widehat{S}_{T_{i}}}{\widehat{S}_{T_{i-1}}}, \qquad 2\leq i
\leq d-1,
\end{aligned}
\end{equation*}
and
\begin{equation*}
\varphi (y,z) :=\bar{N}S_{0}\sum _{i=2}^{d}\rho _{T_{i}}\Delta _{i}
\mathrm{e} ^{\bar\kappa T_{i-1}}\bigg(y\prod _{j=1}^{i-1}z_{j}-1\bigg),
\end{equation*}
for $y\in \mathbb{R}$,
$z=(z_{1},\dots ,z_{d-1})\in \mathbb{R}^{d-1}$ and
$(U_{i})_{0\leq i\leq d-1}\text{ i.i.d.} \sim \mathcal{N}(0,1)$. For any
bias parameter $h=\frac{1}{K}\in \mathcal{H}$, $X_{0}$ can be approximated
by using \eqref{eq:swap:nested}.

\subsubsection{Numerical results}
\label{sec4.2.1}

We set $r=2\%$, $S_{0}=1\%$, $\bar\kappa =12\%$, $\bar\sigma =20\%$,
$\Delta _{i}=3\,\mathrm{months}$, \text{${\bar{T}=1\,\mathrm{year}}$,}
\text{${\delta =1\,\mathrm{week}}$} and $\alpha =85\%$. We use the
$30/360$ day count convention, meaning that
$1\,\mathrm{month}=30\,\mathrm{days}$ and
$1\,\mathrm{year}=360\,\mathrm{days}$. Finally, we set the nominal
$\bar{N}$ such that each leg of the swap is valued at $1$ at time
$0$, that is,
\begin{equation*}
\bar{N}=
\frac{1}{S_{0}\sum _{i=1}^{d}\rho _{T_{i}}\Delta _{i}\mathrm{e} ^{\bar\kappa T_{i-1}}}.
\end{equation*}
Given the above parameter set, the theoretical values of the VaR and ES
obtained with \eqref{eq:swap:xi} and \eqref{eq:swap:C} are
$\xi ^{0}_{\star}\approx 219.64$ and
$\chi ^{0}_{\star}\approx 333.91$.

For the stochastic approximation of these quantities, we use the step sizes
\text{${\gamma _{n}=100/n}$} for the SA scheme and $\gamma _{n}=50/n$ for
the nested SA scheme. We run these schemes with their respective optimal
iteration amounts. For the multilevel SA scheme, we work under the scenario
\eqref{assumption:finite:Lp:moment} with $p_{\star}=4$, which is tuned
by using a grid search. The VaR and ES are simulated by using their respective
optimal amounts $N_{0},\dots ,N_{L}$ of iterations. For each prescribed
accuracy
$\varepsilon \in \{\frac{1}{32},\frac{1}{64},\frac{1}{128},
\frac{1}{256},\frac{1}{512}\}$, we perform a grid search to tune the initial
bias parameter $h_{0}$ (giving the number~$L$ of levels) and the step sizes
$(\gamma _{n})_{n\geq 1}$. Table~\ref{tbl:swap:mlsa} lists the ensuing
parametrisations by prescribed accuracy. Figure~\ref{fig:swap:comparison}
plots on a logarithmic scale the joint evolution of the realised RMSE and
the average execution time over $200$ runs of each SA scheme, for an accuracy
$\varepsilon $ varying in
$\{\frac{1}{32},\frac{1}{64},\frac{1}{128},\frac{1}{256},
\frac{1}{512}\}$. The second row of Fig.~\ref{fig:swap:comparison} displays
the realised RMSE as a function of the prescribed accuracy. Table~\ref{tbl:swap:mlsa:slopes}
reports the regressed slopes on these curves as depicted in dashed lines
in Fig.~\ref{fig:swap:comparison}.

\begin{table}[H]
\tabcolsep=1em
\caption{Algorithm~\ref{alg:mlsa} parametrisation by prescribed accuracy}
\label{tbl:swap:mlsa}
\centerline{
\begin{tabular}{c|ccc|ccc}
\hline
\multirow{2}{20mm}{\centering Prescibed accuracy $\varepsilon$}
& \multicolumn{3}{c|}{VaR estimation}
& \multicolumn{3}{c}{ES estimation} \\
\cline{2-7}
& $h_0$ & $L$ & $\gamma_n$ 
& $h_0$ & $L$ & $\gamma_n$ \\
\specialrule{0.1em}{0em}{0em}
$\frac{1}{32}$
& $\frac{1}{8}$ & $2$ & $\frac{6}{10+n}$
& $\frac{1}{8}$ & $2$ & $\frac{5}{10+n}$ \\
\hline
$\frac{1}{64}$
& $\frac{1}{16}$ & $2$ & $\frac{20}{500+n}$
& $\frac{1}{16}$ & $2$ & $\frac{20}{500+n}$ \\
\hline
$\frac{1}{128}$
& $\frac{1}{16}$ & $3$ & $\frac{21}{10^3+n}$
 & $\frac{1}{16}$ & $3$ & $\frac{20}{500+n}$ \\
\hline
$\frac{1}{256}$
& $\frac{1}{16}$ & $4$ & $\frac{20}{2\times10^3+n}$
& $\frac{1}{16}$ & $4$ & $\frac{20}{750+n}$ \\
\hline
$\frac{1}{512}$
& $\frac{1}{16}$ & $5$ & $\frac{21}{3\times10^3+n}$
& $\frac{1}{32}$ & $4$ & $\frac{50}{2\times10^3+n}$ \\
\hline
\end{tabular}
}
\end{table}

\begin{figure}[H]
\centerline{
\begin{subfigure}{.5\linewidth}
%\centering
\flushleft
\includegraphics[width=.975\textwidth]{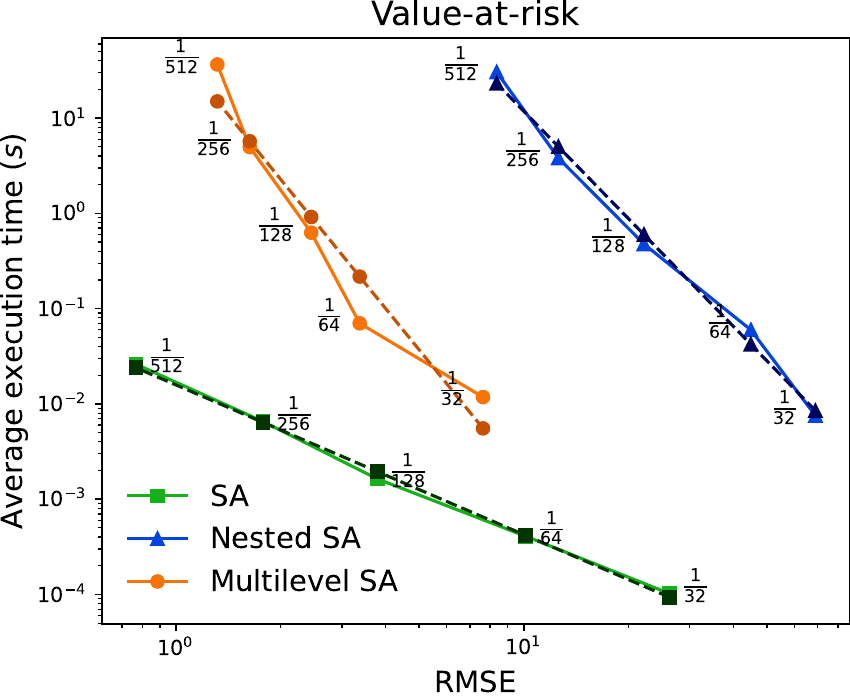}
\end{subfigure}%
\begin{subfigure}{.5\linewidth}
%\centering
\flushright
\includegraphics[width=.975\textwidth]{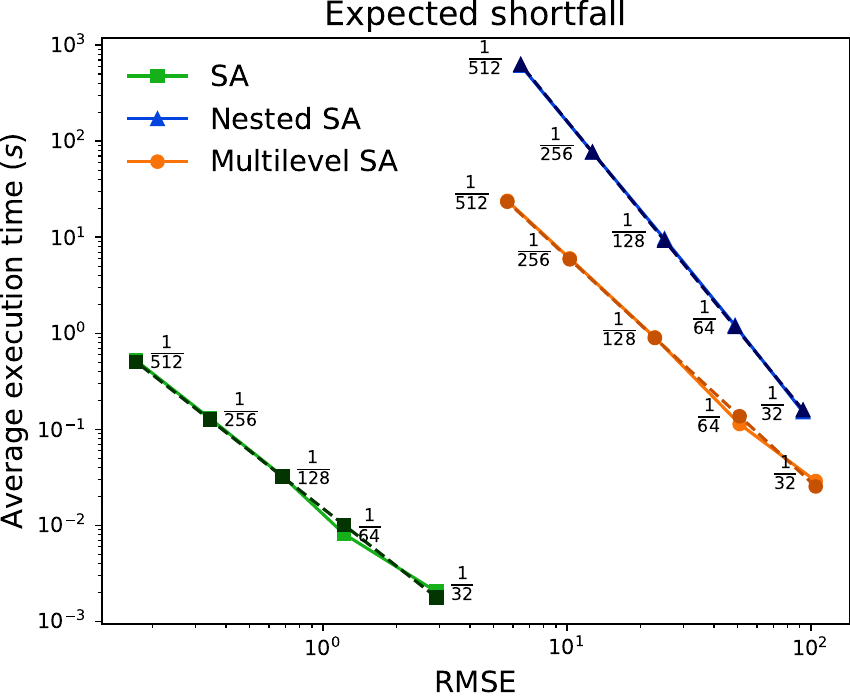}
\end{subfigure}
}
\medskip
\centerline{
\begin{subfigure}{.5\linewidth}
%\centering
\flushleft
\includegraphics[width=.975\textwidth]{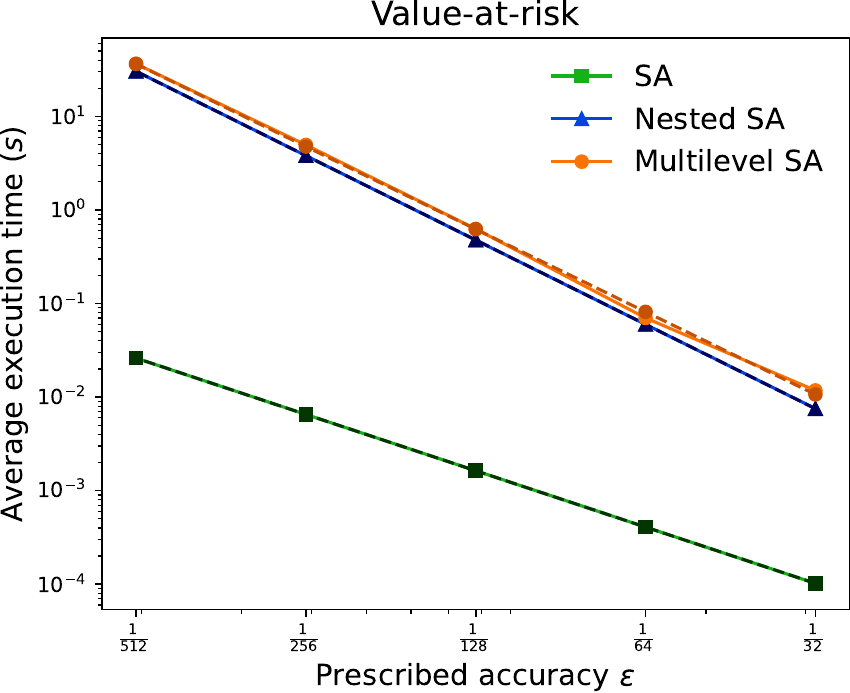}
\end{subfigure}%
\begin{subfigure}{.5\linewidth}
%\centering
\flushright
\includegraphics[width=.975\textwidth]{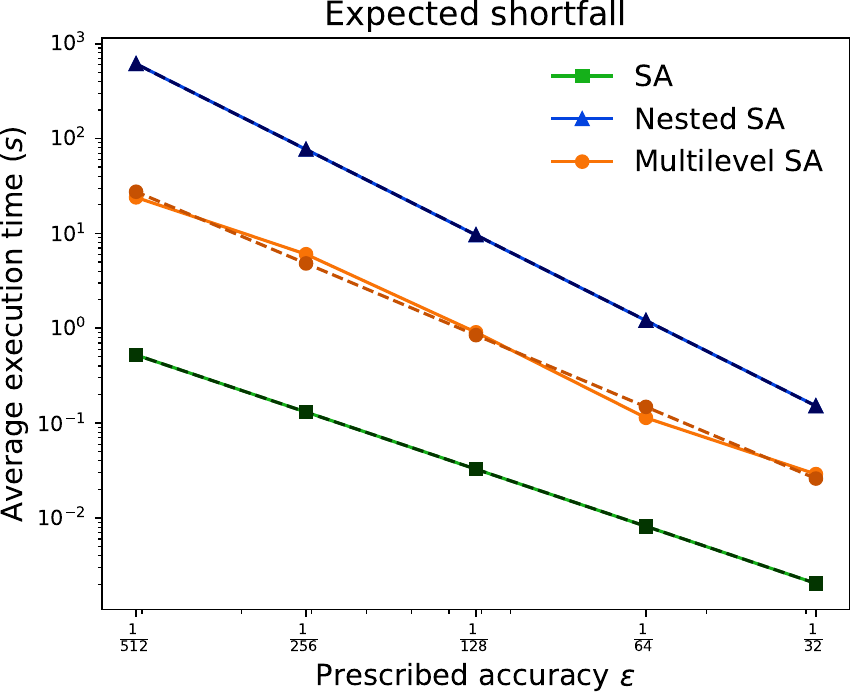}
\end{subfigure}
}
\caption{Performance comparison of Algorithms~\ref{alg:robbins-monro},~\ref{alg:nested-sa} and~\ref{alg:mlsa}.}
\label{fig:swap:comparison}
\end{figure}

\begin{table}[H]
\tabcolsep=1em
\caption{Reported slopes in Fig.~\ref{fig:swap:comparison}}
\label{tbl:swap:mlsa:slopes}
\centerline{
\begin{tabular}{l|cc|cc}
\hline
\multirow{2}{20mm}{SA scheme}
& \multicolumn{2}{c|}{VaR estimation}
& \multicolumn{2}{c}{ES estimation} \\
\cline{2-5}
& RMSE
& Accuracy
& RMSE
& Accuracy \\
\specialrule{0.1em}{0em}{0em}
Nested SA (Alg.~\ref{alg:nested-sa}) & ${-3.75}$ & ${-3.00}$ & ${-3.11}$ & ${-3.00}$ \\
\hline
Multilevel SA (Alg.~\ref{alg:mlsa}) & ${-4.50}$ & ${-2.93}$ & ${-2.35}$ & ${-2.51}$ \\\hline
SA (Alg.~\ref{alg:robbins-monro}) & ${-1.57}$ & ${-2.00}$ & ${-1.99}$ & ${-2.00}$ \\
\hline
\end{tabular}
}
\end{table}

We observe in the top panels of Fig.~\ref{fig:swap:comparison} that the
multilevel SA scheme significantly outperforms the nested SA scheme in
terms of computational time and achieved error rate. Indeed, for a target
RMSE of the order $10$ for the VaR, MLSA is $10^{3}$ times faster than
NSA as the former executes on average in approximately $10^{-2}s$, while
the latter executes in $10s$. Similarly, for an RMSE of the order
$10$ for the ES, MLSA is $10$~times faster than NSA as the former executes
on average in~$10s$ while the latter executes in~$10^{2}s$. These results
are very promising as they apply to a more realistic scenario than in Sect.~\ref{ssec:numerical}.
Finally, the standard Robbins--Monro algorithm, which is only applicable
if the loss $X_{0}$ is directly simulatable, outperforms both Algorithms~\ref{alg:nested-sa}
and~\ref{alg:mlsa} as expected. As indicated in the case study of Sect.~\ref{ssec:numerical},
the in-between positioning of the MLSA curve relatively to the two other
ones illustrates the partial regain by the multilevel scheme of the performance
lost due to the nested simulation of $X_{0}$. Once again, we point out
that MLSA exhibits some numerical instability when simulating the VaR,
while it remains very robust when simulating the ES. The VaR multilevel
scheme is non-linearly sensitive to the parametrisation of the step sizes
$(\gamma _{n})_{n\geq 1}$. In contrast, the ES multilevel scheme is significantly
more stable with respect to potential mis-tuning of the step sizes
$(\gamma _{n})_{n\geq 1}$.

Within a variation range of the order $\pm 0.2$, the RMSE slopes reported
in Table~\ref{tbl:swap:mlsa:slopes} roughly match the exponents of the
complexities in Theorems~\ref{thm:cost:nsa} and~\ref{thm:cost:mlsa}. The
VaR-focused MLSA slope may be improved with deeper hyperparameter tuning.
This costly hyperparameter tuning phase is in fact an effect of the constraint
$\bar\lambda _{2}\gamma _{1}>2$ on the optimal choice $\beta =1$ as stressed
in the VaR-focused Theorem~\ref{thm:cost:mlsa}(\ref{thm:cost:mlsa:var}).
Although~$\bar\lambda _{2}$ is explicit in \eqref{eq:lambda}, it is tedious
to compute; hence the necessary upstream tuning phase.

The bottom panels of Fig.~\ref{fig:swap:comparison} display a quasi-log-linear
dependency of the average execution time with respect to the prescribed
accuracy. The corresponding slopes in Table~\ref{tbl:swap:mlsa:slopes}
match quite accurately the theoretical exponents from the complexity Theorems~\ref{thm:cost:nsa}
and~\ref{thm:cost:mlsa}.

\section{Conclusion}
\label{sec5}

In this work, we presented a stochastic approximation problem whose solution
retrieves the VaR and ES of the loss of a financial derivatives portfolio.
This loss, however, can only be simulated in a nested Monte Carlo fashion.
A naively nested algorithm performs one inner simulation layer within each
outer update of a stochastic approximation scheme. A single (bias) parameter
allows to control simultaneously the bias, the amounts of iterations and
the complexity of the nested algorithm. But considering a prescribed accuracy
$\varepsilon $, this algorithm can only optimally achieve a complexity
of the order $\varepsilon ^{-3}$. Our multilevel algorithm combines multiple
paired estimates obtained by a geometric progression of biased nested schemes
in a manner that reduces the overall simulation complexity.
\\

According to user intent, whether it is to simulate the VaR or the ES,
our theoretical study identifies two corresponding optimal parametrisations.
The VaR-focused parameter setup results optimally in a complexity of the
order $\varepsilon ^{-2-\delta}$, where $\delta \in (0,1)$ is a specific
parameter that depends on the integrability degree of the loss, while the
ES-focused parameter setup yields a complexity of the order
$\varepsilon ^{-2}|\ln{\varepsilon}|^{2}$. These complexities are attained
for a step size that is proportional to $\frac{1}{n}$, albeit at the cost
of a nontrivial constraint on the proportionality factor. The empirical
studies highlight the clear outperformance of the multilevel algorithm
over the nested one. They demonstrate how in situations where the inner
simulations cannot be cut short by explicit formulas, utilising the multilevel
scheme results in a partial yet significant regain on the performance gap
between the nested SA scheme and the classical (non-nested) Robbins--Monro
scheme. This finding is of great interest when the latter scheme does not
apply. The numerics have also underlined some degree of instability of
MLSA when estimating the VaR, while they highlighted high robustness when
estimating the ES. The former behaviour is in line with the aforementioned
proportionality factor constraint on the step sizes. As the ES is increasingly
replacing the VaR as a fundamental financial risk measure, our ES-focused
multilevel algorithm, running in
$\OO (\varepsilon ^{-2}|\ln{\varepsilon}|^{2})$ for a prescribed accuracy
$\varepsilon $, should be of practical value for risk management.
\\

As a follow-up to this paper, we have looked into complementing our non-asymptotic
error controls with asymptotic error distributions via central limit theorems;
see Cr\'epey et al.~\cite{amlsa}. We have also studied there the Polyak--Ruppert
variants of the NSA and MLSA methods as a means to rid the computationally
optimal case of the nontrivial constraint on the learning rate. Antithetic
multilevel schemes are known to reduce the variances of the paired estimators
within each independent level of MLSA; see~Giles et al.~\cite{giles2024efficient},
Giles and Szpruch~\cite{10.1007/978-3-642-41095-6_16}, Bourgey et al.~\cite{BourgeyDeMarcoGobetZhou+2020+131+161}
and Xu et al.~\cite{XU2024115745}. This incentivises future research in
this direction. A natural question that arises from our study is whether
in the case of the VaR estimation, one could close in a simple way the
performance gap between the multilevel scheme and the classical Robbins--Monro
scheme. A possible line of research in this direction would be to explore
the added benefit of an adaptive selection of inner samples at each level
of MLSA, as proposed by Giles and Haji-Ali~\cite{doi:10.1137/18M1173186}.
Finally, in order to build stronger estimators, one could investigate other
aggregation methods such as the multi-step Richardson--Romberg extrapolation
of Frikha and Huang~\cite{FRIKHA20154066}.

\paragraph{Acknowledgements.}
The research of S. Cr\'{e}pey has benefited from the support of \textit{Chair Capital Markets Tomorrow: Modeling and Computational Issues} under the aegis of the Institut Europlace de Finance, a joint initiative of Laboratoire de Probabilit\'{e}s, Statistique et Mod\'{e}lisation (LPSM) / Universit\'{e} Paris Cit\'{e} and Cr\'{e}dit Agricole CIB; \textit{Chair Futures of Quantitative Finance}, a partnership between LPSM at Universit\'{e} Paris Cit\'{e}, CERMICS at \'{E}cole Nationale des Ponts et Chauss\'{e}es, and BNP Paribas Global Markets.

The research of N. Frikha has benefited from the support of the Institut Europlace de Finance.

\appendix

\section{Proof of Lemma~\ref{lmm:thetah->theta0}}
\label{prf:thetah->theta0}

\noindent
\textbf{\emph{Step~1. Accumulation points of $(\xi ^{h}_{\star})_{h\in \mathcal{H}}$.}}\
Let $h\in \mathcal{H}$. By Lemma~\ref{lmm:stochopt},
the set~$\Theta _{h}$ is a closed non-empty bounded interval that coincides
with the set of roots of $V_{h}'$.

Because $(X_{h})_{h\in \mathcal{H}}$ converges in distribution to
$X_{0}$ as $\mathcal{H}\ni h\downarrow 0$, the sequence
$(F_{X_{h}})_{h\in \mathcal{H}}$ of distribution functions converges pointwise
on $\mathbb{R}$ towards~$F_{X_{0}}$ as~\text{${\mathcal{H}\ni h
\downarrow 0}$.} The distribution function $F_{X_{0}}$ being continuous,
the second Dini theorem implies that
$(F_{X_{h}})_{h\in \mathcal{H}}$ converges uniformly on $\mathbb{R}$ towards
$F_{X_{0}}$ as $\mathcal{H}\ni h\downarrow 0$. According to the definition
of the functions $(V_{h}')_{h\in \overline{\mathcal{H}}}$, it follows that
$(V_{h}')_{h\in \mathcal{H}}$ converges uniformly on $\mathbb{R}$ towards
$V_{0}'$. For any $\xi ^{h}_{\star}\in \Theta _{h}$,
$h\in \mathcal{H}$, one has
\begin{equation*}
|V_{0}'(\xi ^{h}_{\star})|=|V_{0}'(\xi ^{h}_{\star})-V_{h}'(\xi ^{h}_{
\star})|\leq \sup _{\xi \in \mathbb{R}}{|V_{0}'(\xi )-V_{h}'(\xi )|}.
\end{equation*}
Let
$\bar\xi ^{h}_{\star}\in\Argmax _{\xi \in \Theta _{h}}\dist (\xi ,
\Theta _{0})$.
$\Theta _{h}$ being a closed non-empty bounded interval, it is convex and $\bar\xi ^{h}_{\star}$ exists.
As
$\lim _{\mathcal{H}\ni h\downarrow 0}{\sup _{\xi \in \mathbb{R}}{|V_{0}'(
\xi )-V_{h}'(\xi )|}}=0$, we get
\begin{equation}
\label{eq:lim-k-thetah}
\lim _{\mathcal{H}\ni h\downarrow 0}{V_{0}'(\bar\xi _{\star}^{h})}=0.
\end{equation}
With $\{V_{0}'=0\}$ being bounded according to Lemma~\ref{lmm:stochopt},
we deduce that the sequence
$(\bar\xi ^{h}_{\star})_{h\in \mathcal{H}}$ is bounded. Using
\eqref{eq:lim-k-thetah} and the continuity of $V'_{0}$, any accumulation
point~$\bar\xi ^{0}_{\star}$ of the sequence
$(\bar\xi ^{h}_{\star})_{h\in \mathcal{H}}$ satisfies
$V'_{0}(\bar\xi ^{0}_{\star})=0$, so that
$\bar\xi ^{0}_{\star}\in \Theta _{0}$ and
\begin{equation*}
\lim _{\mathcal{H}\ni h\downarrow 0}{\dist (\bar\xi ^{h}_{\star},
\Theta _{0})}=0.
\end{equation*}

\noindent
\textbf{\emph{Step~2. Limit of $(\chi ^{h}_{\star})_{h\in \mathcal{H}}$.}}\
We now assume that
$(X_{h})_{h\in \mathcal{H}}$ converge to $X_{0}$ in~$L^{1}(
\mathbb{P})$. By Lemma~\ref{lmm:stochopt}, for any
$h\in \overline{\mathcal{H}}$,
\begin{equation*}
|\chi ^{h}_{\star}-\chi ^{0}_{\star}| =\Big|\min _{\xi}{V_{h}(\xi )}-
\min _{\xi}{V_{0}(\xi )}\Big| \leq \max _{\xi}{|V_{h}(\xi)-V_{0}(\xi )|}
\leq \frac{1}{1-\alpha}\mathbb{E}[|X_{h}-X_{0}|]
\end{equation*}
so that
$\lim _{\mathcal{H}\ni h\downarrow 0}\chi ^{h}_{\star}=\chi ^{0}_{
\star}$. \qed

\section{Proof of Proposition~\ref{prp:bias-cv}}
\label{prf:bias-cv}

\noindent
\textbf{\emph{Step~1. Asymptotic expansion of $\xi ^{h}_{\star}$ as $\mathcal{H}\ni h\downarrow 0$.}}\
Since the density function $f_{X_{0}}$
of $X_{0}$ is strictly positive, the set
$\Theta _{0}=\Argmin{V_{0}}=\{\xi ^{0}_{\star}\}$ is a singleton so that
by Lemma~\ref{lmm:thetah->theta0},
$\xi ^{h}_{\star}\rightarrow \xi ^{0}_{\star}$ as
$\mathcal{H}\ni h\downarrow 0$. Also, note that
\begin{equation}
\label{eq:theta-bias-1}
F_{X_{h}}(\xi ^{h}_{\star})-F_{X_{h}}(\xi ^{0}_{\star})=\alpha -F_{X_{h}}(
\xi ^{0}_{\star})=F_{X_{0}}(\xi ^{0}_{\star})-F_{X_{h}}(\xi ^{0}_{
\star}).
\end{equation}
On the one hand, as $\mathcal{H}\ni h\downarrow 0$, Assumption~\ref{asp:Xh->X0}(\ref{asp:Xh->X0-ii})
gives
\begin{equation}
\label{eq:theta-bias-2}
F_{X_{0}}(\xi ^{0}_{\star})-F_{X_{h}}(\xi ^{0}_{\star})=-v(\xi ^{0}_{
\star})h+\oo (h) \qquad \text{as } \mathcal{H}\ni h\downarrow 0.
\end{equation}
On the other hand, a first-order Taylor expansion with integral remainder
gives
\begin{equation}
\label{eq:theta-bias-3}
F_{X_{h}}(\xi ^{h}_{\star})-F_{X_{h}}(\xi ^{0}_{\star}) =(\xi ^{h}_{
\star}-\xi ^{0}_{\star})\int _{0}^{1}f_{X_{h}}\big(t\xi ^{0}_{\star}+(1-t)
\xi ^{h}_{\star}\big)\mathrm{d}t.
\end{equation}
Since
$\lim _{\mathcal{H}\ni h\downarrow 0}{\xi ^{h}_{\star}}=\xi ^{0}_{
\star}$ and $(f_{X_{h}})_{h\in \mathcal{H}}$ converges by Assumption~\ref{asp:Xh->X0}(\ref{asp:Xh->X0-iii})
locally uniformly to the strictly positive density function
$f_{X_{0}}$, we have for small $h\in \mathcal{H}$ that
$\int _{0}^{1}f_{X_{h}}(t\xi ^{0}_{\star}+(1-t)\xi ^{h}_{\star})
\mathrm{d}t>0$ and \text{${\lim _{h\downarrow 0}\int _{0}^{1}f_{X_{h}}(t
\xi ^{0}_{\star}+(1-t)\xi ^{h}_{\star})\mathrm{d}t=f_{X_{0}}(\xi ^{0}_{
\star})}$.} Combining \eqref{eq:theta-bias-1}--\eqref{eq:theta-bias-3}, we deduce
\begin{equation*}
\begin{aligned}
h^{-1}(\xi ^{h}_{\star}-\xi ^{0}_{\star}) &=-
\frac{v(\xi ^{0}_{\star})}{\int _{0}^{1}f_{X_{h}}(t\xi ^{0}_{\star}+(1-t)\xi ^{h}_{\star})\mathrm{d}t}+
\oo (1)
\\
&=-\frac{v(\xi ^{0}_{\star})}{f_{X_{0}}(\xi ^{0}_{\star})}+\oo (1)
\qquad \text{as } \mathcal{H}\ni h\downarrow 0.
\end{aligned}
\end{equation*}

\noindent
\textbf{\emph{Step~2. Asymptotic expansion of $\chi ^{h}_{\star}$ as $\mathcal{H}\ni h\downarrow 0$.}}\
Going back to the definitions of the
functions $(V_{h})_{h\in \overline{\mathcal{H}}}$ in \eqref{eq:V0} and
\eqref{eq:Vh}, one has
\begin{align}
\label{eq:dev-ch-c}
\chi ^{h}_{\star}-\chi ^{0}_{\star }&=V_{h}(\xi ^{h}_{\star})-V_{0}(
\xi ^{0}_{\star})
\nonumber
\\
&=\mathbb{E}[X_{h}\,|\,X_{h}>\xi ^{h}_{\star}]-\mathbb{E}[X_{0}\,|\,X_{0}>
\xi ^{0}_{\star}]
\nonumber
\\
&=\frac{1}{1-\alpha}(\mathbb{E}[X_{h}\mathds1_{\{X_{h}>\xi ^{h}_{
\star}\}}]-\mathbb{E}[X_{0}\mathds1_{\{X_{0}>\xi ^{0}_{\star}\}}])
\nonumber
\\
&=\frac{1}{1-\alpha}\bigg(\int _{\xi ^{h}_{\star}}^{\infty}\xi
\mathrm{d}F_{X_{h}}(\xi )-\int _{\xi ^{0}_{\star}}^{\infty}\xi
\mathrm{d}F_{X_{0}}(\xi )\bigg)
\nonumber
\\
&=\frac{1}{1-\alpha}\bigg(\int _{\xi ^{h}_{\star}}^{\xi ^{0}_{\star}}
\xi f_{X_{h}}(\xi )\mathrm{d}\xi +\int _{\xi ^{0}_{\star}}^{\infty}
\xi \mathrm{d}(F_{X_{h}}-F_{X_{0}})(\xi )\bigg).
\end{align}
On the one hand, a change of variable and Proposition~\ref{prp:bias-cv}
give
\begin{align}
\label{expansion:term:1}
\int _{\xi ^{h}_{\star}}^{\xi ^{0}_{\star}}\xi f_{X_{h}}(\xi )
\mathrm{d}\xi &=(\xi ^{0}_{\star}-\xi ^{h}_{\star})\int _{0}^{1}\big(t
\xi ^{h}_{\star}+(1-t)\xi ^{0}_{\star}\big)f_{X_{h}}\big(t\xi ^{h}_{
\star}+(1-t)\xi ^{0}_{\star}\big)\mathrm{d}t
\nonumber
\\
&=v(\xi ^{0}_{\star})\xi ^{0}_{\star }h+\oo (h).
\end{align}
On the other hand, for $h\in \overline{\mathcal{H}}$, integrating by parts
yields
\begin{equation*}
\begin{aligned}
\int _{\xi ^{0}_{\star}}^{\infty}\xi \mathrm{d}F_{X_{h}}(\xi ) &=-
\xi \big(1-F_{X_{h}}(\xi )\big)\big|_{\xi ^{0}_{\star}}^{\infty}+
\int _{\xi ^{0}_{\star}}^{\infty}\big(1-F_{X_{h}}(\xi )\big)
\mathrm{d}\xi
\\
&=\xi ^{0}_{\star}\big(1-F_{X_{h}}(\xi ^{0}_{\star})\big)+\int _{\xi ^{0}_{
\star}}^{\infty}\big(1-F_{X_{h}}(\xi )\big)\mathrm{d}\xi
\end{aligned}
\end{equation*}
and hence
\begin{equation*}
\int _{\xi ^{0}_{\star}}^{\infty}\xi \mathrm{d}(F_{X_{h}}-F_{X_{0}})(
\xi ) =\xi ^{0}_{\star}\big(F_{X_{0}}(\xi ^{0}_{\star})-F_{X_{h}}(
\xi ^{0}_{\star})\big)+\int _{\xi ^{0}_{\star}}^{\infty}\big(F_{X_{0}}(
\xi )-F_{X_{h}}(\xi )\big) \mathrm{d}\xi .
\end{equation*}
By plugging the first-order Taylor expansion of
$F_{X_{h}}-F_{X_{0}}$ in Assumption~\ref{asp:Xh->X0}(\ref{asp:Xh->X0-ii})
into the right-hand side of the above identity, we obtain
\begin{equation}
\label{eq:int-asymp}
\int _{\xi ^{0}_{\star}}^{\infty}\xi \mathrm{d}(F_{X_{h}}-F_{X_{0}})(
\xi ) =-\bigg(\xi ^{0}_{\star }v(\xi ^{0}_{\star})+\int _{\xi ^{0}_{
\star}}^{\infty }v(\xi ) \mathrm{d}\xi \bigg)h+\oo (h).
\end{equation}
Finally, by combining \eqref{eq:dev-ch-c}--\eqref{eq:int-asymp}, we get
\begin{equation*}
\chi ^{h}_{\star}-\chi ^{0}_{\star}=-h\int _{\xi ^{0}_{\star}}^{
\infty }\frac{v(\xi )}{1-\alpha} \mathrm{d}\xi +\oo (h), \qquad
\text{as } \mathcal{H}\ni h\downarrow 0.\qed
\end{equation*}

\section{Proof of Theorem~\ref{thm:variance-cv}}
\label{prf:variance-cv}

We follow a strategy similar to the one developed in Costa and Gadat~\cite[Sect.~2.1]{10.1214/21-EJS1908}
and define, for $h\in \overline{\mathcal{H}}$, $\mu\geq0$, and
$q \in \mathbb{N}_0$, the Lyapunov function
$\mathcal{L}_{h,q}:\mathbb{R}\to \mathbb{R}_{+}$ by
\begin{equation}
\label{eq:Lq}
\mathcal{L}^{\mu}_{h,q}(\xi )=\big(V_{h}(\xi )-V_{h}(\xi ^{h}_{\star})
\big)^{q}\exp \Big(\mu\big(V_{h}(\xi )-V_{h}(\xi ^{h}_{\star})\big)\Big),
\qquad \xi \in \mathbb{R}.
\end{equation}
Hereafter, we set $\mu_{h,q}=q\|V_h''\|_\infty$ and $\bar{\mathcal{L}}_{h,q}=\mathcal{L}^{\mu_{h,q}}_{h,q}$, $h\in\overline{\mathcal{H}}$, for $q\in\mathbb{N}$.
\begin{lemma}
\label{lmm:lyapunov}
Under Assumption~\ref{asp:misc}, for any
$h\in \overline{\mathcal{H}}$, $\mu\geq0$ and $q\in \mathbb{N}$, we have:
\begin{enumerate}[\rm(i)]
\item
\label{lmm:lyapunov-i}
$\mathcal{L}^{\mu}_{h,q}$ is twice continuously differentiable on
$\mathbb{R}$.
\item
\label{lmm:lyapunov-i-bis}
For any $\xi \in \mathbb{R}$,
\begin{equation*}
\bar{\mathcal{L}}_{h,q}(\xi )
\leq k_{\alpha}^{q} |\xi -\xi ^{h}_{\star}|^{q}
\exp \bigg(\frac{q}{1-\alpha}k_{\alpha}\sup_{h'\in\mathcal{H}}\|f_{X_{h'}}\|_\infty|\xi-\xi^{h}_{\star}|\bigg).
\end{equation*}
\item
\label{lmm:lyapunov-ii}
For any $\xi \in \mathbb{R}$,
\begin{equation*}
V_{h}'(\xi )(\mathcal{L}^{\mu}_{h,q})'(\xi )\geq \lambda^{\mu} _{h,q}\mathcal{L}^{\mu}_{h,q}(\xi ),
\qquad \text{where }
\lambda^{\mu}_{h,q}:=\frac{3}{8}qV_{h}''(\xi ^{h}_{\star})
\wedge\mu\frac{V_{h}''(\xi ^{h}_{\star})^{4}}{4[V_{h}'']_{{\mathrm{Lip}}}^{2}}.
\end{equation*}
Furthermore, for any $q'\geq q$,
$\inf _{h\in \overline{\mathcal{H}}}{\lambda^{\mu_{h,q'}}_{h,q}}>0$.
In particular, denoting $\bar\lambda_{h,q}:=\lambda^{\mu_{h,q}}_{h,q}$, then $\inf_{h\in\overline{\mathcal{H}}}\bar{\lambda}_{h,q}>0$.
\item
\label{lmm:lyapunov-iii}
For any $\xi \in \mathbb{R}$,
\begin{equation*}
|(\mathcal{L}^{\mu}_{h,q})''(\xi )| \leq \eta^{\mu}_{h,q}\big(
\mathcal{L}^{\mu}_{h,q}(\xi )+\mathcal{L}^{\mu}_{h,q-1}(\xi )\big),
\end{equation*}
where
\begin{equation*}
\eta^{\mu}_{h,q}
:= (q\vee\mu)\|V_{h}''\|_{\infty}
+ k_{\alpha}^{2}\mu(\mu\vee2)
+ q\big(2\mu \vee (q-1)\big)
\bigg(\frac{3k_{\alpha}^{2}[V_{h}'']_{{\mathrm{Lip}}}^{2}}{V_{h}''(\xi ^{h}_{\star})^{3}}
\vee \frac{3\|V_{h}''\|_{\infty}^{2}}{V_{h}''(\xi ^{h}_{\star})}\bigg).
\end{equation*}
Besides, for any $q'\geq q$,
$|\lambda^{\mu_{h,q'}}_{h,q}|^{2}\leq\eta^{\mu_{h,q'}}_{h,q}$ and
$\sup _{h\in \overline{\mathcal{H}}}{\eta^{\mu_{h,q'}}_{h,q}}<\infty$.
In particular, denoting $\bar{\eta}_{h,q}:=\eta^{\mu_{h,q}}_{h,q}$, then $|\bar{\lambda}_{h,q}|^2\leq\bar{\eta}_{h,q}$ and $\sup _{h\in \overline{\mathcal{H}}}{\bar{\eta}_{h,q}}<\infty$.
\item
\label{lmm:lyapunov-iv}
For any $\mu'\geq0$ and any $\xi \in \mathbb{R}$,
\begin{equation*}
(\xi -\xi ^{h}_{\star})^{2q} \leq \kappa _{h,q}\big(\mathcal{L}^{\mu}_{h,q}(
\xi )+\mathcal{L}^{\mu'}_{h,2q}(\xi )\big), \qquad \text{where }
\kappa _{h,q}:=\frac{3^{q}}{V_{h}''(\xi ^{h}_{\star})^{q}}\vee
\frac{3^{2q}[V_{h}'']_{{\mathrm{Lip}}}^{2q}}{V_{h}''(\xi ^{h}_{\star})^{4q}}.
\end{equation*}
Moreover,
$\sup _{h\in \overline{\mathcal{H}}}{\kappa _{h,q}}<\infty $.
\end{enumerate}
\end{lemma}
\begin{proof}
See Appendix~\ref{prf:lyapunov}.
\end{proof}

The next result concerns our special choice of learning sequence.

\begin{lemma}
\label{lmm:gamma}
Let
\begin{equation}
\label{eq:phi}
\varphi _{\eta}(t):=
\begin{cases}
\eta ^{-1}(t^{\eta}-1), &\quad \eta \neq 0,
\\
\ln{t}+\gamma ^{\mathcal{E}}, &\quad \eta =0,
\end{cases}
\qquad t>0,
\end{equation}
where $\gamma ^{\mathcal{E}}$ is the Euler--Mascheroni constant. Let
$\gamma _{n}=\gamma _{1}n^{-\beta}$, $\gamma _{1}>0$,
$\beta \in (0,1]$, and for $0\leq k\leq n$,
\begin{equation*}
\Pi _{k+1:n}:=\prod _{j=k+1}^{n}(1-\lambda \gamma _{j}+\zeta \gamma _{j}^{2}),
\end{equation*}
where $\lambda $, $\zeta $ are some positive parameters, with the convention
$\prod _{\varnothing}=1$. Assume that
$1-\lambda \gamma _{n}+\zeta \gamma _{n}^{2}>0$ for any positive integer
$n$. Then for $n\geq 2$, we have:
\begin{enumerate}[\rm(i)]
\item For any $0\leq k\leq n$,
\begin{equation*}
\Pi _{k+1:n}\leq
\begin{cases}
\exp (\zeta \frac{\pi ^{2}}{6}\gamma _{1}^{2}+
\frac{\lambda\gamma _{1}}{2})
\frac{(k+1)^{\lambda \gamma _{1}}}{(n+1)^{\lambda \gamma _{1}}}, &
\quad \beta =1,
\\
\exp (-\lambda \gamma _{1}(\varphi _{1-\beta}(n+1)-\varphi _{1-\beta}(k+1)))
\\
\quad
\times \exp (2^{2\beta}\zeta \gamma _{1}^{2}(\varphi _{1-2\beta}(n+1)-
\varphi _{1-2\beta}(k+1))), &\quad \beta \in (0,1).
\end{cases}
\end{equation*}
\item For any $p\geq 2$,
\begin{equation*}
\begin{aligned}
&\sum _{k=1}^{n}\gamma _{k}^{p}\Pi _{k+1:n}
\\
&\leq
\begin{cases}
\gamma _{1}^{p}\exp (\zeta \frac{\pi ^{2}}{6}\gamma _{1}^{2}+
\frac{\lambda\gamma _{1}}{2}+\ln (2)(\lambda\gamma _{1} \vee p))
\frac{\varphi _{\lambda \gamma _{1}-p+1}(n+1)}{(n+1)^{\lambda \gamma _{1}}},
&\quad \beta =1,
\\
\zeta^{-1}\gamma _{1}^{p-2}
\exp (2^{2\beta +1}\zeta
\gamma _{1}^{2}\varphi _{1-2\beta}(n+1)-
\frac{\lambda\gamma _{1}}{2}\varphi _{1-\beta}(n+1))\\
+2^{2\beta}\gamma _{1}^{p}\exp (-2^{-(\beta +2)}\lambda \gamma _{1}n^{1-
\beta})\varphi _{1-2\beta}(n+1)+
\frac{2^{1+(p-1)\beta}\gamma _{1}^{p-1}}{\lambda n^{(p-1)\beta}}, &
\quad \beta \in (0,1).
\end{cases}
\end{aligned}
\end{equation*}
\end{enumerate}
\end{lemma}

\begin{proof}
See Appendix~\ref{prf:gamma}.
\end{proof}

As in the proof of Lemma~\ref{lmm:lyapunov} in Appendix~\ref{prf:lyapunov},
we drop the superscript $\mu $ from our notation and write
$\mathcal{L}_{h,q}(\xi )$ for $\mathcal{L}^{\mu}_{h,q}(\xi )$. For any
$h\in \mathcal{H}$, we introduce the filtration
$\mathbb{F}^{h}=(\mathcal{F}^{h}_{n})_{n\geq 0}$, where
$\mathcal{F}^{h}_{n}=\sigma (\xi ^{h}_{0},\chi ^{h}_{0},X_{h}^{(1)},
\dots ,X_{h}^{(n)})$, $n\geq 0$.
Throughout, we fix a bias parameter
$h\in \mathcal{H}$, a strictly positive integer $q$ and $\mu >0$.

\begin{proof}[Proof of Theorem~\ref{thm:variance-cv}]
\noindent\textbf{\emph{Step~1. General inequality on
$\mathbb{E}[\mathcal{L}_{h,q}(\xi ^{h}_{n})]$.}}\
We first prove a general
inequality on $\mathbb{E}[\mathcal{L}_{h,q}(\xi ^{h}_{n})]$,
$n\geq 0$. We decompose the dynamics of $(\xi ^{h}_{n})_{n\geq 0}$, given
by the first component of \eqref{approximate:sgd:algorithm}, into
\begin{equation}
\label{eq:scheme-bis}
\xi ^{h}_{n+1}=\xi ^{h}_{n}-\gamma _{n+1}V_{h}'(\xi ^{h}_{n})-\gamma _{n+1}e^{h}_{n+1},
\end{equation}
where
\begin{equation*}
e^{h}_{n+1}:= H_{1}(\xi ^{h}_{n},X_{h}^{(n+1)})-V_{h}'(\xi ^{h}_{n}),
\qquad n\geq 0,
\end{equation*}
is an $(\mathbb{F}^{h},\mathbb{P})$-martingale increment. We test the Lyapunov
function $\mathcal{L}_{h,q}$ along the dynamics~\eqref{eq:scheme-bis}.
Via a second-order Taylor expansion, we get
\begin{equation*}
\begin{aligned}
\mathcal{L}_{h,q}(\xi ^{h}_{n+1}) &=\mathcal{L}_{h,q}\big(\xi ^{h}_{n}-
\gamma _{n+1}V_{h}'(\xi ^{h}_{n})-\gamma _{n+1}e^{h}_{n+1}\big)
\\
&=\mathcal{L}_{h,q}(\xi ^{h}_{n})-\gamma _{n+1}\mathcal{L}_{h,q}'(
\xi ^{h}_{n})\big(V_{h}'(\xi ^{h}_{n})+e^{h}_{n+1}\big)
\\
&
\hphantom{=:}
+\gamma _{n+1}^{2} \big(H_{1}(\xi ^{h}_{n},X_{h}^{(n+1)})\big)^{2}
\int _{0}^{1}(1-t)\mathcal{L}_{h,q}''\big(t\xi ^{h}_{n+1}+(1-t)\xi ^{h}_{n}
\big) \mathrm{d}t.
\end{aligned}
\end{equation*}
It follows from Lemma~\ref{lmm:lyapunov}(\ref{lmm:lyapunov-ii})-(\ref{lmm:lyapunov-iii})
that
\begin{eqnarray}
\label{eq:dynamics}
&&\mathcal{L}_{h,q}(\xi ^{h}_{n+1})
\nonumber
\\
&&\leq \mathcal{L}_{h,q}(\xi ^{h}_{n})(1-\lambda ^{\mu}_{h,q}\gamma _{n+1})-
\gamma _{n+1}\mathcal{L}_{h,q}'(\xi ^{h}_{n})e^{h}_{n+1}+\eta ^{\mu}_{h,q}
\gamma _{n+1}^{2} \big(H_{1}(\xi ^{h}_{n},X_{h}^{(n+1)})\big)^{2}
\nonumber
\\
&&
\hphantom{=:}
\times \int _{0}^{1}(1-t)\Big(\mathcal{L}_{h,q}\big(t\xi ^{h}_{n+1}+(1-t)
\xi ^{h}_{n}\big)+\mathcal{L}_{h,q-1}\big(t\xi ^{h}_{n+1}+(1-t)\xi ^{h}_{n}
\big)\Big)\mathrm{d}t.
\end{eqnarray}
Let $t\in [0,1]$. The mean value theorem guarantees that there exists
$\widetilde{\xi}^{h}_{n}(t)\in \mathbb{R}$ with
\begin{equation}
\label{eq:mean-value}
V_{h}\big(t\xi ^{h}_{n+1}+(1-t)\xi ^{h}_{n}\big)=V_{h}(\xi ^{h}_{n})+tV_{h}'
\big(\widetilde{\xi}^{h}_{n}(t)\big)(\xi ^{h}_{n+1}-\xi ^{h}_{n}).
\end{equation}
From \eqref{eq:H1} and \eqref{approximate:sgd:algorithm},
\begin{equation}
\label{eq:|thetap-p+1|<}
|\xi ^{h}_{n+1}-\xi ^{h}_{n}| =\gamma _{n+1}|H_{1}(\xi ^{h}_{n},X_{h}^{(n+1)})|
\leq k_{\alpha}\gamma _{n+1}
\end{equation}
with $k_{\alpha}=\frac{\alpha}{1-\alpha}\vee 1$. Hence applying the triangle
inequality to \eqref{eq:mean-value} and recalling that $V_{h}$~is minimal
at $\xi ^{h}_{\star}$, we obtain
\begin{equation}
\label{eq:Vht-Vh<}
0\leq V_{h}\big(t \xi ^{h}_{n+1}+(1-t)\xi ^{h}_{n}\big)-V_{h}(\xi ^{h}_{
\star})\leq V_{h}(\xi ^{h}_{n})-V_{h}(\xi ^{h}_{\star})+k_{\alpha}^{2}
\gamma _{n+1}.
\end{equation}
Using the inequality
$\mathrm{e} ^{x}\leq \mathrm{e} \mathds1_{\{x\leq 1\}}+x^{q}\mathrm{e} ^{x}\mathds1_{\{x>1\}}
\leq \mathrm{e} (1+x^{q}\mathrm{e} ^{x})$, $x\in \mathbb{R}$, and \eqref{eq:Lq}, we get
for any $\xi \in \mathbb{R}$ that
\begin{equation}
\label{eq:L0<Lq}
\mathcal{L}_{h,0}(\xi )\leq \mathrm{e} \big(1+\mu ^{q}\mathcal{L}_{h,q}(\xi )
\big).
\end{equation}
Thus using \eqref{eq:Vht-Vh<} and then \eqref{eq:L0<Lq} gives
\begin{align}
\label{eq:Lt<}
&\mathcal{L}_{h,q}\big(t\xi ^{h}_{n+1}+(1-t)\xi ^{h}_{n}\big)
\nonumber
\\
&\leq \big(V_{h}(\xi ^{h}_{n})-V_{h}(\xi ^{h}_{\star})+k_{\alpha}^{2}
\gamma _{n+1}\big)^{q} \exp \Big(\mu \big(V_{h}(\xi ^{h}_{n})-V_{h}(
\xi ^{h}_{\star})+k_{\alpha}^{2}\gamma _{n+1}\big)\Big)
\nonumber
\\
&\leq 2^{q-1}\exp (\mu k_{\alpha}^{2}\gamma _{n+1} ) \big(\mathcal{L}_{h,q}(
\xi ^{h}_{n})+k_{\alpha}^{2q}\gamma _{n+1}^{q}\mathcal{L}_{h,0}(\xi ^{h}_{n})
\big)
\nonumber
\\
&\leq \sigma ^{\mu}_{q}\big(\mathcal{L}_{h,q}(\xi ^{h}_{n})+\gamma _{n+1}^{q}
\big),
\end{align}
where
\begin{equation*}
\sigma ^{\mu}_{q}:=2^{q-1}\exp (\mu k_{\alpha}^{2}\gamma _{1})\big((1+
\mathrm{e} \mu ^{q}k_{\alpha}^{2q}\gamma _{1}^{q})\vee \mathrm{e} k_{\alpha}^{2q}
\big)\geq 2^{q-1}.
\end{equation*}
Besides, by \eqref{eq:|thetap-p+1|<}, we have
\begin{equation}
\label{eq:update2<}
\gamma _{n+1}^{2}\big(H_{1}(\xi ^{h}_{n}, X_{h}^{(n+1)})\big)^{2}
\leq k_{\alpha}^{2}\gamma _{n+1}^{2}.
\end{equation}
Plugging the upper bounds \eqref{eq:Lt<} and \eqref{eq:update2<} into
\eqref{eq:dynamics} yields
\begin{align}
\label{eq:Lhq(xih{n+1})<}
\mathcal{L}_{h,q}(\xi ^{h}_{n+1}) &\leq \mathcal{L}_{h,q}(\xi ^{h}_{n})(1-
\lambda ^{\mu}_{h,q}\gamma _{n+1}+\zeta ^{\mu}_{h,q}\gamma ^{2}_{n+1})
-\gamma _{n+1}\mathcal{L}_{h,q}'(\xi ^{h}_{n})e^{h}_{n+1}
\nonumber
\\
&
\hphantom{=:}
+\zeta ^{\mu}_{h,q}\gamma _{n+1}^{2}\mathcal{L}_{h,q-1}(\xi ^{h}_{n})+
\zeta ^{\mu}_{h,q}\gamma _{n+1}^{q+1}
\end{align}
with
\begin{equation}
\label{eq:xihq}
\zeta ^{\mu}_{h,q} :=\frac{1}{2}\eta ^{\mu}_{h,q}k_{\alpha}^{2}\big((
\gamma _{1}\sigma ^{\mu}_{q}+\sigma ^{\mu}_{q-1})\vee \sigma ^{\mu}_{q}
\big),
\end{equation}
which according to the second part of Lemma~\ref{lmm:lyapunov}(\ref{lmm:lyapunov-iii})
satisfies, for any $q'\geq q$,
\begin{equation*}
\sup _{h\in \overline{\mathcal{H}}} \zeta ^{\mu _{h,q'}}_{h,q}<\infty .
\end{equation*}
In particular,
\begin{equation*}
\bar\zeta _{q}
:= \sup _{h\in \overline{\mathcal{H}}}\bar\zeta _{h,q}
:= \sup _{h\in \overline{\mathcal{H}}} \zeta ^{\mu _{h,q}}_{h,q}<\infty.
\end{equation*}
Since $(e^{h}_{n})_{n\geq 1}$ are $(\mathbb{F}^{h},\mathbb{P})$-martingale
increments, the tower law gives
\begin{equation*}
\mathbb{E}[\mathcal{L}_{h,q}'(\xi ^{h}_{n})e^{h}_{n+1}] =\mathbb{E}
\big[\mathcal{L}_{h,q}'(\xi ^{h}_{n})\mathbb{E}[e^{h}_{n+1}\,|\,
\mathcal{F}^{h}_{n}]\big] =0.
\end{equation*}
Hence taking expectations on both sides of \eqref{eq:Lhq(xih{n+1})<} leads
to
\begin{align}
\label{eq:ELq<}
\mathbb{E}[\mathcal{L}_{h,q}(\xi ^{h}_{n+1})] &\leq \mathbb{E}[
\mathcal{L}_{h,q}(\xi ^{h}_{n})](1- \lambda ^{\mu}_{h,q}\gamma _{n+1}+
\zeta ^{\mu}_{h,q}\gamma _{n+1}^{2})
\nonumber
\\
&
\hphantom{=:}
+\zeta ^{\mu}_{h,q}\gamma _{n+1}^{2}\mathbb{E}[\mathcal{L}_{h,q-1}(
\xi ^{h}_{n})] +\zeta ^{\mu}_{h,q}\gamma _{n+1}^{q+1}.
\end{align}

\noindent\textbf{\emph{Step~2. Inequality on
$\mathbb{E}[\bar{\mathcal{L}}_{h,1}(\xi ^{h}_{n})]$}}\
We prove here a
sharper upper estimate on
$\mathbb{E}[\bar{\mathcal{L}}_{h,1}(\xi ^{h}_{n})]$. Taking $q=1$ in
\eqref{eq:ELq<} and \eqref{eq:L0<Lq}, we get
\begin{align}
\label{inequality:lyapunov:funct:iterate}
\mathbb{E}[\mathcal{L}_{h,1}(\xi ^{h}_{n+1})] &\leq \mathbb{E}[
\mathcal{L}_{h,1}(\xi ^{h}_{n})]\big(1-\lambda ^{\mu}_{h,1}\gamma _{n+1}+(
\mathrm{e} \mu +1)\zeta ^{\mu}_{h,1}\gamma _{n+1}^{2}\big)
\nonumber
\\
&
\hphantom{=:}
+(\mathrm{e} +1)\zeta ^{\mu}_{h,1}\gamma _{n+1}^{2}.
\end{align}
Using \eqref{eq:xihq}, that
$|\bar\lambda _{h,q}|^{2}\leq \bar\eta _{h,q}$ by Lemma~\ref{lmm:lyapunov}(\ref{lmm:lyapunov-ii})
and that $x(1-cx)\leq (4c)^{-1}$ for $x\in \mathbb R$, $c>0$, we obtain
\begin{align}
\label{eq:lambda_gamma-zeta_gamma^2<1/8}
\bar\lambda _{h,q}\gamma _{n+1}-(\mathrm{e} \mu _{h,q}^{q}+1)\bar\zeta _{h,q}
\gamma _{n+1}^{2} & \leq \sqrt{\bar\eta _{h,q}}\,\gamma _{n+1}-
\bar\eta _{h,q}\bar{c}_{h,q}\gamma ^{2}_{n+1}
\nonumber
\\
&=\sqrt{\bar\eta _{h,q}}\,\gamma _{n+1}(1-\bar{c}_{h,q}\sqrt{
\bar\eta _{h,q}}\,\gamma _{n+1})
\nonumber
\\
&\leq \frac{1}{4\bar{c}_{h,q}} \leq \frac{1}{2},
\end{align}
where, recalling that $\mu >0$, $k_{\alpha}\geq 1$ and
$\sigma ^{\mu}_{h,q}\geq 2^{q-1}$,
\begin{align*}
c^{\mu}_{h,q} &:=(\mathrm{e} \mu ^{q} +1)
\frac{\zeta ^{\mu}_{h,q}}{\eta ^{\mu}_{h,q}} =\frac{1}{2}(\mathrm{e} \mu ^{q}+1)k_{
\alpha}^{2}\big((\gamma _{1}\sigma ^{\mu}_{q}+\sigma ^{\mu}_{q-1})
\vee \sigma ^{\mu}_{q}\big)>\frac{1}{2},
\\
\bar{c}_{h,q}&:= c^{\mu _{h,q}}_{h,q}.
\end{align*}
Hence for $q=1$,
\begin{equation*}
1-\bar\lambda _{h,1}\gamma _{n+1}+(\mathrm{e} \mu _{h,1}+1)\bar\zeta _{h,1}
\gamma _{n+1}^{2}\geq \frac{1}{2}.
\end{equation*}
Evaluating \eqref{inequality:lyapunov:funct:iterate} at
$\mu =\mu _{h,1}$ and iterating $n$ times the inequality yields
\begin{equation*}
\mathbb{E}[\bar{\mathcal{L}}_{h,1}(\xi ^{h}_{n})] \leq \mathbb{E}[
\bar{\mathcal{L}}_{h,1}(\xi ^{h}_{0})]\Pi _{1:n}^{h,1,1} +(\mathrm{e} +1)
\bar\zeta _{h,1}\sum _{k=1}^{n}\gamma _{k}^{2}\Pi _{k+1:n}^{h,1,1},
\end{equation*}
where
\begin{equation*}
\Pi _{k:n}^{h,q,q'}:=\prod _{j=k}^{n}
\big(1-\lambda _{h,q}^{\mu_{h,q'}}\gamma _{j}
+(\mathrm{e} \mu _{h,q'}^{q}+1)\zeta _{h,q}^{\mu_{h,q'}}\gamma _{j}^{2}\big),
\end{equation*}
with the convention $\prod _{\varnothing}=1$.

If $\beta =1$, invoking Lemma~\ref{lmm:gamma} with $p=2$ gives
\begin{equation*}
\mathbb{E}[\bar{\mathcal{L}}_{h,1}(\xi ^{h}_{n})]\leq \widehat{K}_{h,1}
\frac{\mathbb{E}[\bar{\mathcal{L}}_{h,1}(\xi ^{h}_{0})]}{(n+1)^{\lambda _{h,1}\gamma _{1}}}+
\bar{K}_{h,1}
\frac{\varphi _{\lambda _{h,1}\gamma _{1}-1}(n+1)}{(n+1)^{\lambda _{h,1}\gamma _{1}}},
\end{equation*}
where
\begin{align*}
\widehat{K}_{h,1}&:=\exp \bigg((\mathrm{e} \mu _{h,1}+1)\zeta _{h,1}
\frac{\pi ^{2}}{6}\gamma _{1}^{2}+
\frac{\gamma _{1}\bar\lambda _{h,1}}{2}\bigg),
\\
\bar{K}_{h,1}&:=\widehat{K}_{h,1}\gamma _{1}^{2}(\mathrm{e} +1)\bar\zeta _{h,1}
2^{\gamma _{1}\bar\lambda _{h,1}\vee 2}.
\end{align*}
Hence if $\bar\lambda _{h,1}\gamma _{1}>1$, recalling \eqref{eq:phi}, we
have
\begin{equation}
\label{estimate:Lypanuov:1:gamma:1/n}
\mathbb{E}[\bar{\mathcal{L}}_{h,1}(\xi ^{h}_{n})]\leq K^{1}_{h,1}
\gamma _{n}
\end{equation}
with
$K^{1}_{h,1}:=\gamma _{1}^{-1}(\widehat{K}_{h,1}\mathbb{E}[
\bar{\mathcal{L}}_{h,1}(\xi ^{h}_{0})]+\bar{K}_{h,1})$.

Otherwise, if $\beta \in (0,1)$, applying Lemma~\ref{lmm:gamma} with
$p=2$ yields
\begin{align*}
&\mathbb{E}[\bar{\mathcal{L}}_{h,1}(\xi ^{h}_{n})]
\\
&\leq \bigg(\mathbb{E}[\bar{\mathcal{L}}_{h,1}(\xi ^{h}_{0})]+
\frac{\mathrm{e} +1}{\mathrm{e} \mu _{h,1}+1}\bigg)
\\
&
\hphantom{=:}
\times \exp \bigg(2^{2\beta +1} (\mathrm{e} \mu _{h,1}+1)\bar\zeta _{h,1}
\gamma _{1}^{2}\varphi _{1-2\beta}(n+1)
-\frac{\bar\lambda_{h,1}\gamma _{1}}{2}\varphi _{1-\beta}(n+1)\bigg)
\\
&
\hphantom{=:}
+(\mathrm{e} +1)\bar\zeta _{h,1}\bigg(2^{2\beta}\gamma _{1}^{2}\exp (-2^{-(
\beta +2)}\bar\lambda _{h,1}\gamma _{1}n^{1-\beta})\varphi _{1-2\beta}(n+1)+
\frac{2^{\beta +1}\gamma _{1}}{\bar\lambda _{h,1}n^{\beta}}\bigg).
\end{align*}
Hence for any $\beta \in (0,1)$ and any positive integer $n$, we have
\begin{equation}
\label{estimate:Lypanuov:1:gamma:1/nbeta}
\mathbb{E}[\bar{\mathcal{L}}_{h,1}(\xi ^{h}_{n})]\leq K_{h,1}^{\beta}
\gamma _{n}
\end{equation}
with
\begin{align*}
K_{h,1}^{\beta }&:=\gamma _{1}^{-1}\bigg(\mathbb{E}[\mathcal{L}_{h,1}(
\xi ^{h}_{0})]+\frac{\mathrm{e} +1}{\mathrm{e} \mu _{h,1}+1}\bigg)
\\
&
\hphantom{=:}
\times \sup_{n\geq1}\bigg(n^\beta\exp \big(2^{2\beta +1} (\mathrm{e} \mu _{h,1}+1)\bar\zeta _{h,1}
\gamma _{1}^{2}\varphi _{1-2\beta}(n+1)\big)
\\
&\hphantom{=:\gamma_1^{-1}\sup_{n\geq1}\bigg(}
\times
\exp \bigg(-\frac{\bar\lambda_{h,1}\gamma _{1}}{2}\varphi _{1-\beta}(n+1)\bigg)\bigg)
\\
&
\hphantom{=:}
+(\mathrm{e} +1)\bar\zeta _{h,1}\\
&
\hphantom{=:}
\times
\bigg(2^{2\beta}\gamma _{1}\sup _{n\geq 1}
\big(n^{\beta}\exp (-2^{-(\beta +2)}\bar\lambda _{h,1}\gamma _{1}n^{1-
\beta})\varphi _{1-2\beta}(n+1)\big)
+\frac{2^{\beta +1}}{\bar\lambda _{h,1}}\bigg).
\end{align*}

\noindent\textbf{\emph{Step~3. Inequality on
$\mathbb{E}[\bar{\mathcal{L}}_{h,2}(\xi ^{h}_{n})]$.}}\
The same reasoning in Step~2 with $q=1$ and $\mu=\mu_{h,2}$ yields, assuming that $\lambda^{\mu_{h,2}}_{h,1}\gamma_1
\geq\bar{\lambda}_{h,1}\gamma_1>1$ if $\beta=1$,
\begin{equation}
\label{estimate:Lypanuov:1:2}
\mathbb{E}[\mathcal{L}_{h,1}(\xi ^{h}_{n})]\leq \widetilde{K}_{h,1}^{\beta}
\gamma _{n},
\end{equation}
for some positive constants $(\widetilde{K}_{h,1}^{\beta})_{h\in\mathcal{H}}$ independent of $n$.

Take now $q=2$ and $\mu =\mu _{h,2}$ in \eqref{eq:ELq<}.
Using \eqref{estimate:Lypanuov:1:2} gives
\begin{align}
\label{inequality:lyapunov:funct:iterate:bis}
&\mathbb{E}[\bar{\mathcal{L}}_{h,2}(\xi ^{h}_{n+1})]
\nonumber
\\
&\leq \mathbb{E}[\bar{\mathcal{L}}_{h,2}(\xi ^{h}_{n})](1-
\bar\lambda _{h,2}\gamma _{n+1}+\bar\zeta _{h,2}\gamma _{n+1}^{2}) +
\bar\zeta _{h,2}\gamma _{n+1}^{2}\mathbb{E}[\mathcal{L}_{h,1}(\xi ^{h}_{n})]
+\bar\zeta _{h,2}\gamma _{n+1}^{3}
\nonumber
\\
&\leq \mathbb{E}[\bar{\mathcal{L}}_{h,2}(\xi ^{h}_{n})]\big(1-
\bar\lambda _{h,2}\gamma _{n+1}+(\mathrm{e} \mu _{h,2}^{2}+1)\bar\zeta _{h,2}
\gamma _{n+1}^{2}\big)
+(2^{\beta}\widetilde{K}_{h,1}^{\beta}+1)\bar\zeta _{h,2}\gamma _{n+1}^{3}.
\end{align}
From \eqref{eq:lambda_gamma-zeta_gamma^2<1/8}, we get
\begin{equation*}
1-\bar\lambda _{h,2}\gamma _{n+1}-(\mathrm{e} \mu _{h,2}^{2}+1)\bar\zeta _{h,2}
\gamma _{n+1}^{2}\geq \frac{1}{2}.
\end{equation*}
Hence iterating $n$ times the inequality
\eqref{inequality:lyapunov:funct:iterate:bis} yields
\begin{equation*}
\mathbb{E}[\bar{\mathcal{L}}_{h,2}(\xi ^{h}_{n})]\leq \mathbb{E}[
\bar{\mathcal{L}}_{h,2}(\xi ^{h}_{0})]\Pi ^{h,2,2}_{1:n}+(2^{\beta }\widetilde{K}_{h,1}^{
\beta}+1)\bar\zeta _{h,2}\sum _{k=1}^{n}\gamma _{k}^{3}\Pi ^{h,2,2}_{k+1:n}.
\end{equation*}

If $\gamma _{n}=\gamma _{1}n^{-1}$, by Lemma~\ref{lmm:gamma} for
$p=3$, we have
\begin{equation}
\label{sum:square:gamma:times:prod:final:case:1/n}
\mathbb{E}[\bar{\mathcal{L}}_{h,2}(\xi ^{h}_{n})] \leq \widehat{K}_{h,2}
\frac{\mathbb{E}[\bar{\mathcal{L}}_{h,2}(\xi ^{h}_{0})]}{(n+1)^{\bar\lambda _{h,2}\gamma _{1}}}
+\bar{K}_{h,2}
\frac{\varphi _{\bar\lambda _{h,2}\gamma _{1} -2}(n+1)}{(n+1)^{\bar\lambda _{h,2}\gamma _{1}}},
\end{equation}
where
\begin{align*}
\widehat{K}_{h,2}& :=\exp \bigg((\mathrm{e} \mu _{h,2}^{2}+1)\bar\zeta _{h,2}
\frac{\pi ^{2}}{6}\gamma _{1}^{2}+
\frac{\gamma _{1}\bar\lambda _{h,2}}{2}\bigg),
\\
\bar{K}_{h,2}&:=\widehat{K}_{h,2}\gamma _{1}^{3}(2\widetilde{K}_{h,1}^{1}+1)
\bar\zeta _{h,2}2^{\gamma _{1}\bar\lambda _{h,2}\vee 3}.
\end{align*}
Since $\bar\lambda _{h,2}\geq\bar\lambda _{h,1}>1$, one has
$\bar\lambda _{h,2}\gamma _{1}>1$ so that for any positive integer
$n$,
\begin{equation}
\label{estimate:Lypanuov:2:gamma:1/n}
\mathbb{E}[\bar{\mathcal{L}}_{h,2}(\xi ^{h}_{n})]\leq K^{1}_{h,2}
\gamma _{n}
\end{equation}
with
\begin{equation*}
K^{1}_{h,2}:=\gamma _{1}^{-1}\bigg(\widehat{K}_{h,2}\mathbb{E}[
\bar{\mathcal{L}}_{h,2}(\xi ^{h}_{0})]+\bar{K}_{h,2} \sup _{n\geq 1}
\frac{\varphi _{\bar\lambda _{h,2}\gamma _{1}-2}(n+1)}{(n+1)^{\bar\lambda _{h,2}\gamma _{1}-1}}
\bigg).
\end{equation*}

Otherwise, if $\gamma _{n}=\gamma _{1}n^{-\beta}$ with
$\beta \in (0,1)$, recalling Lemma~\ref{lmm:gamma} with $p=3$ gives
\begin{align}
\label{sum:square:gamma:times:prod:final:bis}
&\mathbb{E}[\bar{\mathcal{L}}_{h,2}(\xi ^{h}_{n})]
\nonumber
\\
&\leq \bigg(\mathbb{E}[\bar{\mathcal{L}}_{h,2}(\xi ^{h}_{0})]+
\frac{\gamma _{1}(2^{\beta }K^{\beta}_{h,1}+1)}{\mathrm{e} \mu _{h,2}^{2}+1}
\bigg)
\nonumber
\\
&
\hphantom{=:}
\times \exp \bigg(2^{2\beta +1}\bar\zeta _{h,2}(\mathrm{e} \mu _{h,2}^{2}+1)
\gamma _{1}^{2}\varphi _{1-2\beta}(n+1)-\frac{\bar\lambda _{h,2}}{2}
\gamma _{1}\varphi _{1-\beta}(n+1)\bigg)
\nonumber
\\
&
\hphantom{=:}
+(2^{\beta }\widetilde{K}_{h,1}^{\beta}+1)\bar\zeta _{h,2}\bigg(2^{\beta}\gamma _{1}^{3}
\exp (-2^{-(\beta +2)}\bar\lambda _{h,2}\gamma _{1}n^{1-\beta})
\varphi _{1-2\beta}(n+1)
\nonumber
\\
&
\hphantom{=:+(2^{\beta }\widetilde{K}_{h,1}^{\beta}+1)\bar\zeta _{h,2}\bigg(}
+\frac{2^{2\beta +1}\gamma _{1}^{2}}{\bar\lambda _{h,2}n^{2\beta}}
\bigg).
\end{align}
Thus for any $\beta \in (0,1)$ and any positive integer $n$, we have
\begin{equation}
\label{estimate:Lypanuov:1:gamma:1/nbeta:bis}
\mathbb{E}[\bar{\mathcal{L}}_{h,2}(\xi ^{h}_{n})]\leq K_{h,2}^{\beta}
\gamma _{n}
\end{equation}
with
\begin{align*}
K_{h,2}^{\beta }&:=\bigg(\gamma _{1}^{-1}\mathbb{E}[\bar{\mathcal{L}}_{h,2}(
\xi ^{h}_{0})]+
\frac{2^{\beta }\widetilde{K}^{\beta}_{h,1}+1}{\mathrm{e} \mu _{h,2}^{2}+1}\bigg)
\nonumber
\\
&
\hphantom{=::}
\times \sup _{n\geq 1}\bigg(n^{\beta}\exp \big(2^{2\beta +1}
\bar\zeta _{h,2}(\mathrm{e} \mu _{h,2}^{2}+1)\gamma _{1}^{2}\varphi _{1-2
\beta}(n+1)\big)
\\
&
\hphantom{=::\times \sup _{n\geq 1}\bigg(}
\times\exp\bigg(
-\frac{\bar\lambda _{h,2}}{2}\gamma _{1}\varphi _{1-\beta}(n+1)\bigg)
\bigg)
\\
&
\hphantom{=::}
+(2^{\beta }\widetilde{K}_{h,1}^{\beta}+1)\bar\zeta _{h,2}
\\
\nonumber
&
\hphantom{=::+}
\times \bigg(\gamma _{1}^{2}2^{2\beta}\sup _{n\geq 1} \big(n^{\beta}
\exp (-2^{-(\beta +2)}\bar\lambda _{h,2}\gamma _{1}n^{1-\beta})
\varphi _{1-2\beta}(n+1) \big)
\nonumber
\\
&
\hphantom{=::+\times \bigg(}
+\gamma _{1}\frac{2^{2\beta +1}}{\bar\lambda _{h,2}}\bigg).
\end{align*}

\noindent\textbf{\emph{Step~4. Upper bound on
$\mathbb{E}[(\xi ^{h}_{n}-\xi ^{h}_{\star})^{2}]$.}}\ Combining either
\eqref{estimate:Lypanuov:1:gamma:1/n} with
\eqref{estimate:Lypanuov:2:gamma:1/n} or
\eqref{estimate:Lypanuov:1:gamma:1/nbeta} with
\eqref{estimate:Lypanuov:1:gamma:1/nbeta:bis} and then invoking Lemma~\ref{lmm:lyapunov}(\ref{lmm:lyapunov-iv})
with $q=1$, we conclude that
\begin{equation*}
\mathbb{E}[(\xi ^{h}_{n}-\xi ^{h}_{\star})^{2}]\leq \bar{K}^{\beta}_{h,2}
\gamma _{n},
\end{equation*}
where
$\bar{K}^{\beta}_{h,2}:=\kappa _{h,1}(K_{h,2}^{\beta}+K_{h,1}^{\beta})$
satisfies $\sup _{h\in \mathcal{H}}\bar{K}^{\beta}_{h,2}<\infty $ for any
$\beta \in (0,1]$. This follows because for $q=1,2$,
\begin{equation*}
\sup _{h \in \overline{\mathcal{H}}}\mathbb{E}[\bar{\mathcal{L}}_{h,q}(
\xi ^{h}_{0})]\leq \sup _{h\in \overline{\mathcal{H}}}{\mathbb{E}
\bigg[(1+|\xi ^{h}_{0}|^{2})\exp \bigg(\frac{2}{1-\alpha}k_{\alpha }
\sup _{h\in \overline{\mathcal{H}}}\|f_{X_{h}}\|_{\infty}|\xi _{0}^{h}|
\bigg)\bigg]}<\infty ,
\end{equation*}
recalling Lemma~\ref{lmm:lyapunov}(\ref{lmm:lyapunov-i-bis}).
\\

\noindent\textbf{\emph{Step~5. Upper bound on
$\mathbb{E}[|\chi ^{h}_{n}-\chi ^{h}_{\star}|]$.}}\
We now prove an
$L^{1}(\mathbb{P})$-upper estimate for the difference
$\chi ^{h}_{n}-\chi ^{h}_{\star}$, $n\geq 1$, with
$(\chi ^{h}_{n})_{n\geq 0}$ given by
\eqref{approximate:sgd:algorithm}. Without loss of generality, we assume
$\chi _{0}^{h}=0$; the general case $\chi _{0}^{h}\neq 0$ is handled similarly.
Observe~that
\begin{align}
\label{decomposition:stat:error:cvar:sa:algorithm}
\chi ^{h}_{n}-\chi ^{h}_{\star }&=\frac{1}{n}\sum _{k=1}^{n}\bigg(
\xi ^{h}_{k-1}+\frac{1}{1-\alpha}(X^{(k)}_{h}-\xi ^{h}_{k-1})^{+}
\bigg)-V_{h}(\xi ^{h}_{\star})
\nonumber
\\
&=\frac{1}{n}\sum _{k=1}^{n}\varepsilon _{k}^{h}+\frac{1}{n}\sum _{k=1}^{n}V_{h}(
\xi _{k-1}^{h})-V_{h}(\xi ^{h}_{\star}),
\end{align}
where
\begin{equation}
\label{epsilon}
\varepsilon _{k}^{h}:=\xi ^{h}_{k-1}+\frac{1}{1-\alpha}(X^{(k)}_{h}-
\xi ^{h}_{k-1})^{+}-V_{h}(\xi _{k-1}^{h}),\qquad k\geq 1,
\end{equation}
is a sequence of $(\mathbb{F}^{h},\mathbb{P})$-martingale increments, that
is, $\mathbb{E}[\varepsilon _{k}^{h}\,|\,\mathcal{F}^{h}_{k-1}]=0$. Note
that
\begin{align*}
&\mathbb{E}\big[|\varepsilon _{k}^{h}|^{2}\,\big|\,\mathcal{F}^{h}_{k-1}
\big]
\nonumber
\\
&\leq \frac{1}{(1-\alpha )^{2}}\mathbb{E}\big[\big((X^{(k)}_{h}-\xi ^{h}_{k-1})^{+}-
\mathbb{E}[(X^{(k)}_{h}-\xi ^{h}_{k-1})^{+}\,|\,\mathcal{F}^{h}_{k-1}]
\big)^{2}\,\big|\,\mathcal{F}^{h}_{k-1}\big]
\nonumber
\\
&\leq \frac{1}{(1-\alpha )^{2}}\mathbb{E}\big[\big((X^{(k)}_{h}-\xi ^{h}_{k-1})^{+}
\big)^{2}\,\big|\,\mathcal{F}^{h}_{k-1}\big]
\nonumber
\\
&\leq \frac{3}{(1-\alpha )^{2}}\big(\mathbb{E}[|X_{h}|^{2}]+(\xi _{k-1}^{h}-
\xi ^{h}_{\star})^{2}+|\xi ^{h}_{\star}|^{2}\big).
\end{align*}
Using that estimate, the Cauchy--Schwarz inequality and
\eqref{uniform:L2:bound:var:alg} gives
\begin{align*}
&\frac{1}{n}\mathbb{E}\bigg[\bigg|\sum _{k=1}^{n}\varepsilon _{k}^{h}
\bigg|\bigg]
\\
&\leq \frac{1}{n}\bigg(\sum _{k=1}^{n}\mathbb{E}[|\varepsilon _{k}^{h}|^{2}]
\bigg)^{\frac{1}{2}}
\\
&\leq \frac{\sqrt{3}}{(1-\alpha )n}\bigg(n\Big(\sup _{h\in
\mathcal{H}}{\mathbb{E}[|X_{h}|^{2}]}+\sup _{h\in \mathcal{H}}{|\xi ^{h}_{
\star}|^{2}}\Big)+\bar{K}_{h}^{\beta}\sum _{k=2}^{n}\gamma _{k-1}+
\mathbb{E}[(\xi ^{h}_{0}-\xi ^{h}_{\star})^{2}]\bigg)^{\frac{1}{2}}
\\
&\leq \frac{\sqrt{3}}{(1-\alpha )n^{\frac{1}{2}}}\bigg( \Big(\sup _{h
\in \mathcal{H}}\mathbb{E}[|X_{h}|^{2}]^{\frac{1}{2}}+\sup _{h\in
\mathcal{H}}|\xi ^{h}_{\star}|\Big) +(\gamma _{1}\bar{K}_{h}^{\beta})^{
\frac{1}{2}}\frac{1}{(1-\beta )^{\frac{1}{2}}n^{\frac{\beta}{2}}}
\mathds1_{\{\beta \in (0,1)\}}
\\
&
\hphantom{=:\frac{\sqrt{3}}{(1-\alpha )n^{\frac{1}{2}}}\bigg(}
+(\gamma _{1}\bar{K}_{h}^{\beta})^{\frac{1}{2}}
\frac{(\ln{n})^{\frac{1}{2}}}{n^{\frac{1}{2}}}\mathds1_{\{\beta =1\}} +
\frac{\mathbb{E}[(\xi ^{h}_{0}-\xi ^{h}_{\star})^{2}]^{\frac{1}{2}}}{n^{\frac{1}{2}}}
\bigg).
\end{align*}
For the second term appearing on the right-hand side of
\eqref{decomposition:stat:error:cvar:sa:algorithm}, we use a second-order
Taylor expansion,
$\sup _{h\in \mathcal{H}}{\|V_{h}''\|_{\infty}}<\infty $ by Assumption~\ref{asp:misc}(\ref{asp:misc-iv})
and \eqref{uniform:L2:bound:var:alg}. Thus we get
\begin{align*}
&\frac{1}{n}\mathbb{E}\bigg[\bigg|\sum _{k=1}^{n}V_{h}(\xi _{k-1}^{h})-V_{h}(
\xi ^{h}_{\star})\bigg|\bigg]
\\
&\leq \frac{\|V_{h}''\|_{\infty}}{2n}\sum _{k=1}^{n}\mathbb{E}[(\xi ^{h}_{k-1}-
\xi ^{h}_{\star})^{2}]
\\
&\leq \frac{\|V_{h}''\|_{\infty}}{2}\bigg(
\frac{\mathbb{E}[(\xi ^{h}_{0}-\xi ^{h}_{\star})^{2}]}{n}+\gamma _{1}K_{h}^{
\beta}\Big(\frac{1}{(1-\beta )n^{\beta}}\mathds1_{\{\beta \in (0,1)\}}+
\frac{\ln{n}}{n}\mathds1_{\{\beta =1\}}\Big)\bigg).
\end{align*}
Inserting the two previous estimates into
\eqref{decomposition:stat:error:cvar:sa:algorithm} yields an upper bound
on $\mathbb{E}[|\chi ^{h}_{n}-\chi ^{h}_{\star}|]$.
\\

\noindent\textbf{\emph{Step~6. Second inequality on
$\mathbb{E}[\bar{\mathcal{L}}_{h,2}(\xi ^{h}_{n})]$.}}\
The proof of
\eqref{uniform:L4:bound:var:alg} relies on similar arguments to those employed
for \eqref{uniform:L2:bound:var:alg}. For the sake of brevity, we omit
some technical details. Using either
\eqref{sum:square:gamma:times:prod:final:case:1/n}, recalling that
$\bar\lambda _{h,2}\gamma _{1}>2$ and $\lambda^{\mu_{h,2}}_{h,1}\gamma_1\geq\bar\lambda_{h,1}\gamma_1=\frac{\bar\lambda_{h,2}\gamma_1}2>1$, or
\eqref{sum:square:gamma:times:prod:final:bis}, we get
\begin{equation}
\label{eq:E[L2(theta)]<final:bound}
\mathbb{E}[\bar{\mathcal{L}}_{h,2}(\xi ^{h}_{n})]\leq \widetilde{K}_{h,2}^{
\beta}\gamma _{n}^{2},
\end{equation}
for some constants
$(\widetilde{K}_{h,2}^{\beta})_{h\in \mathcal{H}}$ satisfying
$\sup _{h\in \mathcal{H}}\widetilde{K}_{h,2}^{\beta}<\infty $.

Similarly, using $q=2$ and $\mu=\mu_{h,4}$ in \eqref{eq:ELq<}, given that $\lambda^{\mu_{h,4}}_{h,2}\gamma_1\geq\bar\lambda_{h,2}\gamma_1>2$ and $\lambda^{\mu_{h,4}}_{h,1}\gamma_1\geq\bar\lambda_{h,1}\gamma_1>1$ if $\beta=1$, we obtain
\begin{equation}
\label{eq:E[L2(theta)]<final:bound:check}
\mathbb{E}[\mathcal{L}_{h,2}(\xi ^{h}_{n})]\leq \check{K}_{h,2}^{
\beta}\gamma _{n}^{2},
\end{equation}
for some constants
$(\check{K}_{h,2}^{\beta})_{h\in \mathcal{H}}$ satisfying
$\sup _{h\in \mathcal{H}}\check{K}_{h,2}^{\beta}<\infty $.
\\

\noindent\textbf{\emph{Step~7. Inequality on
$\mathbb{E}[\mathcal{L}_{h,3}(\xi ^{h}_{n})]$.}}\
Setting $q=3$ and
$\mu =\mu _{h,4}$ in \eqref{eq:ELq<} and utilising
\eqref{eq:E[L2(theta)]<final:bound} gives
\begin{equation*}
\begin{aligned}
&\mathbb{E}[\mathcal{L}_{h,3}(\xi ^{h}_{n+1})]
\\
&\leq \mathbb{E}[\mathcal{L}_{h,3}(\xi ^{h}_{n})](1-
\lambda _{h,3}^{\mu_{h,4}}\gamma _{n+1}
+\zeta _{h,3}^{\mu_{h,4}}\gamma _{n+1}^{2})
+\zeta _{h,3}^{\mu_{h,4}}\gamma _{n+1}^{2}
\mathbb{E}[\mathcal{L}_{h,2}(\xi ^{h}_{n})]
+\zeta _{h,3}^{\mu_{h,4}}\gamma _{n+1}^{4}
\\
&\leq \mathbb{E}[\mathcal{L}_{h,3}(\xi ^{h}_{n})]
\big(1-\lambda _{h,3}^{\mu_{h,4}}\gamma _{n+1}
+(\mathrm{e} \mu _{h,4}^{3}+1)
\zeta _{h,3}^{\mu_{h,4}}\gamma _{n+1}^{2}\big)
+(2^{2\beta}\check{K}_{h,2}^{\beta}+1)
\zeta _{h,3}^{\mu_{h,4}}\gamma _{n+1}^{4},
\end{aligned}
\end{equation*}
so that
\begin{equation*}
\mathbb{E}[\bar{\mathcal{L}}_{h,3}(\xi ^{h}_{n})]
\leq \mathbb{E}[\bar{\mathcal{L}}_{h,3}(\xi ^{h}_{0})]
\Pi ^{h,3,4}_{1:n}
+(2^{2\beta}\check{K}_{h,2}^{\beta}+1)
\bar\zeta _{h,3}
\sum _{k=1}^{n}\gamma _{k}^{4}\Pi ^{h,3,4}_{k+1:n}.
\end{equation*}
Following similar lines of reasoning as those used in Step~2, we conclude
that if $\gamma _{n}=\gamma _{1}n^{-\beta}$, $\beta \in (0,1]$, with
$\lambda _{h,3}^{\mu_{h,4}}\gamma _{1}
\geq\bar\lambda _{h,3}\gamma _{1}
\geq\bar\lambda _{h,2}\gamma _{1}>2$ if $\beta =1$, then
\begin{equation}
\label{eq:E[L3(theta)]<final:bound}
\mathbb{E}[\bar{\mathcal{L}}_{h,3}(\xi ^{h}_{n})]\leq K_{h,3}^{\beta}
\gamma _{n}^{2}
\end{equation}
for some constants $(K^{\beta}_{h,3})_{h\in \mathcal{H}}$ satisfying
$\sup _{h\in \mathcal{H}}{K^{\beta}_{h,3}}<\infty $.
\\

\noindent\textbf{\emph{Step~8. Inequality on
$\mathbb{E}[\bar{\mathcal{L}}_{h,4}(\xi ^{h}_{n})]$.}}\
Finally, we take
$q=4$ and $\mu =\mu _{h,4}$ in~\eqref{eq:ELq<} and use
\eqref{eq:E[L3(theta)]<final:bound} to obtain
\begin{equation*}
\begin{aligned}
&\mathbb{E}[\bar{\mathcal{L}}_{h,4}(\xi ^{h}_{n+1})]
\\
&\leq \mathbb{E}[\bar{\mathcal{L}}_{h,4}(\xi ^{h}_{n})](1-
\bar\lambda _{h,4}\gamma _{n+1}+\zeta _{h,4}\gamma _{n+1}^{2}) +
\bar\zeta _{h,4}\gamma _{n+1}^{2}\mathbb{E}[\mathcal{L}_{h,3}(\xi ^{h}_{n})]
+\bar\zeta _{h,3}\gamma _{n+1}^{5}
\\
&\leq \mathbb{E}[\bar{\mathcal{L}}_{h,4}(\xi ^{h}_{n})]\big(1-
\bar\lambda _{h,4}\gamma _{n+1}+(\mathrm{e} \mu _{h,4}^{4}+1)\bar\zeta _{h,4}
\gamma _{n+1}^{2}\big)
+(2^{2\beta}K_{h,3}^{\beta}+\gamma _{1})\bar\zeta _{h,4}
\gamma _{n+1}^{4},
\end{aligned}
\end{equation*}
so that
\begin{equation}
\label{neugl4}
\mathbb{E}[\bar{\mathcal{L}}_{h,4,4}(\xi ^{h}_{n})]\leq \mathbb{E}[
\bar{\mathcal{L}}_{h,4}(\xi ^{h}_{0})]\Pi ^{h,4}_{1:n}+(2^{2\beta}K_{h,3}^{\beta}+\gamma _{1})\bar\zeta _{h,4}\sum _{k=1}^{n}
\gamma _{k}^{4}\Pi ^{h,4,4}_{k+1:n}.
\end{equation}
Analogously, we deduce from \eqref{neugl4} that if
$\gamma _{n}=\gamma _{1}n^{-\beta}$, $\beta \in (0,1]$, with
$\bar\lambda _{h,4}\gamma _{1}\geq\bar\lambda _{h,2}\gamma _{1}>2$ if $\beta =1$, then for any
$h\in \mathcal{H}$ and any positive integer $n$, we have
\begin{equation}
\label{eq:E[L4(theta)]<final:bound}
\mathbb{E}[\bar{\mathcal{L}}_{h,4}(\xi ^{h}_{n})]\leq K_{h,4}^{\beta}
\gamma _{n}^{2},
\end{equation}
where $(K_{h,4}^{\beta})_{h\in \mathcal{H}}$ are constants that satisfy
$\sup _{h\in \mathcal{H}}{K_{h,4}^{\beta}}<\infty $.
\\

\noindent\textbf{\emph{Step~9. Upper bound on
$\mathbb{E}[(\xi ^{h}_{n}-\xi ^{h}_{\star})^{4}]$.}}\
By
\eqref{eq:E[L2(theta)]<final:bound},
\eqref{eq:E[L4(theta)]<final:bound} and Lemma~\ref{lmm:lyapunov}(\ref{lmm:lyapunov-iv})
with $q=2$, recalling that $\bar\lambda _{h,4}>\bar\lambda _{h,2}$, we
conclude that if $\gamma _{n}=\gamma _{1}n^{-\beta}$,
$\beta \in (0,1]$, with $\bar\lambda _{h,2}\gamma _{1}>2$ if
$\beta =1$, then for any $h\in \mathcal{H}$ and any positive integer
$n$, we have
\begin{equation*}
\mathbb{E}[(\xi ^{h}_{n}-\xi ^{h}_{\star})^{4}]\leq \bar{K}^{\beta}_{h,4}
\gamma ^{2}_{n},
\end{equation*}
where
$\bar{K}^{\beta}_{h,4}:=\kappa _{h,2}(K_{h,4}^{\beta}+\widetilde{K}^{
\beta}_{h,2})$ satisfies
$\sup _{h\in \mathcal{H}}\bar{K}^{\beta}_{h,4}<\infty $,
$\beta \in (0,1]$, as Lemma~\ref{lmm:lyapunov}(\ref{lmm:lyapunov-i-bis})
implies that for $q=1,2,3,4$ and $\mu=\mu_{h,4}$,
\begin{equation*}
\sup _{h \in \overline{\mathcal{H}}}\mathbb{E}[\mathcal{L}_{h,q}(
\xi ^{h}_{0})]\leq \sup _{h\in \overline{\mathcal{H}}}{\mathbb{E}
\bigg[(1+|\xi ^{h}_{0}|^{4})\exp \bigg(\frac{4}{1-\alpha}k_{\alpha }
\sup _{h\in \overline{\mathcal{H}}}\|f_{X_{h}}\|_{\infty}|\xi _{0}^{h}|
\bigg)\bigg]}<\infty .
\end{equation*}

\noindent\textbf{\emph{Step~10. Upper bound on
$\mathbb{E}[(\chi ^{h}_{n}-\chi ^{h}_{\star})^{2}]$.}}\
Assuming here again
that $\chi _{0}^{h}=0$, we get from the decomposition
\eqref{decomposition:stat:error:cvar:sa:algorithm} and similar computations
to the ones performed in Step~5 that for some constants
$(\bar{C}^{\beta}_{h})_{h\in \mathcal{H}}$ that may change from line to
line, we have
\begin{align*}
\mathbb{E}[(\chi ^{h}_{n}-\chi ^{h}_{\star})^{2}] &\leq \bar{C}^{
\beta}_{h}\bigg(\frac{1}{n^{2}}\sum _{k=1}^{n}\mathbb{E}[|
\varepsilon _{k}^{h}|^{2}]+\frac{1}{n^{2}}\Big(\sum _{k=1}^{n}
\mathbb{E}\big[\big(V_{h}(\xi _{k-1}^{h})-V_{h}(\xi ^{h}_{\star})
\big)^{2}\big]^{\frac{1}{2}}\Big)^{2}\bigg)
\\
&\leq \bar{C}^{\beta}_{h}\bigg(\frac{1}{n}+\frac{1}{n^{2}}\Big(\sum _{k=1}^{n}
\mathbb{E}[(\xi _{k-1}^{h}-\xi ^{h}_{\star})^{4}]^{\frac{1}{2}}\Big)^{2}
\bigg)
\\
&\leq \bar{C}^{\beta}_{h}\bigg(\frac{1}{n}+\frac{1}{n^{2}}\Big(\sum _{k=1}^{n}
\gamma _{k}\Big)^{2}\bigg) \leq
\frac{\bar{C}^{\beta}_{h}}{n^{1\wedge 2\beta}},
\end{align*}
where the last inequality follows by a comparison between series and integrals
for $\gamma _{n}=\gamma _{1}n^{-\beta}$, $\beta \in (0,1]$. This concludes
the proof.
\end{proof}

\section{Proof of Lemma~\ref{lmm:lyapunov}}
\label{prf:lyapunov}

Throughout, we consider $q\in \mathbb{N}$, $h\in \overline{\mathcal{H}}$ and $\mu\geq0$.
We drop the superscript $\mu$ from our notation and write $\mathcal{L}_{h,q}$ for $\mathcal{L}^{\mu}_{h,q}$.
\\

\noindent(\ref{lmm:lyapunov-i})\
From the definition \eqref{eq:Lq},
$\mathcal{L}_{h,q}$ is continuously differentiable and satisfies
\begin{align}
\label{eq:Lq'-0}
\mathcal{L}_{h,q}'(\xi )
&=qV_{h}'(\xi )\mathcal{L}_{h,q-1}(\xi )+ \mu V_{h}'(\xi )\mathcal{L}_{h,q}(\xi )
\nonumber
\\
&=\mathcal{L}_{h,q-1}(\xi )V_{h}'(\xi )\Big(q+ \mu\big(V_{h}(\xi )-V_{h}(\xi ^{h}_{\star})\big)\Big).
\end{align}
One may differentiate again the above identity. If $q\geq 2$,
\begin{align}
\label{eq:Lq''-0}
&\mathcal{L}_{h,q}''(\xi )
\nonumber
\\
&=\mathcal{L}_{h,q-2}(\xi )\big(V_{h}'(\xi )\big)^{2}
\Big(q-1+ \mu\big(V_{h}(\xi )-V_{h}(\xi ^{h}_{\star})\big)\Big)
\Big(q+ \mu\big(V_{h}(\xi )-V_{h}(\xi ^{h}_{\star})\big)\Big)
\nonumber
\\
&
\hphantom{=:}
+\mathcal{L}_{h,q-1}(\xi )V_{h}''(\xi )\Big(q+ \mu\big(V_{h}(\xi )-V_{h}(\xi ^{h}_{\star})\big)\Big)
+ \mathcal{L}_{h,q-1}(\xi ) \mu \big(V_{h}'(\xi )\big)^{2}.
\end{align}
If $q=1$,
\begin{equation}
\label{eq:L1''}
\mathcal{L}_{h,1}''(\xi )
=\Big(V_{h}''(\xi ) + 2 \mu \big(V_{h}'(\xi )\big)^{2}\Big)\mathcal{L}_{h,0}(\xi )
+\Big(\mu^2 \big(V_{h}'(\xi )\big)^{2}+ \mu V_{h}''(\xi )\Big)\mathcal{L}_{h,1}(\xi ).
\end{equation}
All in all, for all $q\geq 1$, $\mathcal{L}_{h,q}$ is twice continuously
differentiable.
\\

\noindent(\ref{lmm:lyapunov-i-bis})\
The upper estimate follows directly from the fact
that for any $\xi \in \mathbb{R}$,
\begin{equation*}
V_{h}(\xi )-V_{h}(\xi ^{h}_{\star}) =(\xi -\xi ^{h}_{\star})\int _{0}^{1}V_{h}'
\big(\xi ^{h}_{\star}+t(\xi -\xi ^{h}_{\star})\big)\mathrm{d}t.
\end{equation*}
Hence $\|V_{h}'\|_{\infty}\leq k_{\alpha}$ implies
$V_{h}(\xi ) - V_{h}(\xi ^{h}_{\star})\leq k_{\alpha }|\xi -\xi ^{h}_{\star}|$.

The remaining properties of Lemma~\ref{lmm:lyapunov} are trivially satisfied
for $\xi =\xi ^{h}_{\star}$. We thus assume without loss of generality
that $\xi \neq \xi ^{h}_{\star}$. We let
\begin{equation}
\label{eq:epsh}
\varepsilon _{h}:=
\frac{V_{h}''(\xi ^{h}_{\star})}{[V_{h}'']_{{\mathrm{Lip}}}}.
\end{equation}
Finally, set
$\mathcal{I}^{h}:=[\xi ^{h}_{\star}-\varepsilon _{h},\xi ^{h}_{\star}+
\varepsilon _{h}]$.
\\

\noindent(\ref{lmm:lyapunov-ii})\
\noindent\textbf{\emph{Step~1. Preliminaries.}}\
From
\eqref{eq:Lq'-0} and the relation
\begin{equation*}
\mathcal{L}_{h,q-1}(\xi )=\big(V_{h}(\xi )-V_{h}(\xi ^{h}_{\star})
\big)^{-1}\mathcal{L}_{h,q}(\xi ),
\end{equation*}
we have
\begin{equation*}
\mathcal{L}_{h,q}'(\xi )
=\mathcal{L}_{h,q}(\xi )V_{h}'(\xi )\bigg(\frac{q}{V_{h}(\xi )-V_{h}(\xi ^{h}_{\star})}+\mu \bigg)
\end{equation*}
so that
\begin{equation}
\label{eq:0}
\mathcal{L}_{h,q}'(\xi )V_{h}'(\xi )
\geq \bigg(\frac{q(V_{h}'(\xi ))^{2}}{V_{h}(\xi )-V_{h}(\xi ^{h}_{\star})}
\vee\mu\big(V_{h}'(\xi )\big)^{2}\bigg)\mathcal{L}_{h,q}(\xi ).
\end{equation}
Let us establish a lower bound for the parenthesised factor above. Since
$V_{h}'(\xi ^{h}_{\star})=0$, a first-order Taylor expansion gives
\begin{eqnarray}
\label{eq:Vh':taylor}
V_{h}'(\xi ) &=&(\xi -\xi ^{h}_{\star})\int _{0}^{1}V_{h}''\big(\xi ^{h}_{
\star}+t(\xi -\xi ^{h}_{\star})\big)\mathrm{d}t
\nonumber
\\
&=&V_{h}''(\xi ^{h}_{\star})(\xi -\xi ^{h}_{\star}) +(\xi -\xi ^{h}_{
\star})\int _{0}^{1}\!\!\Big(V_{h}''\big(\xi ^{h}_{\star}+t(\xi -\xi ^{h}_{
\star})\big)-V_{h}''(\xi ^{h}_{\star})\Big)\mathrm{d}t.
\end{eqnarray}
Combined with the fact that
$V_{h}''=(1-\alpha )^{-1}f_{X_{h}}\geq 0$ and the triangle inequality,
this yields
\begin{equation}
\label{eq:Vh'lowerbound}
|V_{h}'(\xi )| \geq V_{h}''(\xi ^{h}_{\star})|\xi -\xi ^{h}_{\star}|-
\frac{1}{2}[V_{h}'']_{{\mathrm{Lip}}}(\xi -\xi ^{h}_{\star})^{2}.
\end{equation}

\noindent\textbf{\emph{Step~2. Case where $\xi \in \mathcal{I}^{h}$.}}\
If
$\xi \in \mathcal{I}^{h}$, the above inequality yields
\begin{equation}
\label{eq:1}
\big(V_{h}'(\xi )\big)^{2}\geq \frac{1}{4}\big(V_{h}''(\xi ^{h}_{
\star})\big)^{2}(\xi -\xi ^{h}_{\star})^{2}.
\end{equation}
Now, recalling that $V_{h}'(\xi ^{h}_{\star})=0$, a second-order Taylor
expansion gives
\begin{eqnarray}
\label{eq:Vhpseudotaylor}
V_{h}(\xi )-V_{h}(\xi ^{h}_{\star}) &=&\frac{1}{2}V_{h}''(\xi ^{h}_{
\star})(\xi -\xi ^{h}_{\star})^{2}
\nonumber
\\
&&+(\xi -\xi ^{h}_{\star})^{2}\int _{0}^{1}\!\!(1-t)\Big(V_{h}''\big(
\xi ^{h}_{\star}+t(\xi -\xi ^{h}_{\star})\big)-V_{h}''(\xi ^{h}_{
\star})\Big)\mathrm{d}t
\end{eqnarray}
so that
\begin{equation*}
V_{h}(\xi )-V_{h}(\xi ^{h}_{\star}) \leq \frac{1}{2}V_{h}''(\xi ^{h}_{
\star})(\xi -\xi ^{h}_{\star})^{2} +\frac{1}{6}[V_{h}'']_{
{\mathrm{Lip}}}|\xi - \xi ^{h}_{\star}|^{3}.
\end{equation*}
Hence if $\xi \in \mathcal{I}^{h}$, the above inequality implies
\begin{equation}
\label{eq:2}
V_{h}(\xi )-V_{h}(\xi ^{h}_{\star})\leq \frac{2}{3}V_{h}''(\xi ^{h}_{
\star})(\xi -\xi ^{h}_{\star})^{2}.
\end{equation}
Plug \eqref{eq:1} and \eqref{eq:2} into \eqref{eq:0} and recall that
$V_{h}''(\xi ^{h}_{\star})=(1-\alpha )^{-1}f_{X_{h}}(\xi ^{h}_{\star})>0$
by Assumption~\ref{asp:misc}(\ref{asp:misc-iii}) to get
\begin{equation}
\label{eq:Lq'Vh'>insideIh}
\mathcal{L}_{h,q}'(\xi )V_{h}'(\xi ) \geq
\frac{q(V_{h}'(\xi ))^{2}}{V_{h}(\xi )-V_{h}(\xi ^{h}_{\star})}
\mathcal{L}_{h,q}(\xi ) \geq \frac{3}{8}qV_{h}''(\xi ^{h}_{\star})
\mathcal{L}_{h,q}(\xi ).
\end{equation}

\noindent\textbf{\emph{Step~3. Case where $\xi \notin \mathcal{I}^{h}$.}}\
Assume now
$\xi \notin \mathcal{I}^{h}$, say
$\xi >\xi ^{h}_{\star}+\varepsilon _{h}$. The case
$\xi <\xi ^{h}_{\star}-\varepsilon _{h}$ is similar and is omitted. Since
$V_{h}'$ is nondecreasing, \text{${V_{h}'(\xi ^{h}_{\star}+\varepsilon _{h})
\geq V_{h}'(\xi ^{h}_{\star})=0}$,} and evaluating
\eqref{eq:Vh'lowerbound} at $\xi ^{h}_{\star}+\varepsilon _{h}$ and recalling
the definition~\eqref{eq:epsh} gives
\begin{equation*}
V_{h}'(\xi )\geq V_{h}'(\xi ^{h}_{\star}+\varepsilon _{h}) \geq V_{h}''(
\xi ^{h}_{\star})\varepsilon _{h}-\frac{1}{2}[V_{h}'']_{
{\mathrm{Lip}}}\varepsilon _{h}^{2} =
\frac{(V_{h}''(\xi ^{h}_{\star}))^{2}}{2[V_{h}'']_{{\mathrm{Lip}}}}.
\end{equation*}
Plugging this into \eqref{eq:0}, we obtain for any
$\xi >\xi ^{h}_{\star}+\varepsilon _{h}$ that
\begin{equation}
\label{eq:Lq'Vh'>outsideIh}
\mathcal{L}_{h,q}'(\xi )V_{h}'(\xi )
\geq \mu \mathcal{L}_{h,q}(\xi )\big(V_{h}'(\xi )\big)^{2}
\geq \mu \frac{(V_{h}''(\xi ^{h}_{\star}))^{4}}{4[V_{h}'']_{{\mathrm{Lip}}}^{2}} \mathcal{L}_{h,q}(\xi ).
\end{equation}
The above inequality is valid also if
$\xi <\xi ^{h}_{\star}-\varepsilon _{h}$.
\\

\noindent\textbf{\emph{Step~4. Conclusion.}}\
Combining
\eqref{eq:Lq'Vh'>insideIh} and \eqref{eq:Lq'Vh'>outsideIh}, we obtain
\begin{equation*}
\mathcal{L}_{h,q}'(\xi )V_{h}'(\xi )
\geq \bigg(\frac{3}{8}qV_{h}''(\xi ^{h}_{\star})
\wedge \mu \frac{(V_{h}''(\xi ^{h}_{\star}))^{4}}{4[V_{h}'']_{{\mathrm{Lip}}}^{2}}\bigg)
\mathcal{L}_{h,q}(\xi )
=\lambda^{\mu}_{h,q}\mathcal{L}_{h,q}(\xi ),
\qquad \xi \in \mathbb{R}.
\end{equation*}
By Lemma~\ref{lmm:thetah->theta0}, there exists $R>0$ with
$\xi ^{h}_{\star}\in B(\xi ^{0}_{\star},R)$, $h\in \mathcal{H}$. By Assumption~\ref{asp:misc},
\begin{align*}
\inf _{h\in \overline{\mathcal{H}}}{V_{h}''(\xi ^{h}_{\star})} &\geq
\frac{1}{1-\alpha}\inf _{
\substack{h\in \overline{\mathcal{H}}\\\xi \in B(\xi ^{0}_{\star},R)}}{f_{X_{h}}(
\xi )} >0,
\\
\inf _{h\in \overline{\mathcal{H}}}{
\frac{1}{[V_{h}'']_{{\mathrm{Lip}}}}} &=
\frac{1-\alpha}{\sup _{h\in \overline{\mathcal{H}}}{[f_{X_{h}}]_{{\mathrm{Lip}}}}}>0,
\end{align*}
and by Assumption~\ref{asp:misc}(\ref{asp:misc-iii}),
$\inf _{h\in \overline{\mathcal{H}}}\|V_{h}''\|_{\infty}
=(1-\alpha )^{-1}\inf _{h\in \overline{\mathcal{H}}}\|f_{X_{h}}\|_{\infty}>0$.
Thus, for any $q'\geq q$, $\inf _{h\in \overline{\mathcal{H}}}\lambda^{h,q'}_{h,q}>0$, and in particular, $\inf _{h\in \overline{\mathcal{H}}}\bar{\lambda}_{h,q}>0$.
\\

\noindent(\ref{lmm:lyapunov-iii})\
\textbf{\emph{Step~1. Case where $q\geq 2$.}}\
Assume first that $q\geq 2$. Through \eqref{eq:Lq''-0} and the relation
\begin{equation*}
\mathcal{L}_{h,q-1}(\xi )=\big(V_{h}(\xi )-V_{h}(\xi ^{h}_{\star})
\big)^{-1}\mathcal{L}_{h,q}(\xi ),
\end{equation*}
one has
\begin{align}
\label{eq:Lq''}
\mathcal{L}_{h,q}''(\xi )
&=\mathcal{L}_{h,q-1}(\xi )\bigg(qV_{h}''(\xi )+\frac{q(q-1)(V_{h}'(\xi ))^{2}}{V_{h}(\xi )-V_{h}(\xi ^{h}_{\star})}\bigg)
\nonumber
\\
&
\hphantom{=:}
+\mathcal{L}_{h,q}(\xi )\bigg( \mu V_{h}''(\xi )+ \mu^2 \big(V_{h}'(\xi )\big)^{2}
+2q\mu\frac{(V_{h}'(\xi ))^{2}}{V_{h}(\xi )-V_{h}(\xi ^{h}_{\star})}\bigg).
\end{align}

\noindent\textbf{\emph{Step~1.1. Sub-case where $\xi \in \mathcal{I}^{h}$.}}\
Applying
the triangle inequality to \eqref{eq:Vhpseudotaylor} yields
\begin{equation}
\label{13util}
V_{h}(\xi )-V_{h}(\xi ^{h}_{\star}) \geq \frac{1}{2}V_{h}''(\xi ^{h}_{
\star})(\xi -\xi ^{h}_{\star})^{2}-\frac{1}{6}[V_{h}'']_{
{\mathrm{Lip}}}|\xi -\xi ^{h}_{\star}|^{3}.
\end{equation}
Hence if $\xi \in \mathcal{I}^{h}$, we have
$|\xi -\xi ^{h}_{\star}|\leq \varepsilon _{h}$ so that the above inequality
and \eqref{eq:epsh} yield
\begin{equation}
\label{eq:Vh-Vh>onIh}
V_{h}(\xi )-V_{h}(\xi ^{h}_{\star})\geq \frac{1}{3}V_{h}''(\xi ^{h}_{
\star})(\xi -\xi ^{h}_{\star})^{2}.
\end{equation}
Moreover, using \eqref{eq:Vh':taylor} and
$\|V_{h}''\|_{\infty}=(1-\alpha )^{-1}\|f_{X_{h}}\|_{\infty}<\infty $ by
Assumption~\ref{asp:misc}(\ref{asp:misc-iv}), we get
\begin{equation*}
\big(V_{h}'(\xi )\big)^{2}\leq \|V_{h}''\|_{\infty}^{2}(\xi -\xi ^{h}_{
\star})^{2}.
\end{equation*}
Combining the previous two inequalities, if
$\xi \in \mathcal{I}^{h}$, we have
\begin{equation}
\label{eq:frac<1}
\frac{(V_{h}'(\xi ))^{2}}{V_{h}(\xi )-V_{h}(\xi ^{h}_{\star})} \leq
\frac{3\|V_{h}''\|_{\infty}^{2}}{V_{h}''(\xi ^{h}_{\star})}.
\end{equation}

\noindent\textbf{\emph{Step~1.2. Sub-case where $\xi \notin \mathcal{I}^{h}$.}}\
Assume
now $\xi \notin \mathcal{I}^{h}$, say
$\xi >\xi ^{h}_{\star}+\varepsilon _{h}$. The case
$\xi <\xi ^{h}_{\star}-\varepsilon _{h}$ is treated similarly and is thus
omitted. Using that $V_{h}$ is increasing on~$[\xi ^{h}_{\star},
\infty )$, evaluating \eqref{13util} at
$\xi ^{h}_{\star}+\varepsilon _{h}$ and recalling \eqref{eq:epsh}, we obtain
\begin{equation*}
V_{h}(\xi )- V_{h}(\xi ^{h}_{\star}) \geq V_{h}(\xi ^{h}_{\star}+
\varepsilon _{h})-V_{h}(\xi ^{h}_{\star}) \geq
\frac{(V_{h}''(\xi ^{h}_{\star}))^{3}}{3[V_{h}'']_{{\mathrm{Lip}}}^{2}}.
\end{equation*}
This inequality together with the fact that
$\|V_{h}'\|_{\infty}\leq k_{\alpha}$ gives for any
$\xi >\xi ^{h}_{\star}+\varepsilon _{h}$ that
\begin{equation}
\label{eq:frac<2}
\frac{(V_{h}'(\xi ))^{2}}{V_{h}(\xi )-V_{h}(\xi ^{h}_{\star})} \leq
\frac{3k_{\alpha}^{2}[V_{h}'']_{{\mathrm{Lip}}}^{2}}{(V_{h}''(\xi ^{h}_{\star}))^{3}}.
\end{equation}
The above estimate holds, too, if
$\xi <\xi ^{h}_{\star}+\varepsilon _{h}$.
\\

\noindent\textbf{\emph{Step~1.3. Sub-conclusion.}}\
\eqref{eq:frac<1} and
\eqref{eq:frac<2} show that for any $\xi \neq \xi ^{h}_{\star}$,
\begin{equation}
\label{eq:nuh}
\frac{(V_{h}'(\xi ))^{2}}{V_{h}(\xi )-V_{h}(\xi ^{h}_{\star})}\leq
\frac{3k_{\alpha}^{2}[V_{h}'']_{{\mathrm{Lip}}}^{2}}{(V_{h}''(\xi ^{h}_{\star}))^{3}}
\vee \frac{3\|V_{h}''\|_{\infty}^{2}}{V_{h}''(\xi ^{h}_{\star})}=:
\nu _{h}.
\end{equation}
Combining \eqref{eq:Lq''} and \eqref{eq:nuh} and recalling that
$\|V_{h}'\|_{\infty}\leq k_{\alpha}$, we get
\begin{equation}
\label{neugl5}
|\mathcal{L}_{h,q}''(\xi )|
\leq \Big((q\vee \mu)\|V_{h}''\|_{\infty}
+\mu^2k_{\alpha}^{2}
+q\big(2\mu\vee (q-1)\big)\nu _{h}\Big)
\big(\mathcal{L}_{h,q}(\xi )+\mathcal{L}_{h,q-1}(\xi )\big).
\end{equation}

\noindent\textbf{\emph{Step~2. Case where $q=1$.}}\
If $q=1$, since
$\|V_{h}'\|_{\infty}\leq k_{\alpha}$, \eqref{eq:L1''} directly gives
\begin{equation}
\label{neugl6}
|\mathcal{L}_{h,1}''(\xi )|
\leq \big((\mu\vee 1)\|V_{h}''\|_{\infty}
+\mu(\mu\vee 2) k_{\alpha}^{2}\big)
\big(\mathcal{L}_{h,0}(\xi )
+\mathcal{L}_{h,1}(\xi )\big).
\end{equation}

\noindent\textbf{\emph{Step~3. Conclusion.}}\
Combining \eqref{neugl5} and
\eqref{neugl6}, we conclude that for all $q\in \mathbb{N}$ and all
$\xi \neq \xi ^{h}_{\star}$, we have
\begin{equation*}
|\mathcal{L}_{h,q}''(\xi )|
\leq \Big((q\vee \mu) \|V_{h}''\|_{\infty}+
\mu(\mu\vee 2)k_{\alpha}^{2}+q\big(2\mu \vee (q-1)\big)\nu _{h}
\Big)
\big(\mathcal{L}_{h,q}(\xi )+\mathcal{L}_{h,q-1}(\xi )\big).
\end{equation*}
For $q'\geq q$,
\begin{equation*}
|\lambda^{\mu_{h,q'}}_{h,q}|^{2}
\leq \frac{9}{64}q^{2} \big(V_h''(\xi^{h}_{\star})\big)^2
\leq 6qq'\frac{\|V_h''\|_{\infty}^3}{V''(\xi^h_\star)}
\leq \eta^{\mu_{h,q'}}_{h,q}.
\end{equation*}
The property $|\bar{\lambda}_{h,q}|^{2}\leq \bar{\eta}_{h,q}$ then follows.

By Assumption~\ref{asp:misc}(\ref{asp:misc-iv}),
\begin{align}
\label{eq:supVh''-Lip}
\sup _{h\in \overline{\mathcal{H}}}{\|V_{h}''\|_{\infty}}&=
\frac{\sup _{h\in \overline{\mathcal{H}}}{\|f_{X_{h}}\|_{\infty}}}{1-\alpha}<
\infty ,
\nonumber
\\
\sup _{h\in \overline{\mathcal{H}}}{[V_{h}'']_{{\mathrm{Lip}}}}&=
\frac{\sup _{h\in \mathcal{H}}{[f_{X_{h}}]_{{\mathrm{Lip}}}}}{1-\alpha}<
\infty .
\end{align}
Moreover, since $(\xi ^{h}_{\star})_{h\in \mathcal{H}}$ converges to
$\xi ^{0}_{\star}$ as $\mathcal{H}\ni h\downarrow 0$, there exists
$R>0$ such that $\xi ^{h}_{\star}\in B(\xi ^{0}_{\star},R)$,
$h\in \mathcal{H}$. Thus by Assumption~\ref{asp:misc}(\ref{asp:misc-iii}),
\begin{equation}
\label{eq:sup-1/Vh''<infty}
\sup _{h\in \overline{\mathcal{H}}}{
\frac{1}{V_{h}''(\xi ^{h}_{\star})}} \leq
\frac{1-\alpha}{\inf _{\substack{h\in \overline{\mathcal{H}}\\\xi \in B(\xi ^{0}_{\star},R)}}{f_{X_{h}}(\xi )}}
<\infty .
\end{equation}
Coming back to \eqref{eq:nuh}, \eqref{eq:supVh''-Lip} and
\eqref{eq:sup-1/Vh''<infty} show that
$\sup _{h\in \overline{\mathcal{H}}}{\nu _{h}}<\infty $.
Finally,
\begin{equation*}
\begin{aligned}
\sup _{h\in \overline{\mathcal{H}}}\eta _{h,q}^{\mu_{h,q'}}
&=\sup _{h\in \overline{\mathcal{H}}}
\Big((q\vee \mu_{h,q'}) \|V_{h}''\|_{\infty}\\
&\phantom{=\sup _{h\in \overline{\mathcal{H}}}
\Big(}
+\mu_{h,q'}(\mu_{h,q'}\vee 2)k_{\alpha}^{2}
+q\big(2\mu_{h,q'}\vee (q-1)\big)
\nu _{h}\Big)
<\infty .
\end{aligned}
\end{equation*}
In particular, $\sup_{h\in\overline{\mathcal{H}}}\bar{\eta}_{h,q}<\infty$.
\\

\noindent(\ref{lmm:lyapunov-iv})\
On $\mathcal{I}^{h}$, \eqref{eq:Vh-Vh>onIh} shows that every
$\xi \neq \xi ^{h}_{\star}$ satisfies
\begin{equation}
\label{eq:theta-<1}
(\xi -\xi ^{h}_{\star})^{2q} \leq
\frac{3^{q}}{(V_{h}''(\xi ^{h}_{\star}))^{q}}\big(V_{h}(\xi )-V_{h}(
\xi ^{h}_{\star})\big)^{q} \leq
\frac{3^{q}}{(V_{h}''(\xi ^{h}_{\star}))^{q}}\mathcal{L}^{\mu}_{h,q}(\xi ).
\end{equation}
Outside of $\mathcal{I}^{h}$, say if
$\xi >\xi ^{h}_{\star}+\varepsilon _{h}$ (the case
$\xi <\xi ^{h}_{\star}-\varepsilon _{h}$ is similar), given that
$V_{h}$ is increasing on $[\xi ^{h}_{\star},\infty )$, \eqref{13util} and
\eqref{eq:epsh} yield
\begin{equation*}
\begin{aligned}
V_{h}(\xi )-V_{h}(\xi ^{h}_{\star}) &\geq V_{h}(\xi ^{h}_{\star}+
\varepsilon _{h})-V_{h}(\xi ^{h}_{\star})
\\
&\geq (\xi -\xi ^{h}_{\star})\bigg(\frac{1}{2}V_{h}''(\xi ^{h}_{\star})
\varepsilon _{h} -\frac{1}{6}[V_{h}'']_{{\mathrm{Lip}}}\varepsilon _{h}^{2}
\bigg)
\\
&=(\xi -\xi ^{h}_{\star})
\frac{(V_{h}''(\xi ^{h}_{\star}))^{2}}{3[V_{h}'']_{{\mathrm{Lip}}}},
\end{aligned}
\end{equation*}
hence, for $\mu'\geq0$,
\begin{equation}
\label{eq:theta-<2}
(\xi -\xi ^{h}_{\star})^{2q} \leq
\frac{3^{2q}[V_{h}'']_{{\mathrm{Lip}}}^{2q}}{(V_{h}''(\xi ^{h}_{\star}))^{4q}}
\big(V_{h}(\xi )-V_{h}(\xi ^{h}_{\star})\big)^{2q} \leq
\frac{3^{2q}[V_{h}'']_{{\mathrm{Lip}}}^{2q}}{(V_{h}''(\xi ^{h}_{\star}))^{4q}}
\mathcal{L}_{h,2q}^{\mu'}(\xi ).
\end{equation}
The above upper bound also holds if
$\xi <\xi ^{h}_{\star}-\varepsilon _{h}$. Putting together
\eqref{eq:theta-<1} and \eqref{eq:theta-<2} gives for any
$\xi \neq \xi ^{h}_{\star}$ that
\begin{equation*}
\begin{aligned}
(\xi -\xi ^{h}_{\star})^{2q} &\leq \bigg(
\frac{3^{q}}{(V_{h}''(\xi ^{h}_{\star}))^{q}}\vee
\frac{3^{2q}[V_{h}'']_{{\mathrm{Lip}}}^{2q}}{(V_{h}''(\xi ^{h}_{\star}))^{4q}}
\bigg) \big(\mathcal{L}^{\mu}_{h,q}(\xi )+\mathcal{L}^{\mu'}_{h,2q}(\xi )\big)
\\
&=\kappa _{h,q}\big(\mathcal{L}^{\mu}_{h,q}(\xi )+\mathcal{L}^{\mu'}_{h,2q}(\xi )
\big).
\end{aligned}
\end{equation*}
The relations \eqref{eq:supVh''-Lip} and \eqref{eq:sup-1/Vh''<infty} yield
$\sup _{h\in \overline{\mathcal{H}}}\kappa _{h,q}<\infty $.\qed

\section{Proof of Lemma~\ref{lmm:gamma}}
\label{prf:gamma}

Let $n\geq 2$, $0\leq k\leq n$ and $p\geq 2$. Using the bound
$1+x\leq \mathrm{e} ^{x}$, $x\in \mathbb{R}$, gives
\begin{equation*}
\Pi _{k+1:n}\leq \exp \bigg(-\lambda \sum _{j=k+1}^{n}\gamma _{j}
\bigg)\exp \bigg(\zeta \sum _{j=k+1}^{n}\gamma _{j}^{2}\bigg),
\end{equation*}
with the convention $\sum _{\varnothing}=0$.
\\

\noindent\textbf{\emph{Step~1. Case where $\beta =1$.}}\

\noindent\textbf{\emph{Step~1.1. Inequality on $\Pi _{k+1:n}$.}}\
Note that
\begin{equation*}
\sum _{j=k+1}^{n}\gamma _{j}=\gamma _{1}\big(\psi (n+1)-\psi (k+1)
\big),
\end{equation*}
where $\psi $ is the digamma function that satisfies
\begin{equation*}
\ln{x}-\frac{1}{x}\leq \psi (x)\leq \ln{x}-\frac{1}{2x},\qquad x>0.
\end{equation*}
We refer for instance to Abramowitz et al.~\cite[Section~6.3]{ASR88} and Alzer \cite{Alz97} for elaborations on the digamma function.
Hence
\begin{equation*}
\sum _{j=k+1}^{n}\gamma _{j} \geq \gamma _{1}\ln \frac{n+1}{k+1}-
\frac{\gamma _{1}}{n+1}+\frac{\gamma _{1}}{2(k+1)} \geq \gamma _{1}
\ln \frac{n+1}{k+1}-\frac{\gamma _{1}}{2},
\end{equation*}
yielding
\begin{equation*}
\Pi _{k+1:n}\leq \exp \bigg(\zeta \frac{\pi ^{2}}{6}\gamma _{1}^{2}+
\frac{\lambda\gamma _{1}}{2}\bigg)
\frac{(k+1)^{\lambda \gamma _{1}}}{(n+1)^{\lambda \gamma _{1}}}.
\end{equation*}

\noindent\textbf{\emph{Step~1.2. Inequality on
$\sum _{k=1}^{n}\gamma _{k}^{p}\Pi _{k+1:n}$.}}\
The above estimate together
with a comparison between series and integrals gives
\begin{equation*}
\sum _{k=1}^{n}\gamma _{k}^{p}\Pi _{k+1:n} \leq \gamma _{1}^{p}\exp
\bigg(\zeta \frac{\pi ^{2}}{6}\gamma _{1}^{2}+
\frac{\lambda\gamma _{1}}{2}+\ln (2)(\lambda\gamma _{1} \vee p)\bigg)
\frac{\varphi _{\lambda \gamma _{1}-p+1}(n+1)}{(n+1)^{\lambda \gamma _{1}}}.
\end{equation*}

\noindent\textbf{\emph{Step~2. Case where $\beta \in (0,1)$.}}\

\noindent\textbf{\emph{Step~2.1. Inequality on $\Pi _{k+1:n}$.}}\
A comparison between
series and integrals gives
\begin{equation*}
\begin{aligned}
\Pi _{k+1:n} &\leq \exp \Big(-\lambda \gamma _{1}\big(\varphi _{1-
\beta}(n+1)-\varphi _{1-\beta}(k+1)\big)\Big)
\\
&
\hphantom{=:}
\times \exp \Big(2^{2\beta}\zeta \gamma _{1}^{2}\big(\varphi _{1-2
\beta}(n+1)-\varphi _{1-2\beta}(k+1)\big)\Big).
\end{aligned}
\end{equation*}

\noindent\textbf{\emph{Step~2.2. Inequality on
$\sum _{k=1}^{n}\gamma _{k}^{p}\Pi _{k+1:n}$.}}\
Introduce
\begin{equation*}
n_{0}:=\inf \{n\geq2:\gamma _{n}\leq \lambda /(2\zeta )\}-1.
\end{equation*}
Observe that
\begin{equation*}
1-\lambda \gamma _{n}+\zeta \gamma _{n}^{2}\leq 1-\frac{\lambda}{2}
\gamma _{n},\qquad n\geq n_{0}+1,
\end{equation*}
so that
\begin{align}
\label{sum:square:gamma:times:prod}
\sum _{k=1}^{n} \gamma _{k}^{p}\Pi _{k+1:n} &=\sum _{k=1}^{n_{0}
\wedge n}\gamma _{k}^{p}\Pi _{k+1:n_{0}\wedge n}\Pi _{n_{0}\wedge n+1:n}+
\sum _{k=n_{0}\wedge n+1}^{n}\gamma _{k}^{p}\Pi _{k+1:n}
\nonumber
\\
&\leq \bigg(\sum _{k=1}^{n_{0}\wedge n}\gamma _{k}^{p}\prod _{j=k+1}^{n_{0}
\wedge n}(1+\zeta \gamma _{j}^{2})\bigg)\prod _{j=n_{0}\wedge n+1}^{n}
\bigg(1-\frac{\lambda}{2}\gamma _{j}\bigg)
\nonumber
\\
&
\hphantom{=:}
+\sum _{k=1}^{n}\gamma _{k}^{p}\prod _{j=k+1}^{n}\bigg(1-
\frac{\lambda}{2}\gamma _{j}\bigg).
\end{align}
We simplify the first term on the right-hand side of
\eqref{sum:square:gamma:times:prod}. We have
\begin{equation*}
\begin{aligned}
\sum _{k=1}^{n_{0}\wedge n}\gamma _{k}^{p} \prod _{j=k+1}^{n_{0}
\wedge n}(1+\zeta \gamma _{j}^{2})
&\leq \gamma _{1}^{p-2} \sum _{k=1}^{n_{0}
\wedge n}\gamma _{k}^{2}\prod _{j=k+1}^{n_{0}\wedge n}(1+\zeta \gamma _{j}^{2})\\
&\leq \frac{\gamma _{1}^{p-2}}\zeta
\sum _{k=1}^{n_{0}\wedge n}\bigg(
\prod _{j=k}^{n_{0}\wedge n}(1+\zeta \gamma _{j}^{2})-\prod _{j=k+1}^{n_{0}
\wedge n}(1+\zeta \gamma _{j}^{2})\bigg)
\\
&\leq \frac{\gamma _{1}^{p-2}}\zeta
\prod _{j=1}^{n_{0}\wedge n}(1+
\zeta \gamma _{j}^{2})
\\
&\leq \frac{\gamma _{1}^{p-2}}\zeta
\exp \bigg(\zeta \sum _{j=1}^{n_{0}
\wedge n}\gamma _{j}^{2}\bigg).
\end{aligned}
\end{equation*}
Besides,
\begin{equation*}
\begin{aligned}
\prod _{j=n_{0}\wedge n+1}^{n}\bigg(1-\frac{\lambda}{2}\gamma _{j}
\bigg) &\leq \exp \bigg(-\frac{\lambda}{2}\sum _{j=1}^{n}\gamma _{j}
\bigg)\exp \bigg(\frac{\lambda}{2}\sum _{j=1}^{n_{0}\wedge n}\gamma _{j}
\bigg)
\\
&\leq \exp \bigg(-\frac{\lambda}{2}\sum _{j=1}^{n}\gamma _{j}\bigg)
\exp \bigg(\zeta \sum _{j=1}^{n_{0}\wedge n}\gamma ^{2}_{j}\bigg),
\end{aligned}
\end{equation*}
where we used that $\gamma _{j}\geq \lambda /(2\zeta )$ for
$j\leq n_{0}$. Combining the two preceding estimates and using a comparison
between series and integrals, we obtain
\begin{equation*}
\begin{aligned}
&\bigg(\sum _{k=1}^{n_{0}\wedge n}\gamma _{k}^{p}\prod _{j=k+1}^{n_{0}
\wedge n}(1+\zeta \gamma _{j}^{2})\bigg)\prod _{j=n_{0}\wedge n+1}^{n}
\bigg(1-\frac{\lambda}{2}\gamma _{j}\bigg)
\\
& \leq  \frac{\gamma _{1}^{p-2}}\zeta
\exp \bigg(2\zeta \sum _{j=1}^{n_{0}
\wedge n}\gamma _{j}^{2}\bigg)\exp \bigg(-\frac{\lambda}{2}\sum _{j=1}^{n}
\gamma _{j}\bigg)
\\
& \leq  \frac{\gamma _{1}^{p-2}}\zeta
\exp \big(2^{2\beta +1}\zeta
\gamma _{1}^{2}\varphi _{1-2\beta}(n+1)\big)
\exp\bigg(-\frac{\lambda\gamma _{1}}{2}\varphi _{1-\beta}(n+1)\bigg).
\end{aligned}
\end{equation*}
Next, we deal with the second term on the right-hand side of \eqref{sum:square:gamma:times:prod}. Since $(\gamma _{n})_{n\geq 1}$ is
a strictly positive nonincreasing sequence, it holds for any integer
$1\leq m\leq n$ that
\begin{equation*}
\begin{aligned}
&\sum _{k=1}^{n}\gamma _{k}^{p}\prod _{j=k+1}^{n}\bigg(1-
\frac{\lambda}{2}\gamma _{j}\bigg)
\\
&=\sum _{k=1}^{m}\gamma _{k}^{p}\prod _{j=k+1}^{n}\bigg(1-
\frac{\lambda}{2}\gamma _{j}\bigg)+\sum _{k=m+1}^{n}\gamma _{k}^{p}
\prod _{j=k+1}^{n}\bigg(1-\frac{\lambda}{2}\gamma _{j}\bigg)
\\
&\leq \prod _{j=m+1}^{n}\bigg(1-\frac{\lambda}{2}\gamma _{j}\bigg)
\sum _{k=1}^{m}\gamma _{k}^{p}\\
&\qquad\qquad\qquad
+\frac{2\gamma _{m+1}^{p-1}}{\lambda}\sum _{k=m+1}^{n}
\bigg(\prod _{j=k+1}^{n}\bigg(1-\frac{\lambda}{2}\gamma _{j}
\bigg)
-\prod _{j=k}^{n}\bigg(1-\frac{\lambda}{2}\gamma _{j}
\bigg)\bigg)\\
&\leq \exp \bigg(-\frac{\lambda}{2}\sum _{j=m+1}^{n}\gamma _{j}\bigg)
\gamma _{1}^{p-2}\sum _{k=1}^{n}\gamma^2_{k}+
\frac{2\gamma _{m+1}^{p-1}}{\lambda}\bigg(1-\prod _{j=m+1}^{n}\Big(1-
\frac{\lambda}{2}\gamma _{j}\Big)\bigg)
\\
&\leq 2^{2\beta}\gamma _{1}^{p}\exp \Big(-2^{-(\beta +1)}\lambda
\gamma _{1}\big(\varphi _{1-\beta}(n)-\varphi _{1-\beta}(m)\big)\Big)
\varphi _{1-2\beta}(n+1)+
\frac{2\gamma _{1}^{p-1}}{\lambda (m+1)^{(p-1)\beta}}.
\end{aligned}
\end{equation*}
Select $m=\lfloor n/2\rfloor $. By the concavity and nondecreasing monotonicity of $\varphi _{1-\beta}$, we get for $n\geq 2$ that
\begin{equation*}
\varphi _{1-\beta}(n) - \varphi _{1-\beta}(\lfloor n/2\rfloor)\geq n^{1-\beta}/2.
\end{equation*}
Hence for any integer $n\geq2$,
\begin{equation*}
\begin{aligned}
&\sum _{k=1}^{n}\gamma _{k}^{p}\prod _{j=k+1}^{n}\bigg(1-
\frac{\lambda}{2}\gamma _{j}\bigg)
\\
&\leq 2^{2\beta}\gamma _{1}^{p}\exp (-2^{-(\beta +2)}\lambda \gamma _{1}n^{1-
\beta})\varphi _{1-2\beta}(n+1) +
\frac{2^{1+(p-1)\beta}\gamma _{1}^{p-1}}{\lambda n^{(p-1)\beta}}.\qed
\end{aligned}
\end{equation*}

\section{Proof of Proposition~\ref{prop:local:strong:error:indicator:func}}
\label{prf:local:strong:error:indicator:func}

The next result draws inspiration from Giorgi et al.~\cite[Lemma~5.1 and Proposition~5.2]{Giorgi2020}
and Barrera et al.~\cite[Lemma~A.4]{barrera:hal-01710394}.

\begin{lemma}
\label{lmm:Xh1-Xh2}
Let $X$, $Y$ be real-valued random variables with respective bounded density
functions $f_{X}$, $f_{Y}$.
\begin{enumerate}[\rm(i)]
\item
\label{lmm:Xh1-Xh2-i}
If $X-Y$ is in $L^{p_{\star}}(\mathbb{P})$ for some $p_{\star}>1$, then
for every $\xi \in \mathbb{R}$,
\begin{equation*}
\mathbb{E}[|\mathds1_{\{X>\xi \}}\!-\!\mathds1_{\{Y>\xi \}}|] \leq (p_{
\star}^{\frac{p_{\star}}{p_{\star}+1}}+p_{\star}^{
\frac{1}{p_{\star}+1}} )(\|f_{X}\|_{\infty}+\|f_{Y}\|_{\infty})^{
\frac{p_{\star}}{p_{\star}+1}}\mathbb{E}[|X-Y|^{p_{\star}}]^{
\frac{1}{p_{\star}+1}}.
\end{equation*}
\item
\label{lmm:Xh1-Xh2-ii}
Assume there exists a positive constant $d_{X,Y}$ such that for any
$u\in \mathbb{R}$,
\begin{equation}
\label{Gaussian:concentration}
\mathbb{E}\big[\exp \big(u(X-Y)\big)\big]\leq \exp (d_{X,Y}u^{2}).
\end{equation}
Then for any $\xi \in \mathbb{R}$,
\begin{equation*}
\mathbb{E}[|\mathds1_{\{X>\xi \}}-\mathds1_{\{Y>\xi \}}|] \leq 2\sqrt{d_{X,Y}}+(
\|f_{X}\|_{\infty}+\|f_{Y}\|_{\infty})\sqrt{2d_{X,Y}\ln (d_{X,Y}^{-1}
\vee 1)}.
\end{equation*}
\end{enumerate}
\end{lemma}

\begin{proof}
Let $\xi \in \mathbb{R}$. For any $\delta >0$,
\begin{align}
\label{eq:Xh-X0<}
\mathbb{E}[|\mathds1_{\{X>\xi \}}-\mathds1_{\{Y>\xi \}}|] &=
\mathbb{E}[\mathds1_{\{X>\xi \}}\mathds1_{\{Y\leq \xi \}}+\mathds1_{
\{Y>\xi \}}\mathds1_{\{X\leq \xi \}}]
\nonumber
\\
&=\mathbb{P}[Y\leq \xi <X]+\mathbb{P}[X\leq \xi <Y]
\nonumber
\\
&=\mathbb{P}[Y\leq \xi ,\xi +\delta <X]+\mathbb{P}[Y\leq \xi <X\leq
\xi +\delta ]
\nonumber
\\
&
\hphantom{=:}
+\mathbb{P}[X\leq \xi ,\xi +\delta <Y]+\mathbb{P}[X\leq \xi <Y\leq
\xi +\delta ]
\nonumber
\\
&\leq \mathbb{P}[X-Y>\delta ]+\mathbb{P}[X-Y<-\delta ]
\nonumber
\\
&
\hphantom{=:}
+\mathbb{P}[\xi <X\leq \xi +\delta ]+\mathbb{P}[\xi <Y\leq \xi +
\delta ]
\nonumber
\\
&\leq \mathbb{P}[X-Y>\delta ]+\mathbb{P}[X-Y<-\delta ]
\nonumber
\\
&
\hphantom{=:}
+\delta (\|f_{X}\|_{\infty}+\|f_{Y}\|_{\infty}).
\end{align}

\noindent(\ref{lmm:Xh1-Xh2-i})\
By \eqref{eq:Xh-X0<},
\begin{equation*}
\begin{aligned}
\mathbb{E}[|\mathds1_{\{X>\xi \}}-\mathds1_{\{Y>\xi \}}|] &\leq
\mathbb{P}[|X-Y|>\delta ]+\delta (\|f_{X}\|_{\infty}+\|f_{Y}\|_{
\infty})
\\
&\leq \frac{\mathbb{E}[|X-Y|^{p_{\star}}]}{\delta ^{p_{\star}}}+
\delta (\|f_{X}\|_{\infty}+\|f_{Y}\|_{\infty}).
\end{aligned}
\end{equation*}
A straightforward optimisation yields the choice
\begin{equation*}
\delta =
\frac{{p_{\star}}^{\frac{1}{p_{\star}+1}}\mathbb{E}[|X-Y|^{p_{\star}}]^{\frac{1}{p_{\star}+1}}}{(\|f_{X}\|_{\infty}+\|f_{Y}\|_{\infty})^{\frac{1}{p_{\star}+1}}},
\end{equation*}
giving the result.
\\

\noindent(\ref{lmm:Xh1-Xh2-ii})\
By \eqref{eq:Xh-X0<}, the exponential Markov inequality
and \eqref{Gaussian:concentration}, for any $u\geq 0$,
\begin{align*}
&\mathbb{E}[|\mathds1_{\{X>\xi \}}-\mathds1_{\{Y>\xi \}}|]
\\
&\leq \mathbb{P}[X-Y>\delta ]+\mathbb{P}[-(X-Y)>\delta ]+\delta (\|f_{X}
\|_{\infty}+\|f_{Y}\|_{\infty})
\\
&\leq \exp (-u\delta )\Big(\mathbb{E}\big[\exp \big(u(X-Y)\big)\big]+
\mathbb{E}\big[\exp \big(-u(X-Y)\big)\big]\Big)
\\
&
\hphantom{=:}
+\delta (\|f_{X}\|_{\infty}+\|f_{Y}\|_{\infty})
\\
&\leq 2\exp (d_{X,Y}u^{2}-u\delta )+\delta (\|f_{X}\|_{\infty}+\|f_{Y}
\|_{\infty}).
\end{align*}
The right-hand side above is minimised in $u$ for
$\delta /(2d_{X,Y})$, yielding
\begin{equation*}
\mathbb{E}[|\mathds1_{\{X>\xi \}}-\mathds1_{\{Y>\xi \}}|] \leq 2\exp
\bigg(-\frac{\delta ^{2}}{4d_{X,Y}}\bigg)+\delta (\|f_{X}\|_{\infty}+
\|f_{Y}\|_{\infty}).
\end{equation*}
Evaluating the right-hand side of that inequality at
$ \delta =\sqrt{2d_{X,Y}\ln{(d_{X,Y}^{-1}\vee 1)}} $ gives
\begin{equation*}
\mathbb{E}[|\mathds1_{\{X>\xi \}}-\mathds1_{\{Y>\xi \}}|] \leq 2\sqrt{d_{X,Y}}+
\sqrt 2(\|f_{X}\|_{\infty}+\|f_{Y}\|_{\infty})\sqrt{d_{X,Y}\ln{(d_{X,Y}^{-1}
\vee 1)}},
\end{equation*}
completing the proof of (\ref{lmm:Xh1-Xh2-ii}).
\end{proof}

\begin{proof}[Proof of Proposition~\ref{prop:local:strong:error:indicator:func}]
Let $h,h'\in \overline{\mathcal{H}}$. Without loss of generality and to
avoid trivialities, we assume $0\leq h<h'$. We study separately the cases
$0<h<h'$ and $0=h<h'$.
\\

\noindent(\ref{prop:local:strong:error:indicator:func-i})\hyperref[prop:local:strong:error:indicator:func-ia]{a}.\
This proof is inspired from Giorgi et al.~\cite[Lemma~3.2]{Giorgi2020}.
Assume that
\begin{equation*}
\mathbb{E}\big[\big|\varphi (Y,Z)-\mathbb{E}[\varphi (Y,Z)\,|\,Y]
\big|^{p}\big]<\infty \qquad \text{for some } p\geq 1.
\end{equation*}
Then for $h=\frac{1}{K}\in \mathcal{H}$,
\begin{equation*}
\begin{aligned}
\mathbb{E}[|X_{h}-X_{0}|^{p}] &\leq \frac{1}{K}\sum _{k=1}^{K}
\mathbb{E}\big[\big|\varphi (Y,Z^{(k)})-\mathbb{E}[\varphi (Y, Z)\,|
\,Y]\big|^{p}\big]
\\
&=\mathbb{E}\big[\big|\varphi (Y,Z)-\mathbb{E}[\varphi (Y,Z)\,|\,Y]
\big|^{p}\big].
\end{aligned}
\end{equation*}

\noindent
\textbf{\emph{Case 1. $0<h<h'$.}}\
Introduce
$\widetilde{\varphi}(Y,Z):=\varphi (Y,Z)-\mathbb{E}[\varphi (Y,Z)\,|
\,Y]$ and take $h=\frac{1}{K}<h'=\frac{1}{K'}$. Then
\begin{equation}
\label{decdf:Xh}
X_{h}-X_{h'} =h\sum _{k=K'+1}^{K}\widetilde{\varphi}(Y,Z^{(k)}) +(h-h')
\sum _{k=1}^{K'}\widetilde{\varphi}(Y, Z^{(k)})
\end{equation}
so that by the triangle inequality,
\begin{equation*}
\mathbb{E}[|X_{h}-X_{h'}|^{p}]^{\frac{1}{p}} \leq h\mathbb{E}\bigg[
\bigg|\sum _{k=K'+1}^{K}\widetilde{\varphi}(Y,Z^{(k)})\bigg|^{p}
\bigg]^{\frac{1}{p}} +(h'-h)\mathbb{E}\bigg[\bigg|\sum _{k=1}^{K'}
\widetilde{\varphi}(Y,Z^{(k)})\bigg|^{p}\bigg]^{\frac{1}{p}}.
\end{equation*}
By the tower law and the Burkholder--Davis--Gundy inequality,
\begin{align}
\label{strong:error:diff:Xh}
&\mathbb{E}[|X_{h}-X_{h'}|^{p}]^{\frac{1}{p}}
\nonumber
\\
&\leq B_{p}\mathbb{E}\big[\big|\varphi (Y,Z)-\mathbb{E}[\varphi (Y,Z)
\,|\,Y]\big|^{p}\big]^{\frac{1}{p}}\big(h(K-K')^{\frac{1}{2}}+(h'-h){K'}^{
\frac{1}{2}}\big)
\nonumber
\\
&\leq B_{p}\mathbb{E}\big[\big|\varphi (Y,Z)-\mathbb{E}[\varphi (Y,Z)
\,|\,Y]\big|^{p}\big]^{\frac{1}{p}}(h'-h)^{\frac{1}{2}}\bigg(\Big(
\frac{h}{h'}\Big)^{\frac{1}{2}}+\Big(1-\frac{h}{h'}\Big)^{\frac{1}{2}}
\bigg)
\nonumber
\\
&\leq \sqrt{2}B_{p}\mathbb{E}\big[\big|\varphi (Y,Z)-\mathbb{E}[
\varphi (Y,Z)\,|\,Y]\big|^{p}\big]^{\frac{1}{p}}(h'-h)^{\frac{1}{2}}.
\end{align}

\noindent
\textbf{\emph{Case 2. $0 = h<h'$.}}\
Note that for
$h'=\frac{1}{K'}\in \mathcal{H}$,
\begin{equation*}
X_{h'}-X_{0} =\frac{1}{K'}\sum _{k=1}^{K'}\big(\varphi (Y,Z^{(k)})-
\mathbb{E}[\varphi (Y,Z)\,|\,Y]\big) =\frac{1}{K'}\sum _{k=1}^{K'}
\widetilde\varphi (Y,Z^{(k)}).
\end{equation*}
Since $(\widetilde\varphi (Y,Z^{(k)}))_{1\leq k\leq K'}$ is conditionally
on $Y$ a martingale increment sequence, using successively the tower law,
the Burkholder--Davis--Gundy inequality and the triangle inequality, we
obtain
\begin{equation*}
\begin{aligned}
\mathbb{E}[|X_{h'}-X_{0}|^{p}]^{\frac{1}{p}} &=h'\mathbb{E}\bigg[
\bigg|\sum _{k=1}^{K'}\widetilde\varphi (Y,Z^{(k)})\bigg|^{p}\bigg]^{
\frac{1}{p}}
\\
&\leq Ch'\mathbb{E}\bigg[\bigg|\sum _{k=1}^{K'}\widetilde\varphi (Y,Z^{(k)})^{2}
\bigg|^{\frac{p}{2}}\bigg]^{\frac{1}{p}} \leq C\mathbb{E}[|
\widetilde\varphi (Y,Z)|^{p}]^{\frac{1}{p}}(h')^{\frac{1}{2}}.
\end{aligned}
\end{equation*}
We conclude the proof of (\ref{prop:local:strong:error:indicator:func-i})\hyperref[prop:local:strong:error:indicator:func-ia]{a}) by applying Lemma~\ref{lmm:Xh1-Xh2}(\ref{lmm:Xh1-Xh2-i})
with $p=p_{\star}$.
\\

\noindent(\ref{prop:local:strong:error:indicator:func-i})\hyperref[prop:local:strong:error:indicator:func-ib]{b}.\
\textbf{\emph{Case 1. $0<h<h'$.}}\
Coming back to the decomposition
\eqref{decdf:Xh}, using that
$(\widetilde{\varphi}(Y,Z^{(k)}))_{1\leq k\leq K}$ are conditionally on
$Y$ independent and
\eqref{assumption:conditional:gaussian:concentration},
\begin{equation*}
\begin{aligned}
&\mathbb{E}\big[\exp \big(u(X_{h}-X_{h'})\big)\,\big|\,Y\big]
\\
&=\mathbb{E}\big[\exp \big(uh\widetilde{\varphi}(Y,Z)\big)\,\big|\,Y
\big]^{K-K'}\mathbb{E}\big[\exp \big(u(h-h')\widetilde{\varphi}(Y,Z)
\big)\,\big|\,Y\big]^{K'}
\\
&\leq \exp \bigg(C_{g}u^{2}h^{2}\Big(\frac{1}{h}-\frac{1}{h'}\Big)
\bigg)\exp \bigg(C_{g}u^{2}\frac{(h-h')^{2}}{h'}\bigg) =\exp \big(C_{g}(h'-h)u^{2}
\big).
\end{aligned}
\end{equation*}

\noindent\textbf{\emph{Case 2. $0= h<h'$.}}\
Let
$h'=\frac{1}{K'}\in \mathcal{H}$. Via the independence of the sequence
$(\widetilde{\varphi}(Y,Z^{(k)}))_{1\leq k\leq K'}$ conditionally on
$Y$, the assumption
\eqref{assumption:conditional:gaussian:concentration} and the fact that
$K'=\frac{1}{h'}$, we get
\begin{equation*}
\begin{aligned}
\mathbb{E}\big[\exp \big(u(X_{h'}-X_{0})\big)\,\big|\,Y\big] &=
\mathbb{E}\bigg[\exp \bigg(uh'\sum _{k=1}^{K'}\widetilde\varphi (Y,Z^{(k)})
\bigg)\,\bigg|\,Y\bigg]
\\
&=\mathbb{E}\big[\exp \big(uh'\widetilde\varphi (Y,Z)\big)\,\big|\,Y
\big]^{K'}
\\
&\leq \exp (C_{g}u^{2}h').
\end{aligned}
\end{equation*}
Taking expectations on both sides and setting $X=X_{h}$, $Y=X_{h'}$,
$d_{X,Y}=C_{g}(h'-h)$ in Case 1, and $X=X_{h}$, $Y=X_{0}$,
$d_{X,Y}=C_{g}h'$ in Case 2 guarantees the condition~\eqref{Gaussian:concentration}.
We conclude the proof by applying Lemma~\ref{lmm:Xh1-Xh2}(\ref{lmm:Xh1-Xh2-ii}).
\\

\noindent
(\ref{prop:local:strong:error:indicator:func-ii})\
We first write
\begin{equation}
\label{local:variance:ML:decomposition}
\mathbb{E}[|\mathds1_{\{X_{h_{\ell}}>\xi \}}-\mathds1_{\{X_{h_{\ell -1}}>
\xi \}}|] =\mathbb{P}[X_{h_{\ell -1}}\leq \xi < X_{h_{\ell}}]+
\mathbb{P}[X_{h_{\ell}}\leq \xi < X_{h_{\ell -1}}].
\end{equation}
Introducing the random variable
$G_{\ell}=h_{\ell}^{-\frac{1}{2}}(X_{h_{\ell}}-X_{h_{\ell -1}})$,
\begin{equation*}
\begin{aligned}
\mathbb{P}[X_{h_{\ell -1}}\leq \xi < X_{h_{\ell}}] &=\mathbb{P}[X_{h_{
\ell -1}}\leq \xi < X_{h_{\ell -1}}+h_{\ell}^{\frac{1}{2}}G_{\ell}]
\\
&=\mathbb{P}[X_{h_{\ell -1}}\leq \xi < X_{h_{\ell -1}}+h_{\ell}^{
\frac{1}{2}}G_{\ell},G_{\ell}>0]
\\
&=\mathbb{E}\big[\mathbb{P}[\xi -h_{\ell}^{\frac{1}{2}}G_{\ell }< X_{h_{
\ell -1}}\leq \xi ,G_{\ell}>0\,|\,G_{\ell}]\big]
\\
&=\mathbb{E}\big[\mathds1_{\{G_{\ell}>0\}}\big(F_{X_{h_{\ell -1}}|G_{
\ell}}(\xi )-F_{X_{h_{\ell -1}}|G_{\ell}}(\xi -h_{\ell}^{\frac{1}{2}}G_{
\ell})\big)\big].
\end{aligned}
\end{equation*}
The condition~\eqref{assump:unif:lipschitz:integrability:conditional:cdf}
yields
\begin{equation*}
\mathbb{P}[X_{h_{\ell -1}}\leq \xi < X_{h_{\ell}}] \leq h_{\ell}^{
\frac{1}{2}}\mathbb{E}[G_{\ell}^{+}K_{\ell}].
\end{equation*}
Similarly,
\begin{equation*}
\mathbb{P}[X_{h_{\ell}}\leq \xi \leq X_{h_{\ell -1}}]\leq h_{\ell}^{
\frac{1}{2}} \mathbb{E}[G_{\ell}^{-}K_{\ell}].
\end{equation*}
Coming back to \eqref{local:variance:ML:decomposition} and summing up the
two preceding inequalities yields
\begin{equation*}
\mathbb{E}[|\mathds1_{\{X_{h_{\ell}}>\xi \}}-\mathds1_{\{X_{h_{\ell -1}}>
\xi \}}|] \leq h_{\ell}^{\frac{1}{2}}\mathbb{E}[|G_{\ell}|K_{\ell}],
\end{equation*}
completing the proof.
\end{proof}

\section{Proof of Theorem~\ref{thm:ml-variance-cv}}
\label{prf:thm:ml-variance-cv}

Throughout, $\bar{C}$ designates a positive constant that may change from line to
line.
\\

\noindent\textbf{\emph{Step~1. Study of
$\xi ^{{{\mathrm{ML}}}}_{\mathbf{N}}-\xi _{\star}^{h_{L}}$.}}\
We follow a
similar strategy to the one used in Frikha~\cite[Lemma~2.7]{10.1214/15-AAP1109}
and decompose the dynamics of the sequence
$(\xi ^{h}_{n})_{n\geq 0}$ given by~\eqref{approximate:sgd:algorithm} into
\begin{align}
\label{decomposition:var:sa}
\xi ^{h}_{n+1}-\xi ^{h}_{\star }&=\big(1-\gamma _{n+1}V_{0}''(\xi ^{0}_{
\star})\big)(\xi ^{h}_{n}-\xi ^{h}_{\star})+\gamma _{n+1}\big(V_{0}''(
\xi ^{0}_{\star})-V_{h}''(\xi ^{h}_{\star})\big)(\xi ^{h}_{n}-\xi ^{h}_{
\star})
\nonumber
\\
&
\hphantom{=:}
+\gamma _{n+1}r^{h}_{n+1}+\gamma _{n+1}\rho ^{h}_{n+1}+\gamma _{n+1}e^{h}_{n+1},
\end{align}
where
\begin{align*}
r^{h}_{n+1} &:= V_{h}''(\xi ^{h}_{\star})(\xi ^{h}_{n}-\xi ^{h}_{
\star})-V_{h}'(\xi ^{h}_{n}),
\\
\rho ^{h}_{n+1} &:= V_{h}'(\xi ^{h}_{n})-V_{h}'(\xi ^{h}_{\star})-
\big(H_{1}(\xi ^{h}_{n},X_{h}^{(n+1)})-H_{1}(\xi ^{h}_{\star},X_{h}^{(n+1)})
\big),
\\
e^{h}_{n+1} &:= V_{h}'(\xi ^{h}_{\star})-H_{1}(\xi ^{h}_{\star},X_{h}^{(n+1)}).
\end{align*}
Iterating $n$ times \eqref{decomposition:var:sa} gives
\begin{align}
\label{eq:thetaopth-thetastarh}
\xi ^{h}_{n}-\xi ^{h}_{\star }&= (\xi ^{h}_{0}-\xi ^{h}_{\star})\Pi _{1:n}
+\sum _{k=1}^{n}\gamma _{k}\Pi _{k+1:n} \big(V_{0}''(\xi ^{0}_{\star})-V_{h}''(
\xi ^{h}_{\star})\big)(\xi ^{h}_{k}-\xi ^{h}_{\star})
\nonumber
\\
&
\hphantom{=:}
+\sum _{k=1}^{n}\gamma _{k}\Pi _{k+1:n}r^{h}_{k} +\sum _{k=1}^{n}
\gamma _{k}\Pi _{k+1:n}\rho ^{h}_{k} +\sum _{k=1}^{n}\gamma _{k}\Pi _{k+1:n}e^{h}_{k},
\end{align}
where for two positive integers $i$ and $n$,
\begin{equation*}
\Pi _{i:n}:=\prod _{j=i}^{n}\big(1-\gamma _{j}V_{0}''(\xi ^{0}_{\star})
\big)
\end{equation*}
with the convention that $\prod _{\varnothing}=1$. Using the inequality
$1+x\leq \mathrm{e} ^{x}$, $x\in \mathbb{R}$, a comparison between series and integrals
and \eqref{eq:phi} yields
\begin{align}
\label{upper:estimate:pi:i:n}
|\Pi _{i:n}| &\leq \exp \bigg(-V_{0}''(\xi ^{0}_{\star})\sum _{j=i}^{n}
\gamma _{j}\bigg)
\nonumber
\\
&\leq
\begin{cases}
\exp (\gamma _{1}V_{0}''(\xi ^{0}_{\star}))i^{\gamma _{1}V_{0}''(\xi ^{0}_{
\star})}(n+1)^{-\gamma _{1}V_{0}''(\xi ^{0}_{\star})}, &\quad \beta =1,
\\
\exp (-\gamma _{1}V_{0}''(\xi ^{0}_{\star})(\varphi _{1-\beta}(n+1)-
\varphi _{1-\beta}(i))), &\quad \beta \in (0,1).
\end{cases}
\end{align}
Since $\gamma _{n}=\gamma _{1}n^{-\beta}$, $\beta \in (0,1]$ with
$\gamma _{1}V_{0}''(\xi ^{0}_{\star})>\gamma _{1}\bar\lambda _{2}>1$ if
$\beta =1$, using~\eqref{upper:estimate:pi:i:n} and a comparison between
series and integrals (with computations similar to those performed in the
proof of Lemma~\ref{lmm:gamma}), we deduce that for any $b\geq 0$ and
$a>0$ such that $\gamma _{1}aV_{0}''(\xi ^{0}_{\star})>b$ if
$\beta =1$, for any integer $n$,
\begin{equation}
\label{eq:gamma^(1+b)|pi|<gamma^b}
\sum _{k=1}^{n}\gamma _{k}^{1+b}|\Pi _{k+1:n}|^{a} \leq \bar{C}\gamma _{n}^{b}
\end{equation}
as well as
\begin{equation}
\label{bound:Pi:1:n}
|\Pi _{1:n}|\leq \bar{C}\gamma _{n}.
\end{equation}
Via \eqref{eq:thetaopth-thetastarh}, for $\ell \in \{1,\dots ,L\}$,
\begin{equation}
\label{eq:xiML(h)-xi*(hL):aux}
\xi ^{h_{\ell}}_{n}-\xi ^{h_{\ell}}_{\star}-(\xi ^{h_{\ell -1}}_{n}-
\xi ^{h_{\ell -1}}_{\star}) =A^{\ell}_{n}+B^{\ell}_{n}+C^{\ell}_{n}+D^{
\ell}_{n}+E^{\ell}_{n},
\end{equation}
where
\begin{align*}
A^{\ell}_{n} &:=\big(\xi ^{h_{\ell}}_{0}-\xi ^{h_{\ell}}_{\star}-(
\xi ^{h_{\ell -1}}_{0}-\xi ^{h_{\ell -1}}_{\star})\big)\Pi _{1:n},
\\
B^{\ell}_{n} &:=\sum _{k=1}^{n}\gamma _{k}\Pi _{k+1:n} \Big(\big(V_{0}''(
\xi ^{0}_{\star})-V_{h_{\ell}}''(\xi ^{h_{\ell}}_{\star})\big)(\xi ^{h_{
\ell}}_{k}-\xi ^{h_{\ell}}_{\star})
\nonumber
\\
&
\hphantom{=:\sum _{k=1}^{n}\gamma _{k}\Pi _{k+1:n}
\Big(}
-\big(V_{0}''(\xi ^{0}_{\star})-V_{h_{\ell -1}}''(\xi ^{h_{\ell -1}}_{
\star})\big)(\xi ^{h_{\ell -1}}_{k} - \xi ^{h_{\ell -1}}_{\star})
\Big),
\\
C^{\ell}_{n} &:=\sum _{k=1}^{n}\gamma _{k}\Pi _{k+1:n}(r^{h_{\ell}}_{k}-r^{h_{
\ell -1}}_{k}),
\\
D^{\ell}_{n} &:=\sum _{k=1}^{n}\gamma _{k}\Pi _{k+1:n}(\rho ^{h_{\ell}}_{k}-
\rho ^{h_{\ell -1}}_{k}),
\\
E^{\ell}_{n} &:=\sum _{k=1}^{n}\gamma _{k}\Pi _{k+1:n}(e^{h_{\ell}}_{k}-e^{h_{
\ell -1}}_{k}).
\end{align*}
Hence the difference between the multilevel SA estimator
$\xi ^{{\mathrm{ML}}}_{\mathbf{N}}$ of the VaR in~\eqref{eq:ml} and the
solution $\xi ^{h_{L}}_{\star}=\argmin{V_{h_{L}}}$ can be decomposed into
\begin{align}
\label{eq:xiML(h)-xi*(hL)}
\xi ^{{\mathrm{ML}}}_{\mathbf{N}}-\xi ^{h_{L}}_{\star }&=\xi ^{h_{0}}_{N_{0}}-
\xi ^{h_{0}}_{\star }+\sum _{\ell =1}^{L}\big((\xi ^{h_{\ell}}_{N_{
\ell}}-\xi ^{h_{\ell}}_{\star})-(\xi ^{h_{\ell -1}}_{N_{\ell}}-\xi ^{h_{
\ell -1}}_{\star})\big)
\nonumber
\\
&=\xi ^{h_{0}}_{N_{0}}-\xi ^{h_{0}}_{\star }+\sum _{\ell =1}^{L}(A^{
\ell}_{N_{\ell}}+B^{\ell}_{N_{\ell}}+C^{\ell}_{N_{\ell}}+D^{\ell}_{N_{
\ell}}+E^{\ell}_{N_{\ell}}).
\end{align}
We now quantify the contribution of each term appearing in the decomposition~\eqref{eq:xiML(h)-xi*(hL)}.
\\

\noindent\textbf{\emph{Step~1.1. Study of
$\xi ^{h_{0}}_{N_{0}}-\xi ^{h_{0}}_{\star}$.}}\ Theorem~\ref{thm:variance-cv}(\ref{thm:variance-cv:i})
guarantees that
\begin{equation*}
\mathbb{E}[(\xi ^{h_{0}}_{N_{0}}-\xi ^{h_{0}}_{\star})^{2}]
\leq \bar{C} \gamma _{N_{0}}.
\end{equation*}

\noindent\textbf{\emph{Step~1.2. Study of
$\sum _{\ell =1}^{L}A^{\ell}_{N_{\ell}}$.}}\
From \eqref{bound:Pi:1:n},
\begin{equation}
\label{eq:E[A]<:aux}
\mathbb{E}[|A^{\ell}_{n}|^{2}]
\leq 2\bar{C}\sup _{\ell \geq 0}{\mathbb{E}[|\xi ^{h_{\ell}}_{0}-\xi ^{h_{\ell}}_{\star}|^{2}]}\gamma ^{2}_{n}
\leq \bar{C}\gamma ^{2}_{n}.
\end{equation}
Hence by the triangle inequality,
\begin{equation*}
\mathbb{E}\bigg[\bigg(\sum _{\ell =1}^{L}A^{\ell}_{N_{\ell}}\bigg)^{2}\bigg]^{\frac{1}{2}}
\leq \bar{C}\sum _{\ell =1}^{L}\gamma _{N_{\ell}}.
\end{equation*}

\noindent\textbf{\emph{Step~1.3. Study of
$\sum _{\ell =1}^{L}B^{\ell}_{N_{\ell}}$.}}\
By Lemma~\ref{lmm:thetah->theta0},
$(\xi ^{h_{\ell}}_{\star})_{\ell \geq 0}$ is bounded. Let
$\mathcal{K} \subseteq \mathbb{R}$ be a compact set containing the sequence
$(\xi ^{h_{\ell}}_{\star})_{\ell \geq 0}$. For any $\ell \geq 0$,
\begin{equation*}
\begin{aligned}
|V_{0}''(\xi ^{0}_{\star})-V_{h_{\ell}}''(\xi ^{h_{\ell}}_{\star})| &
\leq \frac{1}{1-\alpha}\big(|f_{X_{0}}(\xi ^{0}_{\star})-f_{X_{0}}(
\xi ^{h_{\ell}}_{\star})| +|f_{X_{0}}(\xi ^{h_{\ell}}_{\star})-f_{X_{h_{
\ell}}}(\xi ^{h_{\ell}}_{\star})|\big)
\\
&\leq \frac{1}{1-\alpha}\Big([f_{X_{0}}]_{{\mathrm{Lip}}}|\xi ^{h_{
\ell}}_{\star}-\xi ^{0}_{\star}| +\sup _{\xi \in \mathcal{K}}{|f_{X_{0}}(
\xi )-f_{X_{h_{\ell}}}(\xi )|}\Big)
\\
&\leq \bar{C}(h_{\ell}+h_{\ell}^{\frac{1}{4}+\delta})
\leq \bar{C}h_{\ell}^{(\frac{1}{4}+\delta )\wedge 1},
\end{aligned}
\end{equation*}
where we used Assumptions~\ref{asp:misc}(\ref{asp:misc-iv}) and~\ref{asp:fh-f0}
together with Proposition~\ref{prp:bias-cv}. Using this estimate, Theorem~\ref{thm:variance-cv}(\ref{thm:variance-cv:i})
and \eqref{eq:gamma^(1+b)|pi|<gamma^b}, we get
\begin{align}
\label{eq:E[B]<:aux}
\mathbb{E}[|B^{\ell}_{n}|^{2}]^{\frac{1}{2}} &\leq 2\big(|V_{0}''(
\xi ^{0}_{\star})-V_{h_{\ell}}''(\xi ^{h_{\ell}}_{\star})|\vee |V_{0}''(
\xi ^{0}_{\star})-V_{h_{\ell -1}}''(\xi ^{h_{\ell -1}}_{\star})|\big)
\nonumber
\\
&
\hphantom{=:}
\times \bigg(\sum _{k=1}^{n}\gamma _{k}|\Pi _{k+1:n}|\big(\mathbb{E}[|
\xi ^{h_{\ell}}_{k}-\xi ^{h_{\ell}}_{\star}|^{2}]^{\frac{1}{2}}\vee
\mathbb{E}[|\xi ^{h_{\ell -1}}_{k}-\xi ^{h_{\ell -1}}_{\star}|^{2}]^{
\frac{1}{2}}\big)\bigg)
\nonumber
\\
&\leq \bar{C}h_{\ell}^{(\frac{1}{4}+\delta )\wedge 1}\sum _{k=1}^{n}\gamma _{k}^{
\frac{3}{2}}|\Pi _{k+1:n}|
\leq \bar{C}\gamma _{n}^{\frac{1}{2}}h_{\ell}^{(\frac{1}{4}+\delta )\wedge 1}.
\end{align}
Therefore by the triangle inequality,
\begin{equation*}
\mathbb{E}\bigg[\bigg(\sum _{\ell =1}^{L}B^{\ell}_{N_{\ell}}\bigg)^{2}
\bigg]^{\frac{1}{2}}
\leq \bar{C}\sum _{\ell =1}^{L}\gamma _{N_{\ell}}^{\frac{1}{2}}h_{\ell}^{(\frac{1}{4}+\delta )\wedge 1}.
\end{equation*}

\noindent\textbf{\emph{Step~1.4. Study of
$\sum _{\ell =1}^{L}C^{\ell}_{N_{\ell}}$.}}\
Using a first-order Taylor
expansion, the uniform Lipschitz-regularity of
$(f_{X_{h}})_{h\in \mathcal{H}}$ under Assumption~\ref{asp:misc}(\ref{asp:misc-iv})
and Theorem~\ref{thm:variance-cv}(\ref{thm:variance-cv:ii}) gives
\begin{equation*}
\begin{aligned}
&\mathbb{E}[|r^{h_{\ell}}_{k}|^{2}]^{\frac{1}{2}}+\mathbb{E}[|r^{h_{
\ell -1}}_{k}|^{2}]^{\frac{1}{2}}
\\
&\leq
\frac{\sup _{h\in \mathcal{H}}{[f_{X_{h}}]_{{\mathrm{Lip}}}}}{(1-\alpha )}
\big(\mathbb{E}[(\xi ^{h_{\ell}}_{k-1}-\xi ^{h_{\ell}}_{\star})^{4}]^{
\frac{1}{2}}+\mathbb{E}[(\xi ^{h_{\ell -1}}_{k-1}-\xi ^{h_{\ell -1}}_{
\star})^{4}]^{\frac{1}{2}}\big)
\leq \bar{C}\gamma _{k}.
\end{aligned}
\end{equation*}
Hence re-using \eqref{eq:gamma^(1+b)|pi|<gamma^b}, we obtain
\begin{align}
\label{eq:E[C]<:aux}
\mathbb{E}[|C^{\ell}_{n}|^{2}]^{\frac{1}{2}} &\leq \sum _{k=1}^{n}
\gamma _{k}|\Pi _{k+1:n}|\big(\mathbb{E}[|r^{h_{\ell}}_{k}|^{2}]^{
\frac{1}{2}}+\mathbb{E}[|r^{h_{\ell -1}}_{k}|^{2}]^{\frac{1}{2}}\big)
\nonumber
\\
&\leq \bar{C}\sum _{k=1}^{n}\gamma _{k}^{2}|\Pi _{k+1:n}|
\leq \bar{C}\gamma _{n}.
\end{align}
Thus by the triangle inequality,
\begin{equation*}
\mathbb{E}\bigg[\bigg(\sum _{\ell =1}^{L}C^{\ell}_{N_{\ell}}\bigg)^{2}\bigg]^{\frac{1}{2}}
\leq \bar{C}\sum _{\ell =1}^{L}\gamma _{N_{\ell}}.
\end{equation*}

\noindent\textbf{\emph{Step~1.5. Study of
$\sum _{\ell =1}^{L}D^{\ell}_{N_{\ell}}$.}}\
By the definition
\eqref{eq:H1} of $H_{1}$ and \eqref{uniform:L2:bound:var:alg}, we have
\begin{align}
\label{eq:E[H]<|F|gamma^12}
&\mathbb{E}\big[\big(H_{1}(\xi ^{h}_{k},X_{h}^{(k+1)})-H_{1}(\xi ^{h}_{
\star},X_{h}^{(k+1)})\big)^{2}\big]
\nonumber
\\
&=\frac{1}{(1-\alpha )^{2}}\mathbb{E}[|\mathds1_{\{X_{h}^{(k+1)}>\xi ^{h}_{k}
\}} -\mathds1_{\{X_{h}^{(k+1)}>\xi ^{h}_{\star}\}}|]
\nonumber
\\
&=\frac{1}{(1-\alpha )^{2}}\mathbb{E}\big[\mathbb{E}[\mathds1_{\{\xi ^{h}_{k}<X_{h}^{(k+1)}<
\xi ^{h}_{\star}\}}+\mathds1_{\{\xi ^{h}_{\star}<X_{h}^{(k+1)}<\xi ^{h}_{k}
\}}\,|\,\mathcal{F}^{h}_{k}]\big]
\nonumber
\\
&=\frac{1}{(1-\alpha )^{2}}\mathbb{E}[|F_{X_{h}}(\xi ^{h}_{k})-F_{X_{h}}(
\xi ^{h}_{\star})|] \leq
\frac{\sup _{h\in \mathcal{H}}{\|f_{X_{h}}\|_{\infty}}}{(1-\alpha )^{2}}
\mathbb{E}[(\xi ^{h}_{k}-\xi ^{h}_{\star})^{2}]^{\frac{1}{2}}
\nonumber
\\
&\leq \bar{C}\frac{\sup _{h\in \mathcal{H}}{\|f_{X_{h}}\|_{\infty}}}{(1-\alpha )^{2}}\gamma _{k}^{\frac{1}{2}}.
\end{align}
Observe that
$(\rho _{k}^{h_{\ell}}-\rho _{k}^{h_{\ell -1}})_{k\geq 1}$ is a sequence
of $(\mathbb{F}^{h_{\ell}},\mathbb{P})$-martingale increments. Thus recalling
that $2\gamma _{1} V_{0}''(\xi ^{0}_{\star})>\frac{3}{2}$, we obtain by
\eqref{eq:E[H]<|F|gamma^12} and \eqref{eq:gamma^(1+b)|pi|<gamma^b} that
\begin{align}
\label{estimate:square:L2norm:D_n:final}
\mathbb{E}[|D^{\ell}_{n}|^{2}] &=\mathbb{E}\bigg[\bigg(\sum _{k=1}^{n}
\gamma _{k}\Pi _{k+1:n}(\rho ^{h_{\ell}}_{k}-\rho ^{h_{\ell -1}}_{k})
\bigg)^{2}\bigg]
\nonumber
\\
&=\sum _{k=1}^{n}\gamma _{k}^{2}|\Pi _{k+1:n}|^{2}\mathbb{E}[(\rho ^{h_{
\ell}}_{k}-\rho _{k}^{h_{\ell -1}})^{2}]
\nonumber
\\
&\leq 2\sum _{k=1}^{n}\gamma _{k}^{2}|\Pi _{k+1:n}|^{2}(\mathbb{E}[|
\rho ^{h_{\ell}}_{k}|^{2}]+\mathbb{E}[|\rho _{k}^{h_{\ell -1}}|^{2}])
\nonumber
\\
&\leq 2\sum _{k=1}^{n}\gamma _{k}^{2}|\Pi _{k+1:n}|^{2}\Big(
\mathbb{E}\big[\big(H_{1}(\xi ^{h_{\ell}}_{k},X_{h_{\ell}}^{(k+1)})-H_{1}(
\xi ^{h_{\ell}}_{\star},X_{h_{\ell}}^{(k+1)})\big)^{2}\big]
\nonumber
\\
&
\hphantom{=:2\sum _{k=1}^{n}\gamma _{k}^{2}|\Pi _{k+1:n}|^{2}\Big(}
+\mathbb{E}\big[\big(H_{1}(\xi ^{h_{\ell -1}}_{k},X_{h_{\ell -1}}^{(k+1)})-H_{1}(
\xi ^{h_{\ell -1}}_{\star},X_{h_{\ell -1}}^{(k+1)})\big)^{2}\big]
\Big)
\nonumber
\\
&\leq \bar{C}\sum _{k=1}^{n}\gamma _{k}^{\frac{5}{2}}|\Pi _{k+1:n}|^{2}
\leq \bar{C}\gamma _{n}^{\frac{3}{2}}.
\end{align}
Because the random variables $(D_{n}^{\ell})_{1\leq \ell \leq L}$ are independent and centered, \eqref{estimate:square:L2norm:D_n:final} yields
\begin{equation*}
\mathbb{E}\bigg[\bigg(\sum _{\ell =1}^{L}D_{N_{\ell}}^{\ell}\bigg)^{2}\bigg]
=\sum _{\ell =1}^{L}\mathbb{E}[|D_{N_{\ell}}^{\ell}|^{2}]
\leq \bar{C}\sum _{\ell =1}^{L}\gamma _{N_{\ell}}^{\frac{3}{2}}.
\end{equation*}

\noindent\textbf{\emph{Step~1.6. Study of
$\sum _{\ell =1}^{L}E^{\ell}_{N_{\ell}}$.}}\
Note again that
$(e_{k}^{h_{\ell}}-e _{k}^{h_{\ell -1}})_{k\geq 1}$ is a sequence of
$(\mathbb{F}^{h_{\ell}},\mathbb{P})$-martingale increments so that
\begin{align}
\label{estimate:square:L2norm:E_n}
\mathbb{E}[|E_{n}^{\ell}|^{2}] &=\mathbb{E}\bigg[\bigg(\sum _{k=1}^{n}
\gamma _{k}\Pi _{k+1:n}\big(e^{h_{\ell}}_{k}-e^{h_{\ell -1}}_{k}\big)
\bigg)^{2}\bigg]
\nonumber
\\
&=\sum _{k=1}^{n}\gamma _{k}^{2}|\Pi _{k+1:n}|^{2}\mathbb{E}[(e^{h_{
\ell}}_{k}-e^{h_{\ell -1}}_{k})^{2}]
\nonumber
\\
&\leq \sum _{k=1}^{n}\gamma _{k}^{2}|\Pi _{k+1:n}|^{2}\mathbb{E}\big[
\big(H_{1}(\xi ^{h_{\ell}}_{\star},X_{h_{\ell}})-H_{1}(\xi ^{h_{\ell -1}}_{
\star},X_{h_{\ell -1}})\big)^{2}\big].
\end{align}
The last term on the right-hand side can be bounded from above via
\begin{align}
\label{eq:helper1}
&\mathbb{E}\big[\big(H_{1}(\xi ^{h_{\ell}}_{\star},X_{h_{\ell}})-H_{1}(
\xi ^{h_{\ell -1}}_{\star},X_{h_{\ell -1}})\big)^{2}\big]
\nonumber
\\
&\leq 2\Big(\mathbb{E}\big[\big(H_{1}(\xi ^{h_{\ell}}_{\star}, X_{h_{
\ell}})-H_{1}(\xi ^{h_{\ell -1}}_{\star},X_{h_{\ell}})\big)^{2}\big]
\nonumber
\\
&
\hphantom{=2\Big(}
+\mathbb{E}\big[\big(H_{1}(\xi ^{h_{\ell -1}}_{\star},X_{h_{\ell}})-H_{1}(
\xi ^{h_{\ell -1}}_{\star }, X_{h_{\ell -1}})\big)^{2}\big]\Big)
\nonumber
\\
&\leq \frac{2}{(1-\alpha )^{2}}\big(\mathbb{E}[(\mathds1_{\{X_{h_{
\ell}}>\xi ^{h_{\ell}}_{\star}\}}-\mathds1_{\{X_{h_{\ell}}>\xi ^{h_{
\ell -1}}_{\star}\}})^{2}]
\nonumber
\\
&
\hphantom{=\frac{2}{(1-\alpha )^{2}}\big(}
+\mathbb{E}[(\mathds1_{\{X_{h_{\ell}}>\xi ^{h_{\ell -1}}_{\star}\}}-
\mathds1_{\{X_{h_{\ell -1}}>\xi ^{h_{\ell -1}}_{\star}\}})^{2}]\big).
\end{align}
On the one hand, by Proposition~\ref{prp:bias-cv},
\begin{align*}
\mathbb{E}[(\mathds1_{\{X_{h_{\ell}}>\xi ^{h_{\ell}}_{\star}\}}-
\mathds1_{\{X_{h_{\ell}}>\xi ^{h_{\ell -1}}_{\star}\}})^{2}] &=
\mathbb{E}[\mathds1_{\{\xi ^{h_{\ell -1}}_{\star}<X_{h_{\ell}}<\xi ^{h_{
\ell}}_{\star}\}}+\mathds1_{\{\xi ^{h_{\ell}}_{\star}<X_{h_{\ell}}<
\xi ^{h_{\ell -1}}_{\star}\}}]
\nonumber
\\
&=|F_{X_{h_{\ell}}}(\xi ^{h_{\ell}}_{\star})-F_{X_{h_{\ell}}}(\xi ^{h_{
\ell -1}}_{\star})|
\\
&\leq \sup _{h\in \mathcal{H}}{\|f_{X_{h}}\|_{\infty}}|\xi ^{h_{\ell}}_{\star}-\xi ^{h_{\ell -1}}_{\star}|
\leq \bar{C}h_{\ell}.
\end{align*}
On the other hand, Proposition~\ref{prop:local:strong:error:indicator:func}
yields
\begin{equation*}
\mathbb{E}[(\mathds1_{\{X_{h_{\ell}}>\xi ^{h_{\ell -1}}_{\star}\}}-
\mathds1_{\{X_{h_{\ell -1}}>\xi ^{h_{\ell -1}}_{\star}\}})^{2}]
\leq \bar{C}\epsilon (h_{\ell}),
\end{equation*}
where $\epsilon (h_{\ell})$ is defined in~\eqref{eq:eps(hl)}. Combining the two
previous estimates and recalling \eqref{estimate:square:L2norm:E_n},
\eqref{eq:helper1} and \eqref{eq:gamma^(1+b)|pi|<gamma^b}, we get
\begin{align}
\label{estimate:square:L2norm:E_n:final}
\mathbb{E}[|E_{n}^{\ell}|^{2}] &=\mathbb{E}\bigg[\bigg(\sum _{k=1}^{n}
\gamma _{k}\Pi _{k+1:n}(e^{h_{\ell}}_{k}-e^{h_{\ell -1}}_{k})\bigg)^{2}
\bigg]
\nonumber
\\
&\leq \bar{C}\big(h_{\ell}+\epsilon (h_{\ell})\big)\sum _{k=1}^{n}\gamma _{k}^{2}|\Pi _{k+1:n}|^{2}
\leq \bar{C}\epsilon (h_{\ell})\gamma _{n}.
\end{align}
Observe again that the random variables
$(E_{n}^{\ell})_{1\leq \ell \leq L}$ are independent and centered. Hence
\eqref{estimate:square:L2norm:E_n:final} directly yields
\begin{equation*}
\mathbb{E}\bigg[\bigg(\sum _{\ell =1}^{L}E_{N_{\ell}}^{\ell}\bigg)^{2}\bigg]
=\sum _{\ell =1}^{L}\mathbb{E}[|E_{N_{\ell}}^{\ell}|^{2}]
\leq \bar{C}\sum _{\ell =1}^{L}\gamma _{N_{\ell}}\epsilon (h_{\ell}).
\end{equation*}

\noindent\textbf{\emph{Step~1.7. Conclusion.}}\
Gathering all the previous estimates
on each term of \eqref{eq:xiML(h)-xi*(hL)}, we obtain
\begin{align}
\label{preliminary:estimate:L1:norm:ML:VaR}
\mathbb{E}[(\xi ^{{\mathrm{ML}}}_{\mathbf{N}}-\xi ^{h_{L}}_{\star})^{2}]
&\leq \bar{C}\bigg(\gamma _{N_{0}}+\Big(\sum _{\ell =1}^{L} \big(\gamma _{N_{
\ell}}+\gamma _{N_{\ell}}^{\frac{1}{2}}h_{\ell}^{(\frac{1}{4}+\delta )
\wedge 1}\big)\Big)^{2}
\nonumber
\\
&
\hphantom{=:K\bigg(}
+\sum _{\ell =1}^{L}\big(\gamma _{N_{\ell}}^{\frac{3}{2}}+\gamma _{N_{
\ell}}\epsilon (h_{\ell})\big)\bigg).
\end{align}
From the Cauchy--Schwarz inequality and the fact that
$h_{\ell}^{\frac{1}{2}}\leq \bar{C}\epsilon (h_{\ell})$,
\begin{eqnarray}
\label{local:upper:bound:L1:norm:ML:VaR}
\sum _{\ell =1}^{L}\gamma _{N_{\ell}}^{\frac{1}{2}}h_{\ell}^{(
\frac{1}{4}+\delta )\wedge 1} \leq \bigg(\sum _{\ell =1}^{L}\gamma _{N_{
\ell}}h_{\ell}^{\frac{1}{2}}\bigg)^{\frac{1}{2}}\bigg(\sum _{\ell =1}^{L}h_{
\ell}^{2(\delta \wedge \frac{3}{4})}\bigg)^{\frac{1}{2}}
\leq \bar{C}\bigg(\sum _{\ell =1}^{L}\gamma _{N_{\ell}}\epsilon (h_{\ell})\bigg)^{\frac{1}{2}},
\end{eqnarray}
where we used that
$\sum _{\ell =1}^{\infty }h_{\ell}^{2(\delta \wedge \frac{3}{4})} =
\frac{h_{0}^{2(\delta \wedge \frac{3}{4})}}{M^{2(\delta \wedge \frac{3}{4})}-1}<
\infty $. Plugging \eqref{local:upper:bound:L1:norm:ML:VaR} into
\eqref{preliminary:estimate:L1:norm:ML:VaR} concludes the proof of
\eqref{L2:norm:ML:VaR}.
\\

\noindent\textbf{\emph{Step~2. Study of
$\chi ^{{{\mathrm{ML}}}}_{\mathbf{N}}-\chi _{\star}^{h_{L}}$.}}\
For the sake
of simplicity, we assume \text{${\chi _{0}^{h_{\ell}}=0}$,} \text{${\ell
\in \{0, \dots , L\}}$.} The general case follows from similar arguments.
Recalling the definition \eqref{epsilon}, we derive from the decomposition
\eqref{decomposition:stat:error:cvar:sa:algorithm} that
\begin{equation*}
\begin{aligned}
&\chi ^{h_{\ell}}_{n}-\chi ^{h_{\ell}}_{\star}-(\chi ^{h_{\ell -1}}_{n}-
\chi ^{h_{\ell -1}}_{\star})
\\
&=\frac{1}{n}\sum _{k=1}^{n}(\varepsilon _{k}^{h_{\ell}}-\varepsilon _{k}^{h_{
\ell -1}})
\\
&
\hphantom{=:}
+\frac{1}{n}\sum _{k=1}^{n} \Big(V_{h_{\ell}}(\xi _{k-1}^{h_{\ell}})-V_{h_{
\ell}}(\xi ^{h_{\ell}}_{\star})-\big(V_{h_{\ell -1}}(\xi _{k-1}^{h_{
\ell -1}})-V_{h_{\ell -1}}(\xi ^{h_{\ell -1}}_{\star})\big)\Big).
\end{aligned}
\end{equation*}
Thus
\begin{align}
\label{eq:chiML(h)-chi*(hL)}
\chi ^{{\mathrm{ML}}}_{\mathbf{N}}-\chi ^{h_{L}}_{\star }&=\chi ^{h_{0}}_{N_{0}}-
\chi ^{h_{0}}_{\star }+\sum _{\ell =1}^{L} \big(\chi ^{h_{\ell}}_{N_{
\ell}}-\chi ^{h_{\ell}}_{\star}-(\chi ^{h_{\ell -1}}_{N_{\ell}}-\chi ^{h_{
\ell -1}}_{\star})\big)
\nonumber
\\
&=\chi ^{h_{0}}_{N_{0}}-\chi ^{h_{0}}_{\star }+\sum _{\ell =1}^{L}
\frac{1}{N_{\ell}}\sum _{k=1}^{N_{\ell}}(\varepsilon _{k}^{h_{\ell}}-
\varepsilon _{k}^{h_{\ell -1}})
\nonumber
\\
&
\hphantom{=:}
+\sum _{\ell =1}^{L} \frac{1}{N_{\ell}}\sum _{k=1}^{N_{\ell}} \Big( V_{h_{
\ell}}(\xi _{k-1}^{h_{\ell}})-V_{h_{\ell}}(\xi ^{h_{\ell}}_{\star})
\nonumber
\\
&
\hphantom{=:+\sum _{\ell =1}^{L}
\frac{1}{N_{\ell}}\sum _{k=1}^{N_{\ell}}
\Big(}
-\big(V_{h_{\ell -1}}(\xi _{k-1}^{h_{\ell -1}})-V_{h_{\ell -1}}(\xi ^{h_{
\ell -1}}_{\star})\big)\Big).
\end{align}

\noindent\textbf{\emph{Step~2.1. Study of
$\chi _{N_{0}}^{h_{0}}-\chi ^{h_{0}}_{\star}$.}}\
By Theorem~\ref{thm:variance-cv}(\ref{thm:variance-cv:ii}),
\begin{equation}
\label{L2:estimate:CVaR:level:0}
\mathbb{E}[(\chi _{N_{0}}^{h_{0}}-\chi ^{h_{0}}_{\star})^{2}]\leq
\frac{\bar{C}}{N_{0}^{1\wedge 2\beta}}.
\end{equation}

\noindent\textbf{\emph{Step~2.2. Study of
$\sum _{\ell =1}^{L}\frac{1}{N_{\ell}}\sum _{k=1}^{N_{\ell}}(
\varepsilon _{k}^{h_{\ell}}-\varepsilon _{k}^{h_{\ell -1}})$.}}\
Via the
definition \eqref{epsilon}, the \text{$1$-Lipschitz} property of
$x\mapsto x^{+}$, $x\in \mathbb{R}$, the fact that
$\mathbb{E}[(X_{h_{\ell}}-X_{h_{\ell -1}})^{2}]\leq \bar{C}h_{\ell}$ by~\eqref{strong:error:diff:Xh}
with $p=2$ and the fact that
$|\xi ^{h_{\ell}}_{\star}-\xi ^{h_{\ell -1}}_{\star}|\leq \bar{C}h_{\ell}$ by
Proposition~\ref{prp:bias-cv}, we~have
\begin{align}
\label{helper}
&\mathbb{E}[(\varepsilon _{k}^{h_{\ell}}-\varepsilon _{k}^{h_{\ell -1}})^{2}
\,|\, \mathcal{F}_{k-1}]
\nonumber
\\
&=\frac{1}{(1-\alpha )^{2}}\mathbb{V}\mathrm{ar}[(X^{(k)}_{h_{\ell}}-
\xi ^{h_{\ell}}_{k-1})^{+}-(X^{(k)}_{h_{\ell -1}}-\xi ^{h_{\ell -1}}_{k-1})^{+}
\,|\,\mathcal{F}_{k-1}]
\nonumber
\\
&\leq \frac{1}{(1-\alpha )^{2}}\mathbb{E}\big[\big(X^{(k)}_{h_{\ell}}-
\xi ^{h_{\ell}}_{k-1}-(X^{(k)}_{h_{\ell -1}}-\xi ^{h_{\ell -1}}_{k-1})
\big)^{2}\,\big|\,\mathcal{F}_{k-1}\big]
\nonumber
\\
&\leq \frac{3}{(1-\alpha )^{2}}\Big(\mathbb{E}[(X_{h_{\ell}}-X_{h_{
\ell -1}})^{2}]
\nonumber
\\
&
\hphantom{=:\frac{3}{(1-\alpha )^{2}}\Big(}
+\big(\xi _{k-1}^{h_{\ell}}-\xi _{k-1}^{h_{\ell -1}}-(\xi ^{h_{\ell}}_{
\star}-\xi ^{h_{\ell -1}}_{\star})\big)^{2}+(\xi ^{h_{\ell}}_{\star}-
\xi ^{h_{\ell -1}}_{\star})^{2}\Big)
\nonumber
\\
&\leq \bar{C}\Big(h_{\ell}+\big(\xi _{k-1}^{h_{\ell}}-\xi _{k-1}^{h_{\ell -1}}-(
\xi ^{h_{\ell}}_{\star}-\xi ^{h_{\ell -1}}_{\star})\big)^{2}\Big).
\end{align}
We seek an $L^{2}(\mathbb{P})$-estimate for the quantity
$ \xi _{k}^{h_{\ell}}-\xi _{k}^{h_{\ell -1}}-(\xi ^{h_{\ell}}_{\star}-
\xi ^{h_{\ell -1}}_{\star})$. By the decomposition
\eqref{eq:xiML(h)-xi*(hL):aux}, using \eqref{eq:E[A]<:aux}--\eqref{eq:E[C]<:aux}, \eqref{estimate:square:L2norm:D_n:final} and
\eqref{estimate:square:L2norm:E_n:final}, we obtain
\begin{equation}
\label{L2:bound:diff:var:ML}
\mathbb{E}\big[\big(\xi _{k}^{h_{\ell}}-\xi _{k}^{h_{\ell -1}}-(\xi ^{h_{
\ell}}_{\star}-\xi ^{h_{\ell -1}}_{\star})\big)^{2}\big]
\leq \bar{C}\big(\gamma _{k}^{\frac{3}{2}}+\gamma _{k}\epsilon (h_{\ell})\big).
\end{equation}
Recall that for each $\ell \in \{1,\dots ,L\}$,
$(e^{h_{\ell}}_{k}-e^{h_{\ell -1}}_{k})_{1\leq k\leq n}$ are
$(\mathbb{F}^{h_{\ell}},\mathbb{P})$-martingale increments. Hence taking
expectations on both sides of \eqref{helper}, using
\eqref{L2:bound:diff:var:ML} and the fact that by assumption and Lemma~\ref{lmm:thetah->theta0},
$\sup _{\ell \geq 0}\mathbb{E}[|\xi ^{h_{\ell}}_{0}|]+|\xi ^{h_{\ell}}_{
\star}|<\infty $, we get
\begin{equation*}
\begin{aligned}
\mathbb{E}\bigg[\bigg(\frac{1}{n}\sum _{k=1}^{n}(\varepsilon _{k}^{h_{
\ell}}-\varepsilon _{k}^{h_{\ell -1}})\bigg)^{2}\bigg] &=
\frac{1}{n^{2}}\sum _{k=1}^{n}\mathbb{E}[(\varepsilon _{k}^{h_{\ell}}-
\varepsilon _{k}^{h_{\ell -1}})^{2}]
\\
&\leq \frac{1}{n^{2}}\Big(h_{\ell}+\sup _{\ell \geq 1}\mathbb{E}\big[
\big(\xi _{0}^{h_{\ell}}-\xi _{0}^{h_{\ell -1}}-(\xi ^{h_{\ell}}_{
\star}-\xi ^{h_{\ell -1}}_{\star})\big)^{2}\big]\Big)
\\
&
\hphantom{=:}
+\frac{1}{n^{2}}\sum _{k=2}^{n}\mathbb{E}[(\varepsilon _{k}^{h_{\ell}}-
\varepsilon _{k}^{h_{\ell -1}})^{2}]
\\
&\leq \bar{C}\bigg(\frac{1}{n^{2}}+\frac{h_{\ell}}{n}+\frac{1}{n^{2}}\sum _{k=1}^{n}
\big(\gamma _{k}^{\frac{3}{2}}+\gamma _{k}\epsilon (h_{\ell})\big)
\bigg)
\\
&\leq \bar{C}\bigg(\frac{h_{\ell}}{n}+\frac{1}{n^{2}}\sum _{k=1}^{n}\big(
\gamma _{k}^{\frac{3}{2}}+\gamma _{k}\epsilon (h_{\ell})\big)\bigg).
\end{aligned}
\end{equation*}
As the random variables
$(\frac{1}{N_{\ell}}\sum _{k=1}^{N_{\ell}}(\varepsilon _{k}^{h_{\ell}}-
\varepsilon _{k}^{h_{\ell -1}}))_{1\leq \ell \leq L}$ are independent and
centered, we obtain
\begin{align}
\nonumber&\label{L2:estimate:1/n:sum:increment:martingale}
\mathbb{E}\bigg[\bigg(\sum _{\ell =1}^{L}\frac{1}{N_{\ell}}\sum _{k=1}^{N_{
\ell}}(\varepsilon _{k}^{h_{\ell}}-\varepsilon _{k}^{h_{\ell -1}})
\bigg)^{2}\bigg] \\
&\quad =\sum _{\ell =1}^{L}\mathbb{E}\bigg[\bigg(
\frac{1}{N_{\ell}}\sum _{k=1}^{N_{\ell}}(\varepsilon _{k}^{h_{\ell}}-
\varepsilon _{k}^{h_{\ell -1}})\bigg)^{2}\bigg]
\nonumber
\\
&\quad \leq \bar{C}\sum _{\ell =1}^{L}\bigg(\frac{h_{\ell}}{N_{\ell}}+
\frac{1}{N_{\ell}^{2}}\sum _{k=1}^{N_{\ell}}\big(\gamma _{k}^{
\frac{3}{2}}+\gamma _{k}\epsilon (h_{\ell})\big)\bigg).
\end{align}

\noindent\textbf{\emph{Step~2.3. Study of
\begin{equation*}
\sum _{\ell =1}^{L} \frac{1}{N_{\ell}}\sum _{k=1}^{N_{\ell}} \Big(V_{h_{
\ell}}(\xi _{k-1}^{h_{\ell}})-V_{h_{\ell}}(\xi ^{h_{\ell}}_{\star}) -
\big(V_{h_{\ell -1}}(\xi _{k-1}^{h_{\ell -1}})-V_{h_{\ell -1}}(\xi ^{h_{
\ell -1}}_{\star})\big)\Big).
\end{equation*}
}}

\noindent
We now deal with the last term on the right-hand side of
\eqref{eq:chiML(h)-chi*(hL)}. Using Minkowski's inequality, second-order
Taylor expansions of $V_{h_{\ell}}$ and $V_{h_{\ell -1}}$, the facts that
\text{${V_{h_{\ell}}'(\xi ^{h_{\ell}}_{\star})=0}$,}
$\sup _{h\in \mathcal{H}}{\|V_{h}''\|_{\infty}}<\infty $ by Assumption~\ref{asp:misc}(\ref{asp:misc-iv}),
\eqref{uniform:L4:bound:var:alg} and
$\sum _{k\geq 1}\gamma _{k}=\infty $ leads to
\begin{equation*}
\begin{aligned}
&\mathbb{E}\bigg[\bigg(\frac{1}{n}\sum _{k=1}^{n} \Big(V_{h_{\ell}}(
\xi _{k-1}^{h_{\ell}})-V_{h_{\ell}}(\xi ^{h_{\ell}}_{\star})-\big(V_{h_{
\ell -1}}(\xi _{k-1}^{h_{\ell -1}})-V_{h_{\ell -1}}(\xi ^{h_{\ell -1}}_{
\star})\big)\Big)\bigg)^{2}\bigg]^{\frac{1}{2}}
\\
&\leq C\bigg(\frac{1}{n}+\frac{1}{n-1}\sum _{k=2}^{n}\big(\mathbb{E}[(
\xi ^{h_{\ell}}_{k-1}-\xi ^{h_{\ell}}_{\star})^{4}]^{\frac{1}{2}}+
\mathbb{E}[(\xi ^{h_{\ell -1}}_{k-1}-\xi ^{h_{\ell -1}}_{\star})^{4}]^{
\frac{1}{2}} \big)\bigg)
\\
&\leq \bar{C}\bigg(\frac{1}{n}+\frac{1}{n-1}\sum _{k=1}^{n-1}\gamma _{k}\bigg)
\leq \bar{C}\bar{\gamma}_{n}.
\end{aligned}
\end{equation*}
Thus
\begin{align}
\label{L2:estimate:sum:remainder:cvar}
&\mathbb{E}\bigg[\bigg(\sum _{\ell =1}^{L} \frac{1}{N_{\ell}}\sum _{k=1}^{N_{
\ell}} \Big( V_{h_{\ell}}(\xi _{k-1}^{h_{\ell}})-V_{h_{\ell}}(\xi ^{h_{
\ell}}_{\star}) -\big(V_{h_{\ell -1}}(\xi _{k-1}^{h_{\ell -1}})-V_{h_{
\ell -1}}(\xi ^{h_{\ell -1}}_{\star})\big)\Big)\bigg)^{2}\bigg]^{
\frac{1}{2}}
\nonumber
\\
&\leq \bar{C}\sum _{\ell =1}^{L}\bar\gamma _{N_{\ell}}.
\end{align}

\noindent\textbf{\emph{Step~2.4. Conclusion.}}\
Coming back to
\eqref{eq:chiML(h)-chi*(hL)}, we obtain by
\eqref{L2:estimate:CVaR:level:0},
\eqref{L2:estimate:1/n:sum:increment:martingale} and
\eqref{L2:estimate:sum:remainder:cvar} that
\begin{align*}
&\mathbb{E}[(\chi ^{{\mathrm{ML}}}_{\mathbf{N}}-\chi ^{h_{L}}_{\star})^{2}]
\\
&\leq 3\bigg(\mathbb{E}[(\chi ^{h_{0}}_{N_{0}}-\chi ^{h_{0}}_{\star})^{2}]
+\mathbb{E}\bigg[\bigg(\sum _{\ell =1}^{L}\frac{1}{N_{\ell}}\sum _{k=1}^{N_{
\ell}}(\varepsilon _{k}^{h_{\ell}}-\varepsilon _{k}^{h_{\ell -1}})
\bigg)^{2}\bigg]
\\
&
\hphantom{=3\bigg(}
+\mathbb{E}\bigg[\Big( \sum _{\ell =1}^{L} \frac{1}{N_{\ell}}\sum _{k=1}^{N_{
\ell}} \Big( V_{h_{\ell}}(\xi _{k-1}^{h_{\ell}})-V_{h_{\ell}}(\xi ^{h_{
\ell}}_{\star})
\\
&
\hphantom{=3\bigg(+\mathbb{E}\bigg[\Big(\sum _{\ell =1}^{L}
\frac{1}{N_{\ell}}\sum _{k=1}^{N_{\ell}}\Big(}
-\big(V_{h_{\ell -1}}(\xi _{k-1}^{h_{\ell -1}})-V_{h_{\ell -1}}(\xi ^{h_{
\ell -1}}_{\star})\big) \Big)^{2}\bigg]\bigg)
\\
&\leq \bar{C}\bigg(\sum _{\ell =1}^{L}\Big(\frac{h_{\ell}}{N_{\ell}}+
\frac{1}{N_{\ell}^{2}}\sum _{k=1}^{N_{\ell}}\gamma _{k}^{\frac{3}{2}}+
\gamma _{k}\epsilon (h_{\ell})\Big)+\Big(\sum _{\ell =1}^{L}
\bar{\gamma}_{N_{\ell}}\Big)^{2}\bigg).
\end{align*}
This concludes the proof of \eqref{L2:norm:ML:CVaR}.\qed


\begin{thebibliography}{99}

\bibitem{ASR88}
Abramowitz, C., Stegun, I., Romer, R.:
Handbook of mathematical functions with formulas, graphs, and mathematical tables.
American Journal of Physics 56 (1988)

\bibitem{ACERBI20021487}
Acerbi, C., Tasche, D.: On the coherence of expected shortfall. Journal
of Banking \& Finance 26, 1487--1503 (2002)

\bibitem{AlbaneseArmentiCrepey+2020+25+53}
Albanese, C., Armenti, Y., Cr\'epey, S.: XVA metrics for CCP optimization.
Statistics \& Risk Modeling 37, 25--53 (2020)

\bibitem{Alz97}
Alzer, H.: On some inequalities for the gamma and psi functions.
Mathematics of Computation 66, 373--389 (1997)

\bibitem{BardouFrikhaPages+2009+173+210}
Bardou, O., Frikha, N., Pag\`es, G.: Computing VaR and CVaR using stochastic
approximation and adaptive unconstrained importance sampling. Monte Carlo
Methods and Applications 15, 173--210 (2009)

\bibitem{10.1007/978-3-642-04107-5_11}
Bardou, O., Frikha, N., Pag\`es, G.: Recursive computation of value-at-risk
and conditional value-at-risk using MC and QMC. In: L'Ecuyer, P., Owen,
A. (eds.) Monte Carlo and Quasi-Monte Carlo Methods 2008, pp.~193--208.
Springer, Berlin, Heidelberg (2009)

\bibitem{bardou2016cvar}
Bardou, O., Frikha, N., Pag\`es, G.: CVaR hedging using quantization-based
stochastic approximation algorithm. Mathematical Finance 26, 184--229 (2016)

\bibitem{barrera:hal-01710394}
Barrera, D., Cr\'epey, S., Diallo, B., Fort, G., Gobet, E., Stazhynski,
U.: Stochastic approximation schemes for economic capital and risk margin
computations. ESAIM: Proceedings and Surveys 65, 182--218 (2019)

\bibitem{BISFRTB}
Basel Committee on Banking Supervision: Consultative document: Fundamental
review of the trading book: A revised market risk framework. (2013). Available
online at
bis.org/publ/bcbs265.pdf

\bibitem{BourgeyDeMarcoGobetZhou+2020+131+161}
Bourgey, F., De Marco, S., Gobet, E., Zhou, A.: Multilevel Monte Carlo
methods and lower--upper bounds in initial margin computations. Monte Carlo
Methods and Applications 26, 131--161 (2020)

\bibitem{BroadieDuMoallemi15}
Broadie, M., Du, Y., Moallemi, C.: Risk estimation via regression. Operations
Research 63, 1077--1097 (2015)

\bibitem{10.1214/21-EJS1908}
Costa, M., Gadat, S.: Non asymptotic controls on a recursive superquantile
approximation. Electronic Journal of Statistics 15, 4718--4769 (2021)

\bibitem{amlsa}
Cr\'epey, S., Frikha, N., Louzi, A., Pag\`es, G.: Asymptotic error analysis
of multilevel stochastic approximations for the value-at-risk and expected
shortfall. Electronic Journal of Probability 29, 1--56 (2024)

\bibitem{Dereich2019}
Dereich, S., M\"uller-Gronbach, T.: General multilevel adaptations for
stochastic approximation algorithms of Robbins--Monro and Polyak--Ruppert
type. Numerische Mathematik 142, 279--328 (2019)

\bibitem{doi:10.1002/9780470061602.eqf15003}
F\"ollmer, H., Schied, A.: Convex risk measures. In: Cont, R. (ed.) Encyclopedia
of Quantitative Finance, pp.~355--363. John Wiley \& Sons (2010)

\bibitem{10.1214/15-AAP1109}
Frikha, N.: Multi-level stochastic approximation algorithms. Annals of
Applied Probability 26, 933--985 (2016)

\bibitem{doi:10.1137/120903142}
Frikha, N.: Shortfall risk minimization in discrete time financial market
models. SIAM Journal on Financial Mathematics 5, 384--414 (2014)

\bibitem{FRIKHA20154066}
Frikha, N., Huang, L.: A multi-step Richardson--Romberg extrapolation method
for stochastic approximation. Stochastic Processes and their Applications
125, 4066--4101 (2015)

\bibitem{10.1214/ECP.v17-1952}
Frikha, N., Menozzi, S.: Concentration bounds for stochastic approximations.
Electronic Communications in Probability 17, 1--15 (2012)

\bibitem{10.1287/opre.1070.0496}
Giles, M.: Multilevel Monte Carlo path simulation. Operations Research
56, 607--617 (2008)

\bibitem{doi:10.1137/18M1173186}
Giles, M., Haji-Ali, A.: Multilevel nested simulation for efficient risk
estimation. SIAM/ASA Journal on Uncertainty Quantification 7, 497--525
(2019)

\bibitem{giles2024efficient}
Giles, M., Haji-Ali, A., Spence, J.: Efficient risk estimation for the
credit valuation adjustment. Preprint (2024). Available online at
arXiv:2301.05886

\bibitem{10.1007/978-3-642-41095-6_16}
Giles, M., Szpruch, L.: Antithetic multilevel Monte Carlo estimation for
multidimensional SDEs. In: Dick, J. et al.~(eds.) Monte Carlo and Quasi-Monte
Carlo Methods 2012, pp. 367--384. Springer, Berlin, Heidelberg (2013)

\bibitem{Giorgi2020}
Giorgi, D., Lemaire, V., Pag\`es, G.: Weak error for nested multilevel
Monte Carlo. Methodology and Computing in Applied Probability 22, 1325--1348
(2020)

\bibitem{Gordy}
Gordy, M., Juneja, S.: Nested simulation in portfolio risk measurement.
Management Science 56, 1833--1848 (2010)

\bibitem{doi:10.1137/21M1447064}
Haji-Ali, A., Spence, J., Teckentrup, A.: Adaptive multilevel Monte Carlo
for probabilities. SIAM Journal on Numerical Analysis 60, 2125--2149 (2022)

\bibitem{10.21314/JOR.2000.038}
Rockafellar, R., Uryasev, S.: Optimization of conditional value-at-risk.
Journal of Risk 2(3), 21--41 (2000)

\bibitem{XU2024115745}
Xu, Z., He, Z., Wang, X.: Efficient risk estimation via nested multilevel
quasi-Monte Carlo simulation. Journal of Computational and Applied Mathematics
443, 115745 (2024)

\end{thebibliography}
\end{document}